\documentclass[journal, twocolumn, 10pt]{IEEEtran}
\usepackage{graphicx,amssymb,amstext,amsmath,cite}
\usepackage{tikz}
\usepackage{amsthm}
\usepackage{kbordermatrix}

\usepackage{color}
\usepackage{soul}
\usepackage{dsfont}
\usepackage{multirow}
\usepackage{ctable} 

\usepackage{graphicx}
\usepackage{caption}
\usepackage{subcaption}

\newcounter{actr}
{\begin{list}{(\alph{actr})}{\usecounter{actr}}}{\end{list}}

\newcounter{ictr}
{\begin{list}{(\roman{ictr})}{\usecounter{ictr}}}{\end{list}}

\newtheorem{remark}{Remark}
\newtheorem{thm}{Theorem}
\newtheorem{lemma}{Lemma}
\newtheorem{claim}{Claim}
\newtheorem{corol}{Corollary}
\newtheorem{prop}{Proposition}

\newenvironment{new-proof}[1]
{{\em Proof }:\\}%
{ \noindent\qed }
\newcommand{\defeq}{\stackrel{\Delta}{=}}

\newcommand{\bc}{{\mathbf{c}}}

\newcommand{\cC}{{\mathcal{C}}}

\newcommand{\bd}{{\mathbf{d}}}

\newcommand{\cF}{{\mathcal{F}}}

\newcommand{\cG}{{\mathcal{G}}}

\newcommand{\cL}{{\mathcal{L}}}

\newcommand{\cM}{{\mathcal{M}}}

\newcommand{\cN}{{\mathcal{N}}}

\newcommand{\bs}{{\mathbf{s}}}

\newcommand{\cS}{{\mathcal{S}}}

\newcommand{\cW}{{\mathcal{W}}}

\newcommand{\al}{\alpha}

\newcommand{\bt}{\boldsymbol{t}}

\newcommand{\g}{\gamma}

\newcommand{\eps}{\varepsilon}

\DeclareMathAlphabet{\mathbsf}{OT1}{cmss}{bx}{n}
\DeclareMathAlphabet{\mathssf}{OT1}{cmss}{m}{sl}

\DeclareSymbolFont{bsfletters}{OT1}{cmss}{bx}{n}
\DeclareSymbolFont{ssfletters}{OT1}{cmss}{m}{n}
\DeclareMathSymbol{\bsfGamma}{0}{bsfletters}{'000}
\DeclareMathSymbol{\ssfGamma}{0}{ssfletters}{'000}
\DeclareMathSymbol{\bsfDelta}{0}{bsfletters}{'001}
\DeclareMathSymbol{\ssfDelta}{0}{ssfletters}{'001}
\DeclareMathSymbol{\bsfTheta}{0}{bsfletters}{'002}
\DeclareMathSymbol{\ssfTheta}{0}{ssfletters}{'002}
\DeclareMathSymbol{\bsfLambda}{0}{bsfletters}{'003}
\DeclareMathSymbol{\ssfLambda}{0}{ssfletters}{'003}
\DeclareMathSymbol{\bsfXi}{0}{bsfletters}{'004}
\DeclareMathSymbol{\ssfXi}{0}{ssfletters}{'004}
\DeclareMathSymbol{\bsfPi}{0}{bsfletters}{'005}
\DeclareMathSymbol{\ssfPi}{0}{ssfletters}{'005}
\DeclareMathSymbol{\bsfSigma}{0}{bsfletters}{'006}
\DeclareMathSymbol{\ssfSigma}{0}{ssfletters}{'006}
\DeclareMathSymbol{\bsfUpsilon}{0}{bsfletters}{'007}
\DeclareMathSymbol{\ssfUpsilon}{0}{ssfletters}{'007}
\DeclareMathSymbol{\bsfPhi}{0}{bsfletters}{'010}
\DeclareMathSymbol{\ssfPhi}{0}{ssfletters}{'010}
\DeclareMathSymbol{\bsfPsi}{0}{bsfletters}{'011}
\DeclareMathSymbol{\ssfPsi}{0}{ssfletters}{'011}
\DeclareMathSymbol{\bsfOmega}{0}{bsfletters}{'012}
\DeclareMathSymbol{\ssfOmega}{0}{ssfletters}{'012}

\renewcommand{\defeq}{\triangleq}

\newcommand{\rvN}{{\mathssf{N}}}    

\newcommand{\rvX}{{\mathssf{X}}}    

\newcommand{\rvY}{{\mathssf{Y}}}    
\newcommand{\rvZ}{{\mathssf{Z}}}    

\newcommand{\rva}{{\mathssf{a}}}    

\newcommand{\rve}{{\mathssf{e}}}    

\newcommand{\rvf}{{\mathssf{f}}}    

\newcommand{\rvg}{{\mathssf{g}}}    

\newcommand{\rvm}{{\mathssf{m}}}    

\newcommand{\rvn}{{\mathssf{n}}}    

\newcommand{\rvs}{{\mathssf{s}}}    
\newcommand{\rvbs}{{\mathbsf{s}}}

\newcommand{\rvbt}{{\mathbsf{t}}}

\newcommand{\rvu}{{\mathssf{u}}}    

\newcommand{\rvbu}{{\mathbsf{u}}}

\newcommand{\rvv}{{\mathssf{v}}}    

\newcommand{\rvw}{{\mathssf{w}}}    

\newcommand{\rvz}{{\mathssf{z}}}    

\newcommand{\rvU}{{\mathssf{U}}}

\author{{Farrokh~Etezadi,  Ashish~Khisti and Mitchell~Trott}
\thanks{Manuscript submitted June 2013, revised December 2013.}
\thanks{Farrokh Etezadi (fetezadi@comm.utoronto.ca) and Ashish Khisti ({akhisti@comm.utoronto.ca}) are with the University of Toronto, Toronto, ON, Canada, Mitchell Trott was with HP Labs, Palo Alto, USA. 
This work was supported by an NSERC Discovery Research Grant, a Hewlett-Packard Innovation Research Program award and an Ontario Early Research Award. This work was presented in parts at the 2012 IEEE Data Compression Conference and the 2012 Allerton Conference on Communication, Control and Computing. }}

\title{Zero-Delay Sequential Transmission of Markov Sources over Burst Erasure Channels}

\begin{document}
 
\maketitle

\begin{abstract}
A setup involving zero-delay sequential transmission of a vector Markov source over a burst erasure channel  is studied.  A sequence of source vectors is compressed in a causal fashion at the encoder,  and the resulting output is transmitted over a burst erasure channel.  The destination is required to reconstruct each source vector with zero-delay, but those source sequences that are observed either during the burst erasure, or in the interval of length $W$ following the burst erasure need not be reconstructed.   The minimum achievable compression rate is called the rate-recovery function.   We assume that each source vector is sampled  i.i.d.\ across the spatial dimension and from a stationary, first-order Markov process across the temporal dimension. 

For discrete sources the case of lossless recovery is considered, and upper and lower bounds on the rate-recovery function are established. Both these bounds can be expressed as the rate for predictive coding, plus a term that decreases at least inversely with the recovery window length $W$. For Gauss-Markov sources and a quadratic distortion measure, upper and lower bounds on the minimum rate are established when $W=0$. These bounds are shown to coincide in the high resolution limit.  Finally another setup involving  i.i.d. Gaussian sources is studied and the rate-recovery function is completely characterized in this case.

\end{abstract}



\section{Introduction}

\IEEEPARstart{R}eal-time streaming applications require both the sequential compression, and playback of multimedia frames under strict latency constraints.  Linear predictive techniques such as DPCM have long been used to exploit the source memory in such systems. However predictive coding  schemes also exhibit a significant level of error propagation in the presence of packet losses~\cite{wang2000error}.   In practice one must develop transmission schemes that satisfy both the real-time constraints and are robust to channel errors.  

There exists an inherent  tradeoff between the underlying {\em transmission-rate} and the {\em error-propagation} at the receiver in all video streaming applications.  Commonly used video compression formats such as H.264/MPEG and HEVC use a combination of intra-coded and predictively-coded frames to limit the amount of error propagation. The predictively-coded frames are used to improve the compression efficiency whereas the intra-coded frames limit the amount of error propagation. 
Other techniques  including forward error correction codes~\cite{tan}, leaky DPCM~\cite{Huang:08} and distributed video coding~\cite{pradhanRamchandran:03} can also be used to trade off the transmission rate with error propagation. Despite this, such a tradeoff is not well understood even in the case of a single isolated packet loss~\cite{Wang:06}.


In this paper we study the information theoretic tradeoff between the transmission rate and error propagation in a simple source-channel model. 
We assume that the channel introduces an  isolated erasure burst of a certain maximum length, say $B$. The encoder observes a sequence of vector sources and compresses them in a causal fashion. The decoder is required to reconstruct each source vector with zero delay, except those that occur during the error propagation window. The decoder can declare a \emph{don't-care} for all the source sequences that occur in this window. We assume that is period spans the duration of the erasure burst, as well an interval of length $W$ immediately following it.  We study the minimum rate required $R(B,W)$, and define it as the \emph{rate-recovery function}.  

We first consider the case of discrete sources and lossless reconstruction and establish  upper and lower bounds on the minimum rate. Both these bounds can be expressed as the rate of the predictive coding scheme, plus an additional term that decreases at-least as $H(\rvs)/(W+1)$ where $H(\rvs)$ denotes the entropy of the source symbol.  Our lower bound is obtained through connection to a certain multi-terminal source coding problem that captures the tension in encoding a source sequence during the error-propagation period, and outside it. The upper bound is based on a natural random-binning scheme. We also consider the  case of Gauss-Markov sources and a quadratic distortion measure. We again establish upper and lower bounds on the minimum rate when $W=0$, i.e., when instantaneous recovery following the burst erasure is imposed. We observe that our upper and lower bounds coincide in the high resolution limit, thus establishing the rate-recovery function in this regime. Finally we consider a different setup involving i.i.d.\ Gaussian sources, and a special recovery constraint, and obtain an exact characterization of the  rate-recovery function in this special case.  Many of our results also naturally extend to the case when the channel introduces multiple erasure bursts.

The remainder of the paper is organized as follows.  We discuss related literature  in Section~\ref{sec:Background}. The problem setup is described in Section~\ref{sec:statement} and a summary of the main results is provided in Section~\ref{sec:Results}. We treat the case of discrete sources and lossless recovery in Section~\ref{sec:THM1} and establish upper and lower bounds on the minimum rate. The optimality of binning for the special case of symmetric sources and memoryless encoders is established in Section~\ref{sec:Symmetric}. In Section~\ref{sec:GM} we consider the case of Gauss-Markov source with a quadratic distortion constraint. Section~\ref{sec:Gauss} studies another setup involving independent Gaussian sources and a sliding window recovery constraint, where an exact characterization of the minimum rate is obtained. Conclusions appear in Section~\ref{sec:Conclusion}. 

{\bf Notations}: Throughout this paper we represent the  Euclidean norm operator by $||\cdot||$ and the expectation operator by $E[\cdot]$. 
The notation ``$\log$'' is used for the binary logarithm, and rates are expressed in bits.
 The operations $H(.)$ and $h(.)$ denote the entropy and the differential entropy, respectively.  
The ``slanted sans serif'' font $\rva$ and the normal font $a$ represent random variables and their realizations respectively. 
The notation $\rva_{i}^{n}=\{\rva_{i,1},\ldots,\rva_{i,n}\}$ represents a length-$n$ sequence of symbols at time $i$. 
The notation $[\rvf]_{i}^{j}$ for $i< j$ represents $\rvf_{i},\rvf_{i+1}, \ldots, \rvf_{j}$.

%
%
%
%
%
%

\section{ Related Works}
\label{sec:Background}
Problems involving real-time coding and compression have been studied from many different perspectives in related literature.
The compression of a Markov source, with zero encoding and decoding delays, was studied in an early work by Witsenhausen \cite{Witsenhausen:79}. In this setup, the encoder must  sequentially compress a (scalar) Markov source and transmit it
over an ideal channel. The decoder must reconstruct the source symbols with zero-delay and under an average distortion constraint.
It was shown in \cite{Witsenhausen:79} that for a $k$-th order Markov source model, an encoding rule that only depends on the $k$ most recent source symbols, and the decoder's memory, is sufficient to achieve the optimal rate.  Similar structural results have been  obtained in a number of followup works, see e.g.,~\cite{Teneketzis} and references therein. The authors in \cite{Asnani} considered real-time communication of a  memoryless source over memoryless channels, with or without the presence of unit-delay feedback. The encoding and decoding is sequential with a fixed finite lookahead at the encoder. The authors  propose conditions under which  symbol-by-symbol encoding and decoding, without lookahead, is optimal and more generally characterize the optimal encoder as a solution to a dynamic programming problem.


In another line of work, the problem of sequential coding of  correlated vector sources in a multi-terminal source coding framework was introduced by Viswanathan and Berger~\cite{berger}.  In this setup, a set of correlated sources must be sequentially compressed by the encoder, whereas the decoder at each stage is required to reconstruct the corresponding source sequence, given all the encoder outputs up to that time. It is noted in~\cite{berger} that the correlated source sequences can model consecutive video frames and each stage at the decoder maps to sequential reconstruction of a particular source frame. This setup is an extension of the  well-known  successive refinement problem in source coding~\cite{equitzCover:91}. In followup works, in reference~\cite{Wu} the authors consider the case where the encoders at each time have access to previous encoder outputs rather than previous source frames.
Reference~\cite{ishwar} considers an extension where the encoders and decoders can introduce non-zero delays. All these works assume ideal channel conditions.
Reference~\cite{songChen12} considers an extension of~\cite{berger}  where at any given stage the decoder has either all the previous outputs, or only the present output. A robust extension of the predictive coding scheme is proposed and shown to achieve the minimum sum-rate. However this setup does not capture the effect of packet losses over a channel, where the destination has access to all the non erased symbols.  To our knowledge, only reference~\cite{Huang:08} considers the setting of  sequential coding over a random packet erasure channel. The source is assumed to be Gaussian, spatially i.i.d.\, and temporally autoregressive. A class of  linear predictive coding schemes is studied and an optimal scheme within this class, with respect to the excess distortion ratio metric is proposed. Our proposed coding scheme is qualitatively different from~\cite{Huang:08, songChen12} and involves a random binning based approach,  which is inherently robust to the side-information at the decoder.

In other related works, the joint source-channel coding of a vector Gaussian source over a vector Gaussian channel with zero reconstruction delay has also been extensively studied. While optimal analog mappings are not known in general, a number of interesting approaches have been proposed in e.g.~\cite{Chung:00,tuncel} and related references.  Reference~\cite{arildsen} studies the problem of sequential coding of the scalar Gaussian source over a channel with random erasures. In~\cite{Wang:06}, the authors  consider a joint source-channel coding setup and propose the use of distributed source coding  to compensate the effect of channel losses. However no optimality results are presented for the proposed scheme. Sequential random binning techniques for streaming scenarios have been proposed in e.g.~\cite{Chang:07}, \cite{draper} and the references therein.

To the best of our knowledge, there has been no prior work that studies an information theoretic tradeoff between error-propagation and compression efficiency in real-time streaming systems.

%
%
%
%
%
%
%

\begin{figure*}
\begin{center}
\vspace{1em}
\begin{tikzpicture}

\fill[color=red!40!white] (-1.3,-.3) rectangle (-.7,.3);
\draw [black](-1.3,-.3) rectangle (-.7,.3);
\draw [color=red!40!white] (-1,0) -- (-1,0) node {$\color{black}\rvs^n_{-1}$};
\draw [->](-0.7,0) -- (-0.3,0);

\draw [white] (0,0) -- (0,0) node {$\color{black}\rvs^n_0$};
\draw [->](0,-.3) -- (0,-.7);
\draw [white] (0,-1) -- (0,-1) node {$\color{black}\rvf_0$};
\draw [->](0,-1.3) -- (0,-1.7);
\draw [white] (0,-2) -- (0,-2) node {$\color{black}\rvf_0$};
\draw [->](0,-2.3) -- (0,-2.7);
\draw [white] (0,-3) -- (0,-3) node {$\color{black}\hat{\rvs}^n_0$};
\draw [->](0.3,0) -- (0.7,0);

\draw [white] (1,0) -- (1,0) node {$\color{black}\rvs^n_1$};
\draw [->](1,-.3) -- (1,-.7);
\draw [white] (1,-1) -- (1,-1) node {$\color{black}\rvf_1$};
\draw [->](1,-1.3) -- (1,-1.7);
\draw [white] (1,-2) -- (1,-2) node {$\color{black}\rvf_1$};
\draw [->](1,-2.3) -- (1,-2.7);
\draw [white] (1,-3) -- (1,-3) node {$\color{black}\hat{\rvs}^n_1$};
\draw [->](1.3,0) -- (1.7,0);

\draw [white] (2,0) -- (2,0) node {$\color{black}\rvs^n_{2}$};
\draw [->](2,-.3) -- (2,-.7);
\draw [white] (2,-1) -- (2,-1) node {$\color{black}\rvf_2$};
\draw [->](2,-1.3) -- (2,-1.7);
\draw [white] (2,-2) -- (2,-2) node {$\color{black}\rvf_2$};
\draw [->](2,-2.3) -- (2,-2.7);
\draw [white] (2,-3) -- (2,-3) node {$\color{black}\hat{\rvs}^n_2$};
\draw [->](2.3,0) -- (2.7,0);

\draw [dotted](2.8,0) -- (3.1,0);

\draw [white] (3.5,0) -- (3.5,0) node {$\color{black}\rvs^n_{j-1}$};
\draw [->](3.5,-.3) -- (3.5,-.7);
\draw [white] (3.5,-1) -- (3.5,-1) node {$\color{black}\rvf_{j-1}$};
\draw [->](3.5,-1.3) -- (3.5,-1.7);
\draw [white] (3.5,-2) -- (3.5,-2) node {$\color{black}\rvf_{j-1}$};
\draw [->](3.5,-2.3) -- (3.5,-2.7);
\draw [white] (3.5,-3) -- (3.5,-3) node {$\color{black}\hat{\rvs}^n_{j-1}$};
\draw [->](3.9,0) -- (4.2,0);

\draw [white] (4.5,0) -- (4.5,0) node {$\color{black}\rvs^n_{j}$};
\draw [->](4.5,-.3) -- (4.5,-.7);
\draw [white] (4.5,-1) -- (4.5,-1) node {$\color{black}\rvf_{j}$};
\draw [->](4.5,-1.3) -- (4.5,-1.7);
\draw [white] (4.5,-2) -- (4.5,-2) node {$\color{black}\star$};
\draw [->](4.5,-2.3) -- (4.5,-2.7);
\draw [white] (4.5,-3) -- (4.5,-3) node {$\color{black}-$};
\draw [->](4.8,0) -- (5.1,0);

\draw [white] (5.5,0) -- (5.5,0) node {$\color{black}\rvs^n_{j+1}$};
\draw [->](5.5,-.3) -- (5.5,-.7);
\draw [white] (5.5,-1) -- (5.5,-1) node {$\color{black}\rvf_{j+1}$};
\draw [->](5.5,-1.3) -- (5.5,-1.7);
\draw [white] (5.5,-2) -- (5.5,-2) node {$\color{black}\star$};
\draw [->](5.5,-2.3) -- (5.5,-2.7);
\draw [white] (5.5,-3) -- (5.5,-3) node {$\color{black}-$};

\draw [dotted](5.9,0) -- (6.2,0);

\draw [white] (6.8,0) -- (6.8,0) node {$\color{black}\rvs^n_{j+B-1}$};
\draw [->](6.8,-.3) -- (6.8,-.7);
\draw [white] (6.8,-1) -- (6.8,-1) node {$\color{black}\rvf_{j+B-1}$};
\draw [->](6.8,-1.3) -- (6.8,-1.7);
\draw [white] (6.8,-2) -- (6.8,-2) node {$\color{black}\star$};
\draw [->](6.8,-2.3) -- (6.8,-2.7);
\draw [white] (6.8,-3) -- (6.8,-3) node {$\color{black}-$};
\draw [->](7.4,0) -- (7.7,0);

\draw [white] (8.2,0) -- (8.2,0) node {$\color{black}\rvs^n_{j+B}$};
\draw [->](8.2,-.3) -- (8.2,-.7);
\draw [white] (8.2,-1) -- (8.2,-1) node {$\color{black}\rvf_{j+B}$};
\draw [->](8.2,-1.3) -- (8.2,-1.7);
\draw [white] (8.2,-2) -- (8.2,-2) node {$\color{black}\rvf_{j+B}$};
\draw [->](8.2,-2.3) -- (8.2,-2.7);
\draw [white] (8.2,-3) -- (8.2,-3) node {$\color{black}-$};
\draw [->](8.7,0) -- (9,0);

\draw [dotted](9.1,0) -- (9.5,0);

\draw [->](9.6,0) -- (10,0);
\draw [white] (10.9,0) -- (10.9,0) node {$\color{black}\rvs^n_{j+B+W-1}$};
\draw [->](10.9,-.3) -- (10.9,-.7);
\draw [white] (10.9,-1) -- (10.9,-1) node {$\color{black}\rvf_{j+B+W-1}$};
\draw [->](10.9,-1.3) -- (10.9,-1.7);
\draw [white] (10.9,-2) -- (10.9,-2) node {$\color{black}\rvf_{j+B+W-1}$};
\draw [->](10.9,-2.3) -- (10.9,-2.7);
\draw [white] (10.9,-3) -- (10.9,-3) node {$\color{black}-$};

\draw [->](11.8,0) -- (12.1,0);
\draw [white] (12.9,0) -- (12.9,0) node {$\color{black}\rvs^n_{j+B+W}$};
\draw [->](12.9,-.3) -- (12.9,-.7);
\draw [white] (12.9,-1) -- (12.9,-1) node {$\color{black}\rvf_{j+B+W}$};
\draw [->](12.9,-1.3) -- (12.9,-1.7);
\draw [white] (12.9,-2) -- (12.9,-2) node {$\color{black}\rvf_{j+B+W}$};
\draw [->](12.9,-2.3) -- (12.9,-2.7);
\draw [white] (12.9,-3) -- (12.9,-3) node {$\color{black}\hat{\rvs}^n_{j+B+W}$};

\draw [->](13.7,0) -- (13.9,0);
\draw [white] (14.8,0) -- (14.8,0) node {$\color{black}\rvs^n_{j+B+W+1}$};
\draw [->](14.8,-.3) -- (14.8,-.7);
\draw [white] (14.8,-1) -- (14.8,-1) node {$\color{black}\rvf_{j+B+W+1}$};
\draw [->](14.8, -1.3) -- (14.8,-1.7);
\draw [white] (14.8,-2) -- (14.8,-2) node {$\color{black}\rvf_{j+B+W+1}$};
\draw [->](14.8,-2.3) -- (14.8,-2.7);
\draw [white] (14.8,-3) -- (14.8,-3) node {$\color{black}\hat{\rvs}^n_{j+B+W+1}$};

\draw [white] (5.8,-3.6) -- (5.8,-3.6) node {$\color{black}\textrm{\footnotesize{Erased}}$};
\draw [white] (9.75,-3.6) -- (9.75,-3.6) node {$\color{black}\textrm{\footnotesize{Not to be recovered}}$};


\draw [dashed](4.2,-3.4) rectangle (7.5,-1.6);
\draw [dashed](7.6,-3.4) rectangle (11.9,-1.6);
\draw (4.1,-3.9) rectangle (12,-1.5);
\draw [white] (8,-4.2) -- (8,-4.2) node {$\color{black}\textrm{\small{Error Propagation Window}}$};
\end{tikzpicture}

\caption{ Problem Setup: The encoder output $\rvf_i$ is a  function of  all the past source sequences.
The channel introduces a burst erasure of length up to $B$.
The decoder produces $\hat{\rvs}_i^n $ upon observing the  channel outputs up to time $i$.  
As indicated, the decoder is not required to produce those source sequences that are observed either during the burst erasure, or a period of $W$ 
following it. The first sequence, $\rvs_{-1}^n$ is a synchronization frame available to both the source and destination.}
\label{fig:setup}
\end{center}
\end{figure*}
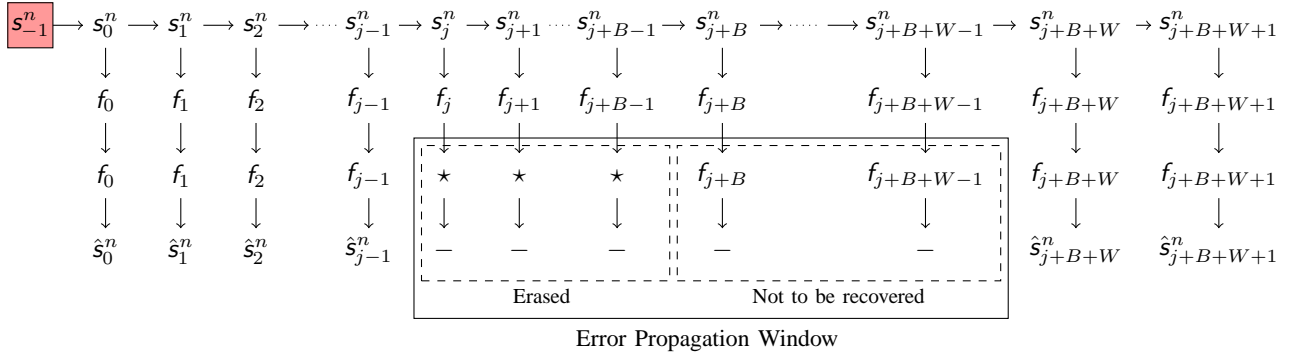

\section{Problem Statement}
\label{sec:statement}
In this section we introduce our source and channel models and the associated definition of the rate-recovery function. 

We assume that the communication spans the interval $i\in [-1, \mathcal{L}]$. At each time $i$, a source vector $\{\rvs_i^n\}$ is sampled,  whose symbols are drawn independently across the spatial dimension, and from a first-order Markov chain across the temporal dimension, i.e., 
\begin{multline}
\Pr(~\rvs_i^n = s_i^n~|~\rvs_{i-1}^n = s_{i-1}^n,~\rvs_{i-2}^n = s_{i-2}^n,\ldots, \rvs^n_{-1}=s^n_{-1})  \\ = \prod_{k=1}^n p_{1}(s_{i,k}|s_{i-1,k}),
\qquad  0\le i \le \mathcal{L}.\label{eq:Markov}
\end{multline}
The underlying random variables $\{\rvs_i\}$ constitute a time-invariant, stationary and a first-order Markov chain with a common marginal distribution denoted by $p_\rvs(\cdot)$ over an alphabet $\cS$. The sequence $\rvs_{-1}^n$ is sampled i.i.d.\ from $p_\rvs(\cdot)$ and revealed to both the encoder and decoder before the start of the communication.  It plays the role of  a synchronization frame. 

A rate-$R$  encoder computes  an index $\rvf_i \in [1,2^{nR}]$ at  time $i$, according to an encoding function 
\begin{align}
\rvf_i=\mathcal{F}_i\left({\rvs}^n_{-1},{\rvs}^n_0,..., {\rvs}^n_{i}\right), \qquad  0\le i \le \mathcal{L}. \label{eq:f-enc}
\end{align} 
Note that the encoder in~\eqref{eq:f-enc} is a causal function of the source sequences. 
A {\em memoryless} encoder satisfies $\mathcal{F}_i(\cdot) = \cF_i(\rvs_i^n)$ i.e., the encoder does  not use the knowledge of the past sequences.  Naturally a memoryless encoder is very restrictive, and we will only use it to establish some special results. 

The channel takes each $\rvf_i$ as input and either outputs $\rvg_i = \rvf_i$ or an erasure symbol i.e., $\rvg_i=\star$. We consider the class of burst erasure channels. For some particular $j \ge 0$, it introduces a burst erasure such that  $\rvg_{i}=\star$ for $i\in\{j,j+1,...,j+B'-1\}$ and  $\rvg_{i}=\rvf_{i}$ otherwise i.e.,\begin{align}
\rvg_i =\begin{cases}
\star, &i \in [j,j+1,\ldots, j+B'-1]\\
\rvf_i, &\text{ else},
\end{cases}\label{eq:chModel}
\end{align} where the burst length $B'$ is upper bounded by $B$.

Upon observing the sequence $\{\rvg_i\}_{i\ge 0}$, the decoder is required to reconstruct each source sequence with zero delay i.e., 
\begin{align}
\hat{\rvs}^{n}_i=\mathcal{G}_{i}(\rvg_0, \rvg_1, \ldots, \rvg_{i}, \rvs_{-1}^n), \quad i \notin \{j, \ldots, j+B'+W-1\} \label{eq:decoder-def}
\end{align}
where $\hat{\rvs}_i^{n}$ denotes the reconstruction sequence and $j$ denotes the time at which burst erasure starts in~\eqref{eq:chModel}. The destination is not required to produce the source vectors that appear either during the burst erasure or in the period of length $W$ following it. We call this period the error propagation window. Fig.~\ref{fig:setup} provides a schematic of the causal encoder~\eqref{eq:f-enc}, the channel model~\eqref{eq:chModel}, and the decoder~\eqref{eq:decoder-def}. 

\subsection{Rate-Recovery Function}  
We define the rate-recovery function under lossless and lossy reconstruction constraints. 
\subsubsection{Lossless Rate-Recovery Function}
We first consider the case when the reconstruction in~\eqref{eq:decoder-def} is required to be lossless. 
We assume that the source alphabet is discrete and the entropy $H(\rvs)$ is finite. 
A rate $R_{\mathcal{L}}(B, W)$ is feasible if there exists a sequence of encoding and decoding functions and a sequence $\epsilon_{n}$ that approaches zero as $n \to \infty$ such that, $\Pr(\rvs_i^n \neq \hat{\rvs}_i^{n})\leq\epsilon_{n}$ for all source sequences reconstructed as in~\eqref{eq:decoder-def}. We seek the minimum feasible rate $R_{\mathcal{L}}(B,W)$, which is the {\em lossless rate-recovery} function. 
In this paper, we will focus on infinite-horizon case, $R(B,W) =\lim_{\mathcal{L}\to\infty}R_{\mathcal{L}}(B,W)$, which will be  called the rate-recovery function for simplicity.

\subsubsection{Lossy Rate-Recovery Function}
We also consider the case where reconstruction in~\eqref{eq:decoder-def} is required to satisfy an average distortion constraint:
\begin{align}
\limsup_{n\to \infty} E\left[ \frac{1}{n}\sum_{k=1}^n d(\rvs_{i,k},\hat{\rvs}_{i,k})\right] \le D \label{eq:D-def}
\end{align} for some distortion measure $d:\mathbb{R}^2 \rightarrow [0,\infty)$.
The rate $R$ is feasible if a sequence of encoding and decoding functions exists that satisfies the average distortion constraint. The minimum feasible rate $R_{\mathcal{L}}(B, W, D)$, is the {\em lossy rate-recovery function}.  
The study of lossy rate-recovery function for the general case appears to be quite challenging. 
In this paper we will focus on the class of Gaussian-Markov sources, with quadratic distortion measure, i.e. $d(\rvs,\hat{\rvs}) = (\rvs-\hat{\rvs})^2,$ where the analysis simplifies.
We will again focus on infinite-horizon case, $R(B,W,D) =\lim_{\mathcal{L}\to\infty}R_{\mathcal{L}}(B,W,D)$ which we simply call the rate-recovery function.
Table~\ref{tab:1} summarizes the  notation used throughout the paper. 

\begin{table}
\caption {Summary of notation used in the paper.} \label{tab:1} 
\begin{center}
\vspace{1em}
%
%
%
%
%

\begin{tabular}{ c| c|c  }
\specialrule{0.1em}{1em}{0em}
\multirow{4}{*}{\begin{tabular}{ c}Source\\ Parameters \end{tabular}} & Source Symbol& $\rvs$ \\\cline{2-3}
 & Source Reproduction & $\hat{\rvs}$ \\ \cline{2-3}
 & \begin{tabular}{ c} Temporal Correlation Coefficient\\
 of Gauss-Markov Source Model \end{tabular}& $\rho$\\
 \specialrule{.1em}{0em}{0em}
 \multirow{4}{*}{\begin{tabular}{ c}Channel\\ Parameters \end{tabular}} & Channel Input & $\rvf$ \\\cline{2-3}
 & Channel Output & $\rvg$ \\ \cline{2-3}
 & Maximum Burst Length & $B$ \\ \cline{2-3}
 & Guard Length between Consecutive Bursts & $L$ \\ 
 \specialrule{.1em}{0em}{0em}
 \multirow{3}{*}{\begin{tabular}{ c}System\\ Parameters \end{tabular}} & Length of Source Sequences & $n$\\\cline{2-3}
 & Communication Duration & $\mathcal{L}$\\\cline{2-3}
 & Recovery Window Length & $W$ \\ 
  \specialrule{.1em}{0em}{0em}
  \multirow{2}{*}{\begin{tabular}{ c}Performance \\ Metrics \end{tabular}} & Rate  & $R$ \\\cline{2-3}
 & Average Distortion & $D$\\ \specialrule{.1em}{0em}{0em}
\end{tabular}
\end{center}
\end{table}

\begin{remark}
Note that our proposed setup only considers a single burst erasure during the entire duration of communication. When we consider lossless recovery at the destination our results immediately extend to channels involving multiple burst erasures with a certain guard interval separating consecutive bursts. When we consider Gauss-Markov sources with a quadratic distortion measure, we will explicitly treat the channel with multiple burst erasures and compare the achievable rates with that of a single burst erasure channel.
\end{remark}

\subsection{Practical Motivation}

Note that our setup assumes that the size of both the source frames and  channel packets is sufficiently large. A  relevant application for the proposed setup is video streaming. Video frames are generated at a  rate of approximately 60 Hz and each frame typically contains several hundred thousand pixels.  The inter-frame interval is thus $\Delta_s \approx 17 $ ms. Suppose that the underlying broadband communication channel has a bandwidth of $W_s = 2.5$~MHz. Then in the interval of $\Delta_s$ the number of symbols transmitted using ideal synchronous modulation is $N=2\Delta_s W_s \approx 84,000$. Thus the block length between successive frames is sufficiently long that capacity achieving codes could be used and the erasure model and large packet sizes is justified. 
The assumption of spatially i.i.d.\ frames could reasonably approximate the video \emph{innovation process} generated by applying suitable transform on original video frames. 
Such models have been also used in earlier works e.g.,~\cite{berger,songChen12,Huang:08,ishwar,Wu}.   

Possible applications of the burst loss model considered in our setup  include fading wireless channels and congestion in wired networks. We note that the present paper does not consider a statistical channel model but  instead considers a worst case channel model.  As mentioned before even the effect of such a single burst loss has not been well understood in the video streaming setup and therefore our proposed setup is a natural starting point. Furthermore while the statistical models are used to capture the typical behaviour of channel errors, the atypical behaviour is often modelled (see e.g.,~\cite[Sec.~6.10]{gallager}) using a worst-case approach. Therefore in low-latency applications where the local channel dynamics are relevant such models are often used (see e.g.,~\cite{martinianThesis, badrinfo:13,tekin,leong}). Finally we note that earlier works (see e.g.,~\cite{Huang:08}) that consider statistical channel models, also ultimately simplify the system by analyzing the effect of each burst erasure separately in steady state.

%
%
%
%
%
%
%

\section{Main Results}
\label{sec:Results}
We summarize the main results of this paper. We note in advance that throughout the paper, the upper bound on the rate-recovery function indicates the rate achievable by a proposed coding scheme and the lower bound corresponds to a necessary condition that the rate-recovery function of any feasible coding scheme has to satisfy. Section~\ref{sec:LLRR} treats the lossless rate-recovery function and presents
lower and upper bounds in Theorem~\ref{thm:genUB_LB}. Corollary~\ref{thm:binning} presents the lossless rate-recovery function for a special case of symmetric sources, when restricted to memoryless encoders. Section~\ref{sec:GMS} treats the lossy rate-recovery function for the class of Gauss-Markov sources.  Prop.~\ref{prop:GML} presents a lower bound, whereas Prop.~\ref{prop:GMAch} and Prop.~\ref{prop:GM-ME} present upper bounds on lossy rate-recovery function for the single and multiple burst erasure channel models respectively. Our bounds coincide in the high resolution limit, as stated in Corollary~\ref{corol:HR}.  Finally Section~\ref{subsec:GSW} treats another setup involving  independent Gaussian sources, with a sliding window recovery constraint, and establishes the associated rate-recovery function.

\subsection{Lossless Rate-Recovery Function}
\label{sec:LLRR}
\begin{thm}(Lossless Rate-Recovery Function)
For the stationary, first-order Markov, discrete source process, the lossless rate-recovery function satisfies the following upper and lower bounds: $R^-(B,W) \le R(B,W)\le R^+(B,W),$ where
\begin{align}
R^{+}(B, W) &\!=\! H(\rvs_{1}|\rvs_0)+\frac{1}{W+1}I(\rvs_{B};\rvs_{B+1}|\rvs_{0}), \label{eq:genUB}\\
R^-(B, W) &\!=\! H(\rvs_{1}|\rvs_0)+\frac{1}{W+1}I(\rvs_{B};\rvs_{B+W+1}|\rvs_{0}). \label{eq:genLB}
\end{align}
\label{thm:genUB_LB}\hfill$\Box$
\end{thm}
Notice that the upper and lower bounds~\eqref{eq:genUB} and~\eqref{eq:genLB} coincide for $W=0$ and $W \rightarrow\infty$, yielding the rate-recovery function in these cases. We can interpret the term $H(\rvs_1|\rvs_0)$ as the amount of uncertainty in $\rvs_i$ when the past sources are perfectly known.  This term is equivalent to the rate associated with ideal predictive coding in absence of any erasures. The second term in both~\eqref{eq:genUB} and~\eqref{eq:genLB} is the additional penalty that arises due to the recovery constraint following a burst erasure. Notice that this term decreases at-least as $H(\rvs)/(W+1)$, thus the penalty decreases as we increase the recovery period $W$. Note that the mutual information term associated with the lower bound is $I(\rvs_B; \rvs_{B+W+1}|\rvs_0)$ while that in the upper bound is $I(\rvs_B; \rvs_{B+1}|\rvs_0)$. Intuitively this difference arises because in the lower bound we only consider the reconstruction of $\rvs_{B+W+1}^n$ following an erasure bust in $[1,B]$ while, as explained below in Corollary~\ref{corol:genUB} the upper bound involves a binning based scheme that reconstructs  all sequences $(\rvs_{B+1}^n,\ldots, \rvs_{B+W+1}^n)$ at time $t=B+W+1$.  

A proof of Theorem~\ref{thm:genUB_LB} is provided in Section~\ref{sec:THM1}. The lower bound involves a connection to a multi-terminal source coding problem. This model captures the different requirements imposed on the encoder output following a burst erasure and in the steady state. The following Corollary  provides an alternate expression for the achievable rate and makes the connection to the binning technique  explicit. 

\begin{corol}
\label{corol:genUB}
The upper bound in~\eqref{eq:genUB} is equivalent to the following expression
\begin{align}
R^+(B,W) = \frac{1}{W+1} H(\rvs_{B+1}, \rvs_{B+2}, \ldots, \rvs_{B+W+1}|\rvs_{0}). \label{eq:genUB_Slepian-Wolf}
\end{align}\hfill$\Box$
\end{corol}
The proof of Corollary~\ref{corol:genUB} is provided in Appendix~\ref{app:Cor1}.  We make several remarks. First, the entropy term in \eqref{eq:genUB_Slepian-Wolf} is equivalent to the sum-rate constraint associated with the Slepian-Wolf  coding scheme in simultaneously recovering $\{\rvs^n_{B+1}, \rvs^n_{B+2}, \ldots, \rvs^n_{B+W+1}\}$ when $\rvs^n_{0}$ is known. Note that due to the stationarity of the source process, the rate expression in \eqref{eq:genUB_Slepian-Wolf} suffices for recovering from any burst erasure of length up to $B$, spanning an arbitrary interval. Second, note that in \eqref{eq:genUB_Slepian-Wolf} we amortize over a window of length $W+1$ as $\{\rvs_{B+1}^n, \ldots, \rvs_{B+W+1}^n\}$ are recovered simultaneously at time $t=B+W+1$. Note that this is the maximum window length over which we can amortize due to the decoding constraint. Third, the results in Theorem~\ref{thm:genUB_LB} immediately apply when the channel introduces multiple bursts with a guard spacing of at least $W+1$. This property arises due to the Markov nature of the source. Given a source sequence at time $i$, all the future source sequences $\{\rvs_t^n\}_{t > i}$ are independent of the past $\{\rvs_t^n\}_{t<i}$ when conditioned on $\rvs_i^n$. Thus when a particular source sequence is reconstructed at the destination, the decoder becomes oblivious to past erasures.  Finally, while the results in Theorem~\ref{thm:genUB_LB} are stated for the rate-recovery function over an infinite horizon, upon examining the proof of Theorem~\ref{thm:genUB_LB}, it can be verified that both the upper and lower bounds  hold for the finite horizon case, i.e. $R_{\mathcal L}(B, W),$ when $\cL \ge B+W$.

A {\em symmetric} source is defined as a Markov source such that the underlying Markov chain is also reversible i.e., the random variables  satisfy   $(\rvs_0,\ldots, \rvs_l) \stackrel{\text{d}}{=} (\rvs_l, \ldots, \rvs_0)$, where the equality is in the sense of distribution~\cite{markov}.  Of particular interest to us is the following property satisfied for each $t$:
\begin{align}
p_{\rvs_{t+1}, \rvs_{t}} (s_a, s_b) = p_{\rvs_{t-1},\rvs_t}(s_a, s_b), \quad \forall s_a, s_b \in \cS\label{eq:symmetric}
\end{align}
i.e., we can ``exchange'' the source pair $(\rvs_{t+1}^n,\rvs_t^n)$ with $(\rvs_{t-1}^n,\rvs_t^n)$ without affecting the joint distribution.  An example of a symmetric source is the binary symmetric source: $\rvs_t^n = \rvs_{t-1}^n \oplus \rvz_t^n$, where $\{\rvz_t^n\}_{t\ge 0}$ is  an i.i.d.\ binary source process (in both temporal and spatial  dimensions)  with the marginal distribution ${\Pr(\rvz_{t,i}=0)=p}$, the marginal distribution $\Pr(\rvs_{t,i}=0) = \Pr(\rvs_{t,i}=1) =\frac{1}{2}$ and $\oplus$ denotes modulo-2 addition. 
\begin{corol}
\label{thm:binning}
For the class of symmetric Markov sources that satisfy~\eqref{eq:symmetric}, the lossless rate-recovery function when restricted to the class of memoryless encoders i.e., $\rvf_i = \cF_i(\rvs_i^n)$, is given by
\begin{align}
R(B,W) = \frac{1}{W+1} H(\rvs_{B+1}, \rvs_{B+2}, \ldots, \rvs_{B+W+1}|\rvs_{0}).\label{eq:rate-rec-binning}
\end{align}
\hfill$\Box$
\end{corol}
The proof of Corollary~\ref{thm:binning} is presented in Section~\ref{sec:Symmetric}.  The converse is obtained by again using a multi-terminal source coding problem, 
but obtaining a tighter bound by exploiting the memoryless property of the encoders and the symmetric structure~\eqref{eq:symmetric}.


\subsection{Gauss-Markov Sources}
\label{sec:GMS}
We study the lossy rate-recovery function when  $\{\rvs_{i}^n\}$  is sampled i.i.d. from a zero-mean Gaussian distribution, $\mathcal{N}(0,\sigma^2_{s})$,  along the spatial dimension and forms a first-order Markov chain across the temporal dimension  i.e., 
\begin{align}
\rvs_{i} = \rho \rvs_{i-1} + \rvn_{i} \label{eq:GM-Def}
\end{align} where $\rho \in (0,1)$ and $\rvn_{i} \sim \mathcal{N}(0,\sigma^{2}_{s}(1-\rho^2))$. Without loss of generality we assume $\sigma^2_{s} = 1$. We  consider the quadratic distortion measure $d(\rvs_{i}, \hat{\rvs}_{i}) = (\rvs_{i} -\hat{\rvs}_{i})^2$ between the source symbol $\rvs_{i}$ and its reconstruction $\hat{\rvs}_{i}$.
In this paper we  focus on the special case of $W=0$, where the reconstruction must begin immediately after the burst erasure.  We briefly remark about the case when $W>0$ at the end of Section~\ref{sec:UpperGM}.
As stated before unlike the lossless case, the results of Gauss-Markov sources for single burst erasure channels do not readily extend to the multiple burst erasures case.  Therefore, we treat the two cases separately.         

\subsubsection{Channels with Single Burst Erasure}
\label{sec:GMS_SE}
In this channel model, as stated in \eqref{eq:chModel}, we assume that  the channel can introduce a single burst erasure of length up to $B$ during the transmission period. Define ${R_{\textrm{GM-SE}}(B,D)} \triangleq {R(B, W=0,D)}$ as the lossy rate-recovery function of Gauss-Markov sources with single burst erasure channel model. 
\begin{prop}[Lower Bound--Single Burst]
\label{prop:GML}
The lossy rate-recovery function of the Gauss-Markov source for single burst erasure channel model when $W=0$ satisfies
\begin{multline}
R_{\textrm{GM-SE}}(B,D) \ge R^{-}_{\textrm{GM-SE}}(B,D) \triangleq \\ \frac{1}{2}   \log \left( \frac{D\rho^2+1-\rho^{2(B+1)} + \sqrt{\Delta}}{2D}\right) \label{eq:thm-1}
\end{multline} where $\Delta \triangleq (D\rho^2+1-\rho^{2(B+1)})^2 - 4D\rho^2(1-\rho^{2B})$.
$\hfill\Box$
\end{prop}

The proof of Prop.~\ref{prop:GML} is presented in Section~\ref{sec:LowerGM}. The proof considers the recovery of a source sequence $\rvs_t^n$, given a burst erasure in the interval $[t-B, t-1]$ and extends the lower bounding technique in Theorem~\ref{thm:genUB_LB} to incorporate the distortion constraint.


\begin{center}
\begin{figure*}
\begin{minipage}[t]{0.45\linewidth}
\begin{center}
\vspace{1em}
\includegraphics[width=0.9\textwidth]{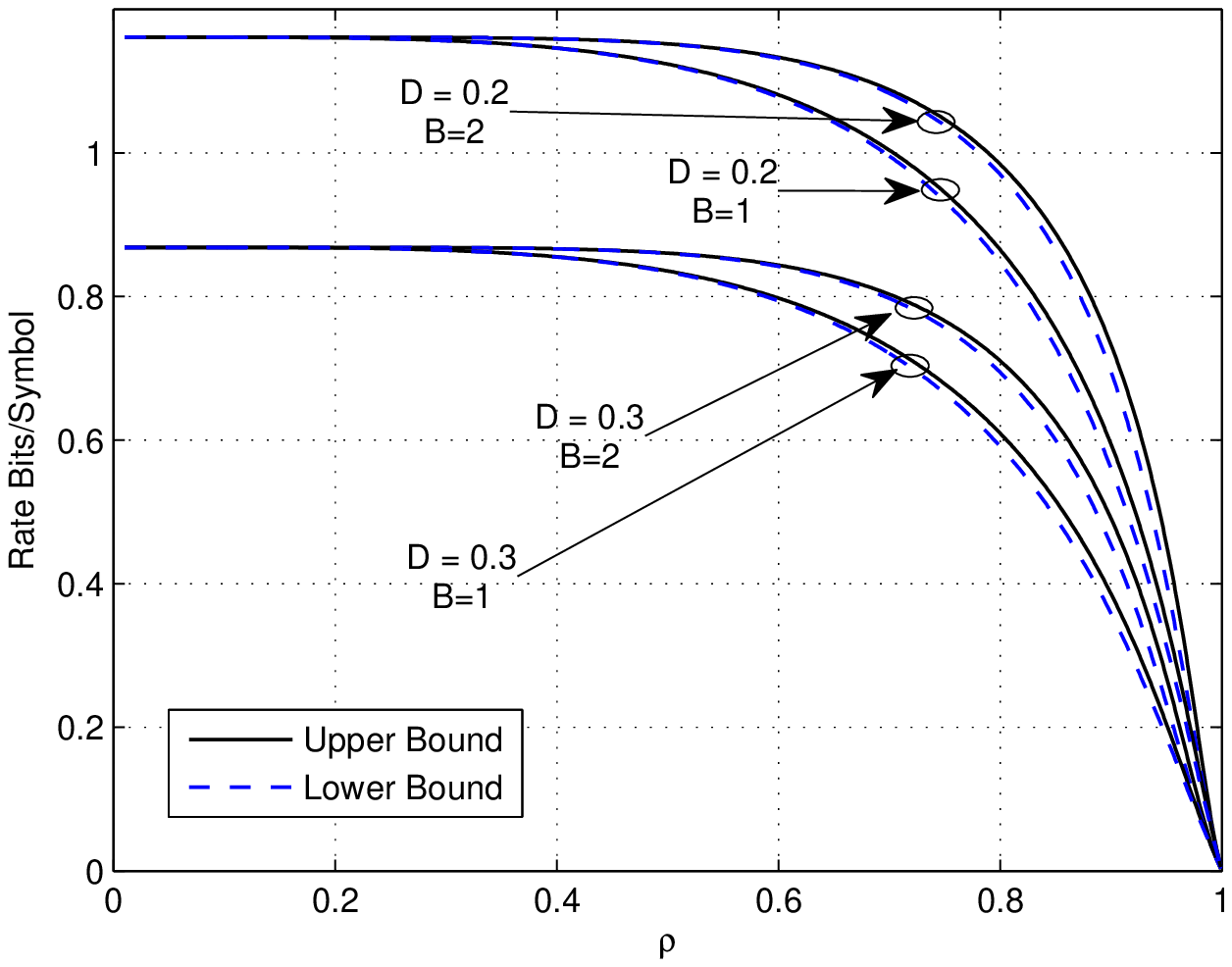}
\caption{Lower and upper bounds of lossy rate-recovery function $R_{\textrm{GM-SE}}(B,D)$ versus $\rho$ for $D=0.2$, $D=0.3$ and $B=1$, $B=2$.}
\label{fig:GM1}
\end{center}
\end{minipage}\hspace{0.7cm}\begin{minipage}[t]{0.45\linewidth}
\begin{center}
\vspace{1em}
\includegraphics[width=0.9\textwidth]{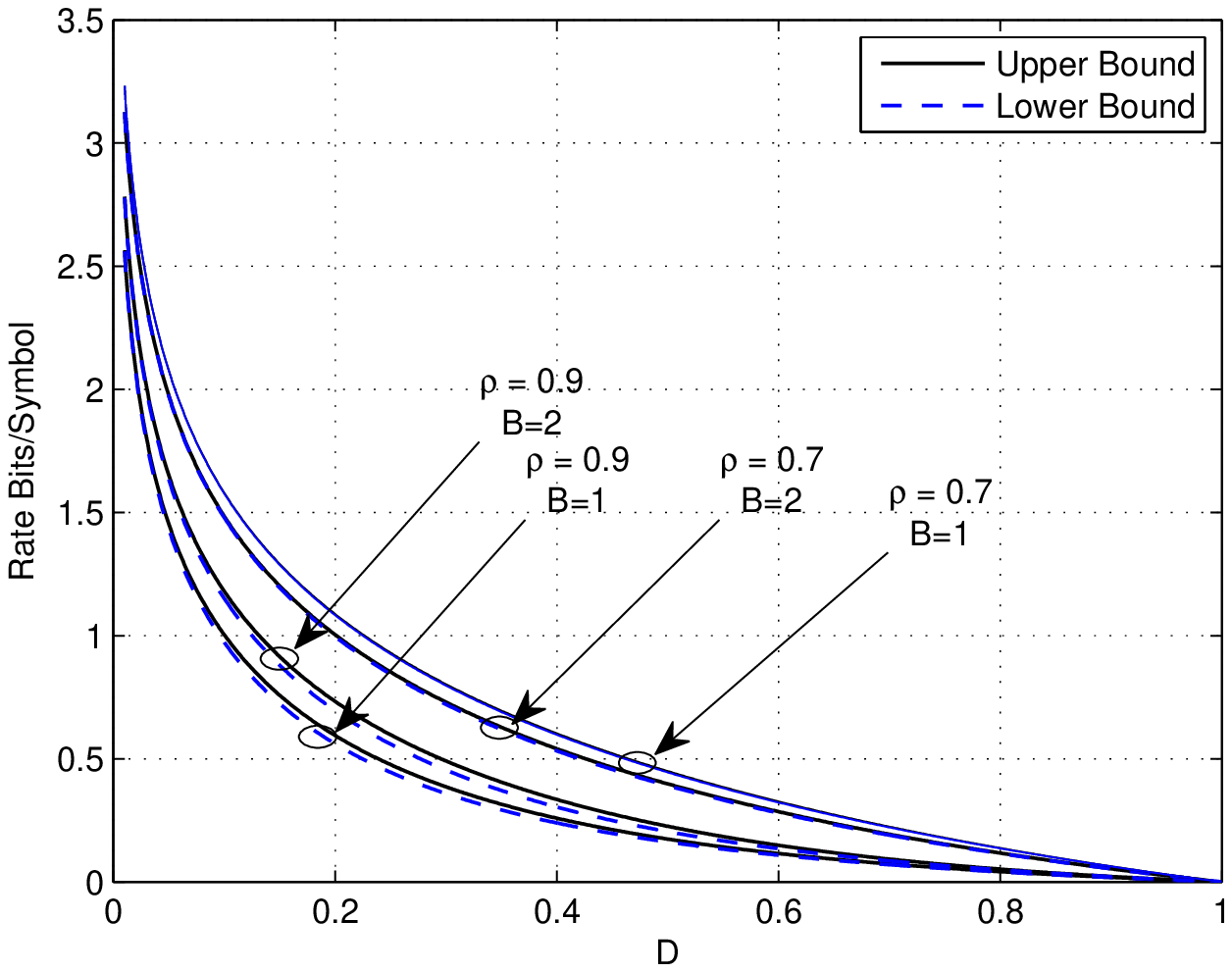}
\caption{Lower and upper bounds of lossy rate-recovery function $R_{\textrm{GM-SE}}(B,D)$ versus $D$ for $\rho=0.9$, $\rho=0.7$ and $B=1$, $B=2$.}
\label{fig:GM2}
\end{center}
\end{minipage}
\end{figure*}
\end{center}

\begin{prop}[Upper Bound--Single Burst]
\label{prop:GMAch}
The lossy rate-recovery function of the Gauss-Markov source for single burst erasure channel model when $W=0$ satisfies
\begin{align}
R_{\textrm{GM-SE}}(B,D) \le R^{+}_{\textrm{GM-SE}}(B,D) &\triangleq  I(\rvs_{t};\rvu_{t} | \tilde{\rvs}_{t-B})\label{eq:GM-LB-R}
\end{align} where $\rvu_{t} \triangleq \rvs_{t} + \rvz_{t},$ and $\rvz_{t}$ is sampled i.i.d.\ from $\mathcal{N}(0, \sigma^2_{z})$. Also $\tilde{\rvs}_{t-B} \defeq {\rvs}_{t-B} +\rve$ and  $\rve\sim \mathcal{N}\left(0, \Sigma(\sigma^{2}_{z})/(1-\Sigma(\sigma^{2}_{z}))\right)$ with 
 \begin{multline}
\Sigma(\sigma^{2}_{z}) \triangleq \\ \frac{1}{2}\sqrt{(1-\sigma^2_{z})^2(1-\rho^2)^2+4\sigma^2_{z}(1-\rho^2) } + \frac{1-\rho^2}{2}(1-\sigma^2_{z}) \label{eq:Sigma1},
\end{multline} is independent of all other random variables. The test channel noise $\sigma^2_{z}>0$ is chosen to satisfy
\begin{align}
\left[\frac{1}{\sigma^2_{z}} + \frac{1}{1-\rho^{2B}(1-\Sigma(\sigma^2_z))}\right]^{-1} \le D.
\end{align} This is equivalent to $\sigma^2_{z}$ satisfying
\begin{align}
 E\left[(\rvs_{t}-\hat{\rvs}_t)^2\right] \le D, \label{eq:GM-LB-RR}
\end{align} where $\hat{\rvs}_{t}$ denotes the minimum mean square estimate (MMSE) of $\rvs_{t}$ from $\{\tilde{\rvs}_{t-B}, \rvu_{t}\}$.  
$\hfill\Box$
\end{prop}


The following alternative rate expression for the achievable rate in Prop.~\ref{prop:GMAch}, provides a more explicit interpretation of the coding scheme.
\begin{align}
R^{+}_{\textrm{GM-SE}}(B,D) =  \lim_{t\to \infty} I(\rvs_{t};\rvu_{t} | [\rvu]_{0}^{t-B-1})\label{eq:asym1}
\end{align} where the random variables $\rvu_t$ are obtained using the same test channel in Prop.~\ref{prop:GMAch}. Notice that the test channel noise $\sigma^2_{z}>0$ is chosen to satisfy $E\left[(\rvs_{t}-\hat{\rvs}_t)^2\right] \le D$ where $\hat{\rvs}_{t}$ denotes the MMSE of $\rvs_{t}$ from $\{[\rvu]_{0}^{t-B-1}, \rvu_{t}\}$ in steady state, i.e. $t\to \infty$. Notice that~\eqref{eq:asym1} is based on a quantize and binning scheme when the receiver has side information sequences $\{\rvu_0^n, \ldots, \rvu_{t-B-1}^n\}$. 
The proof of Prop.~\ref{prop:GMAch} which is presented in Section~\ref{sec:UpperGM} also involves establishing that the worst case erasure pattern during the recovery of  $\hat{\rvs}_t^n$ spans the interval $[t-B-1,t-1]$. The proof is considerably more involved as the reconstruction sequences $\{\rvu_t^n\}$ do not form a Markov chain.

As we will show subsequently, the upper and lower bounds in Prop.~\ref{prop:GML} and Prop.~\ref{prop:GMAch} coincide in the high resolution limit. Numerical evaluations suggest that the bounds are close for a wide range of parameters.  Fig.~\ref{fig:GM1} and Fig.~\ref{fig:GM2} illustrate some sample comparison plots. 

\begin{figure*}
        \centering
        \begin{subfigure}[b]{0.405\textwidth}
                \centering
                \includegraphics[width=\textwidth]{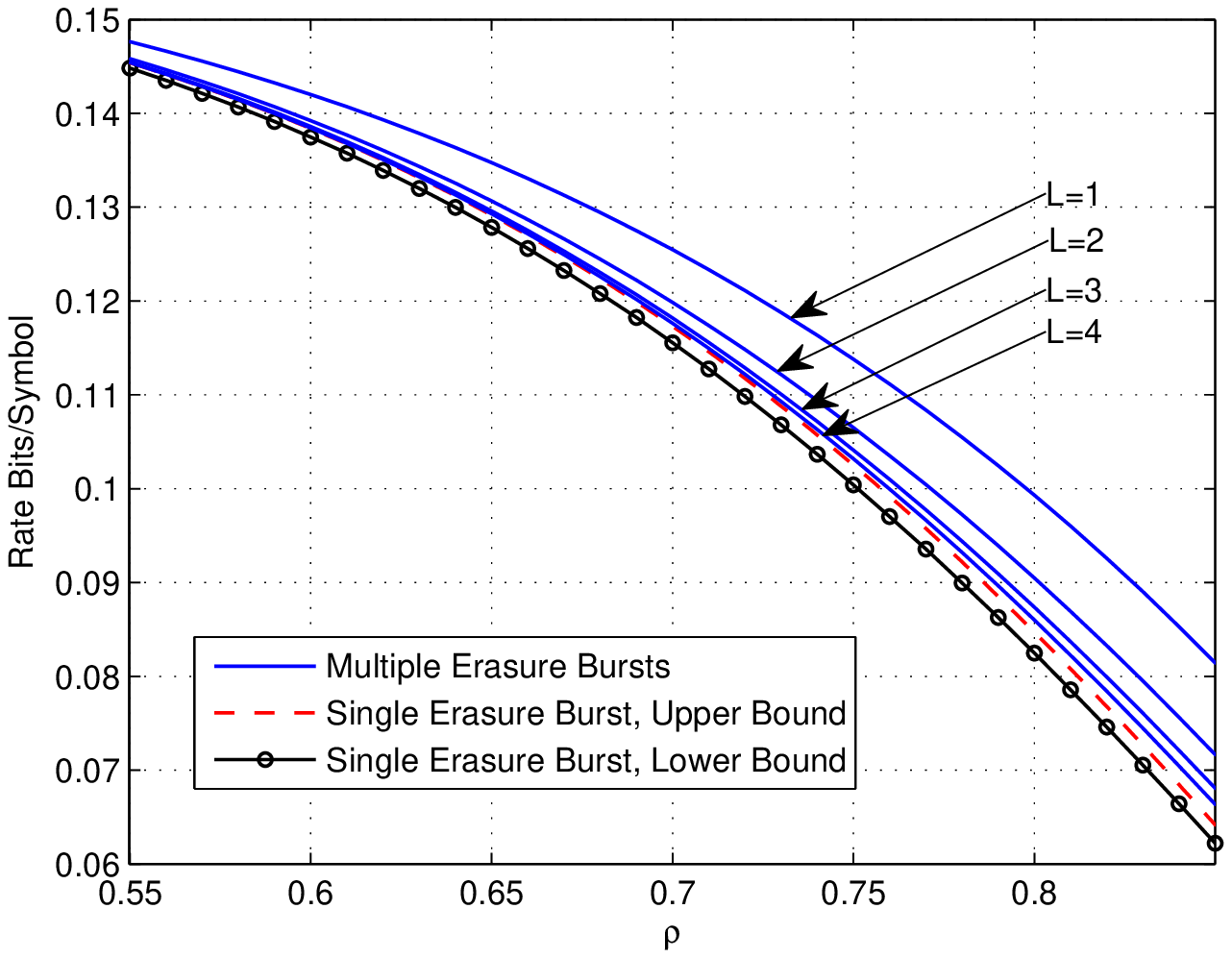}
                \caption{$D=0.8$}
        \end{subfigure}%
        ~ 
        \begin{subfigure}[b]{0.405\textwidth}
                \centering
                \includegraphics[width=\textwidth]{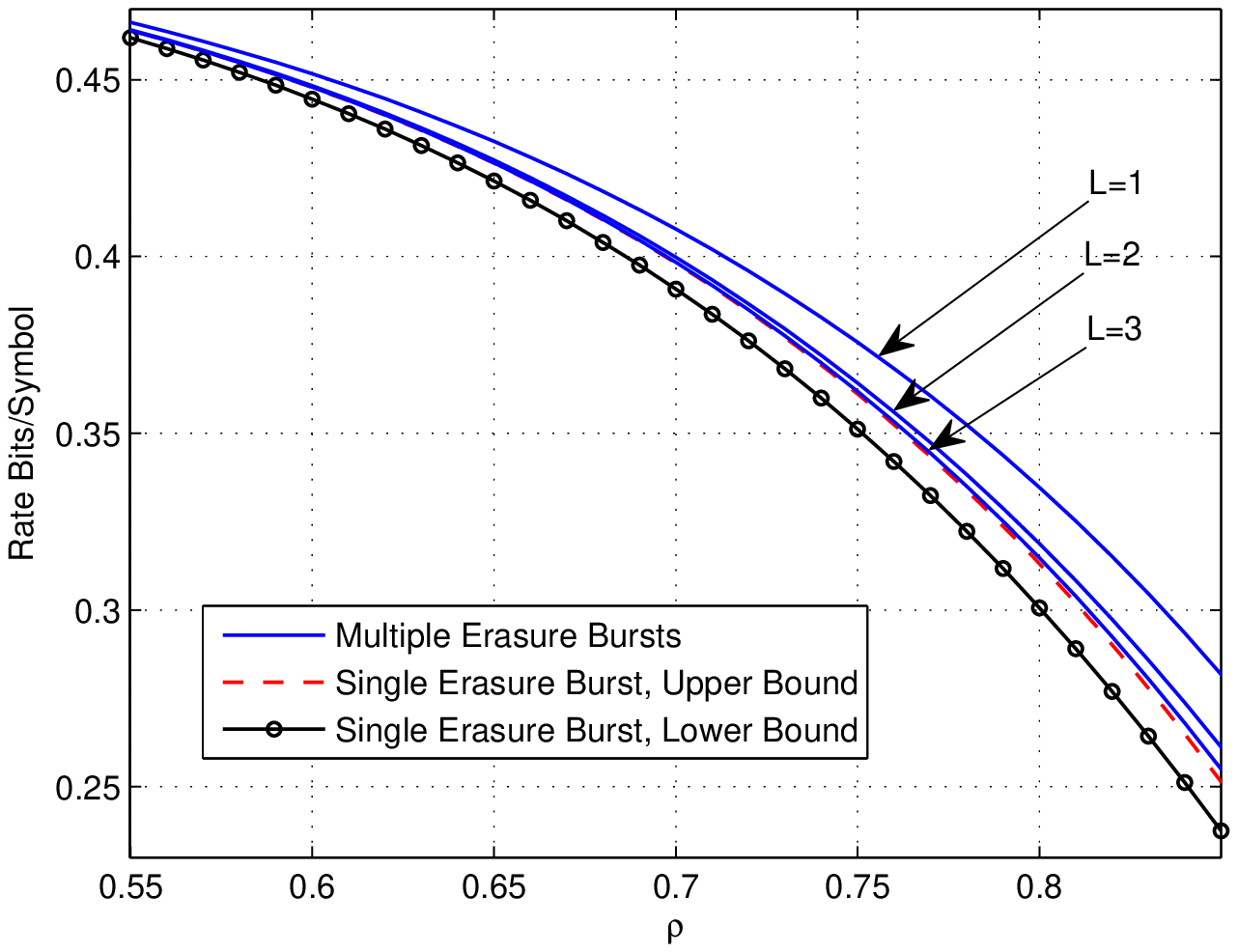}
                \caption{$D=0.5$}
        \end{subfigure}
        \caption{Achievable rates for multiple burst erasures model for different values of guard length $L$ separating burst erasures comparing to single burst erasure. As $L$ grows, the rate approaches the single erasure case. The lower bound for single erasure case is also plotted for comparison ($B=1$).}
        \label{fig:GMvsSE}
\end{figure*}

\subsubsection{Channels with  Multiple Burst Erasures}
\label{sec:GMS_ME} 
We also consider the case where the channel can introduce multiple burst erasures, each of length no greater than $B$
and with a guard interval of length at-least $L$ separating consecutive bursts. The encoder is defined as in \eqref{eq:f-enc}. We again only consider the case when $W=0$.   Upon observing the sequence $\{\rvg_i\}_{i\ge 0}$, the decoder is required to reconstruct each source sequence with zero delay, i.e., 
\begin{align}
\hat{\rvs}^{n}_i=\mathcal{G}_{i}(\rvg_0, \rvg_1, \ldots, \rvg_{i}, \rvs_{-1}^n), \quad \textrm{ whenever } \rvg_{i} \neq \star \label{eq:decoder-def-ME}
\end{align} such that the reconstructed source sequence $\hat{\rvs}^{n}_i$ satisfies an average mean square distortion of $D$. The destination is not required to produce the source vectors that appear during any of the burst erasures. The rate $R(L,B,D)$ is feasible if a sequence of encoding and decoding functions exists that satisfies the average distortion constraint. The minimum feasible rate $R_{\textrm{GM-ME}}(L,B,D)$, is the lossy rate-recovery function.

\begin{prop}[Upper Bound--Multiple Bursts]
\label{prop:GM-ME}
The lossy rate-recovery function  $R_{\textrm{GM-ME}}(L,B,D)$  for Gauss-Markov sources over the multiple burst erasures channel satisfies the following upper bound:
\begin{multline}
R_{\textrm{GM-ME}}(L,B,D) \le R_{\textrm{GM-ME}}^{+}(L,B,D)\triangleq \\ I(\rvu_{t};\rvs_{t} | \tilde{\rvs}_{t-L-B}, [\rvu]_{t-L-B+1}^{t-B-1}) \label{eq:GM-ME-Ach}
\end{multline} where $\tilde{\rvs}_{t-L-B} = {\rvs}_{t-L-B} +\rve$, where $\rve\sim \mathcal{N}(0, D/(1-D))$. Also for any $i$, $\rvu_{i} \triangleq \rvs_{i} + \rvz_{i}$ and $\rvz_{i}$ is sampled i.i.d.\ from $\mathcal{N}(0, \sigma^2_{z})$ and the noise in the test channel, $\sigma^{2}_{z}>0$  satisfies
\begin{align}
E\left[(\rvs_{t}-\hat{\rvs}_t)^2\right] \le D \label{eq:GM-ME-Ach-D}
\end{align}  and $\hat{\rvs}_{t}$ denotes the MMSE estimate of $\rvs_{t}$ from $\{\tilde{\rvs}_{t-L-B}, [\rvu]_{t-L-B+1}^{t-B-1}, \rvu_{t}\}$. $\hfill\Box$
\end{prop}

The proof of Prop.~\ref{prop:GM-ME} presented in Section~\ref{sec:GM_ME} is again based on quantize-and-binning technique and involves characterizing the worst-case erasure pattern by the channel. 
Note also that the rate expression in \eqref{eq:GM-ME-Ach} depends on the minimum guard spacing $L$, the maximum burst erasure length $B$ and distortion $D$, but  is not a function of time index $t$, as the test channel is time invariant and the source process is stationary. An expression for computing $\sigma_z^2$ is provided in Section~\ref{sec:GM_ME}. While we do not provide a lower bound for $R_{\textrm{GM-ME}}(L,B,D)$
we remark that the lower bound in Prop.~\ref{prop:GML} also applies to the multiple burst erasures setup.


Fig.~\ref{fig:GMvsSE} provides numerical evaluation of the achievable rate for different values of $L$. We note that even for  $L$ as small as $4$,  the achievable rate in Prop.~\ref{prop:GM-ME} is virtually identical to the rate for single burst erasure in Prop.~\ref{prop:GMAch}. This strikingly fast convergence to the single burst erasure rate appears due to the exponential decay in  the correlation coefficient between source samples as time-lag increases.

%

\subsubsection{High Resolution Regime}
\label{sec:highres}

For both the single and multiple burst erasures models, the upper and lower bounds on lossy rate-recovery function for $W=0$ denoted by $R(L, B, D)$ coincide in the high resolution limit as stated below. 
\begin{corol}
\label{corol:HR}
In the high resolution limit, the Gauss-Markov lossy rate-recovery function satisfies the following:
\begin{align}
R(L, B, D) = \frac{1}{2} \log\left( \frac{1-\rho^{2(B+1)}}{D} \right) + o(D)  \label{eq:HR} .
\end{align} where $\lim _{D\to 0 }o(D)=0$. \hfill$\Box$
\end{corol}
\begin{figure}
\centering
\includegraphics[width=.5\textwidth]{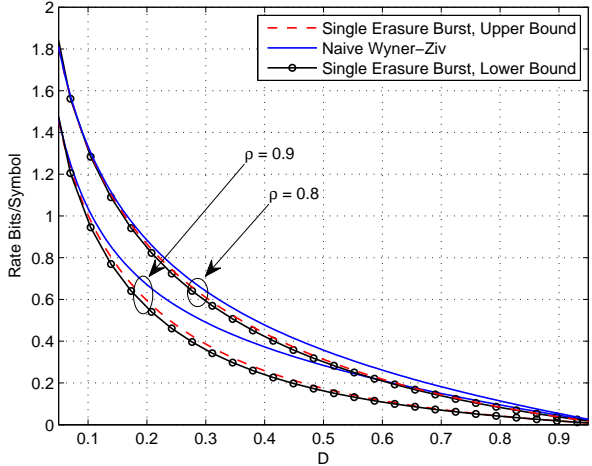}
\caption{A comparison of achievable rates for the Gauss-Markov source ($B=1$).}
\label{fig:HigRes}
\end{figure}

The proof of Corollary~\ref{corol:HR} is presented in Section \ref{sec:HR}. It is based on evaluating the asymptotic behaviour of the lower bound in \eqref{eq:thm-1} and the upper bound in Prop.~\ref{prop:GM-ME}, in high resolution regime. Notice that the rate expression in~\eqref{eq:HR} does not depend on the guard separation $L$. The intuition behind this is as follows. In the high resolution regime, the output of the test channel, i.e. $\rvu_{t}$, becomes very close to the original source $\rvs_{t}$. Therefore the Markov property of the original source is approximately satisfied by these auxiliary random variables and hence the past sequences are not required. The rate in~\eqref{eq:HR} can also be approached by a \emph{Naive Wyner-Ziv} coding scheme that only makes use of the most recently available sequence at the decoder~\cite{oohama1997}. The rate of this scheme is given by:
\begin{align}
R_{\textrm{NWZ}}(B,D) \triangleq I(\rvs_{t} ; \rvu_{t} | \rvu_{t-B-1}) \label{eq:NWZ}
\end{align} where for each $i$, $\rvu_{i} = \rvs_{i} + \rvz_{i}$ and $\rvz_{i} \sim \cN(0,\sigma^2_{z})$ and $\sigma^2_{z}$ satisfies the following distortion constraint
\begin{align}
E[(\rvs_{t}-\hat{\rvs}_{t})^2] \le D 
\end{align} where $\hat{\rvs}_{t}$ is the MMSE estimate of $\rvs_{t}$ from $\{\rvu_{t-B-1}, \rvu_{t}\}$.

Fig.~\ref{fig:HigRes} reveals that while the rate in~\eqref{eq:NWZ} is near optimal in the high resolution limit, it is in general sub-optimal when compared to the rates in~\eqref{eq:GM-ME-Ach} when $\rho=0.9$. As we decrease $\rho$, the performance loss associated with this scheme appears to reduce.

\subsection{Gaussian Sources with Sliding Window Recovery Constraints}
\label{subsec:GSW}

In this section we consider a specialized source model and distortion constraint, where it is possible to improve upon the binning based upper bound. Our proposed scheme attains the rate-recovery function for this special case and is thus optimal. This example illustrates that the binning based scheme can be sub-optimal in general.

{\subsubsection{Source Model}
We consider a sequence of i.i.d.\ Gaussian source sequences i.e., at time $i$, $\rvs^n_{i}$ is sampled i.i.d.\ 
according to a zero mean unit variance Gaussian distribution $\cN(0,1),$ independent of the past sources. At each time we associate an auxiliary source
\begin{align}
\rvbt^n_{i}=\begin{pmatrix}
\rvs^n_{i}~
\rvs^n_{i-1}~\ldots~
\rvs^n_{i-K}\end{pmatrix}
\label{Gauss-model}
\end{align}
which is a collection of the past ${K+1}$ source sequences. Note that $\rvbt_i^n$ constitutes a first-order Markov chain. We will define a reconstruction constraint with the sequence $\rvbt_i^n$.

\subsubsection{Encoder}
The (causal) encoder at time $i$ generates an output given by 
$\rvf_{i} = \cF_i(\rvs_{-1}^n, \ldots, \rvs_i^n) \in [1,2^{nR}]$. 

\subsubsection{Channel Model}
The channel can introduce a burst  erasure of length up to $B$ in an arbitrary interval $[j, j+B-1]$.

\subsubsection{Decoder}
At time $i$ the decoder is interested in reproducing a collection of past $K+1$ sources\footnote{In this section it is sufficient to assume that any source sequence with a time index $j < -1$ is a constant sequence.}
within a distortion vector $\bd=(d_0,d_1,\cdots,d_K)$ i.e., at time $i$ the decoder is interested in reconstructing $(\hat{\rvs}_i^n,\ldots, \hat{\rvs}_{i-K}^n)$ where $E\left[||\rvs_{i-l}^n- \hat{\rvs}_{i-l}^n||^2\right] \le nd_l$ must be satisfied for $l\in[0,K]$.
We assume throughout that $d_0\le d_1\le \ldots \le d_K$ which corresponds to the requirement that the  more recent source sequences must be reconstructed with a smaller average distortion.

In Fig.~\ref{fig:Gauss_Slide}, the source symbols $\rvs_i$ are shown as white circles. The symbols $\rvbt_{i}$ and  $\hat{\rvbt}_i$ are also illustrated for $K=2$. The different shading for the sub-symbols in $\hat{\rvbt}_i$  corresponds to different distortion constraints.

If a burst erasure spans the interval $[j, j+B-1],$  the decoder is not required to output a reproduction of the sequences $\rvbt_i^n$ for $i \in [j, j+B+W-1]$.
}

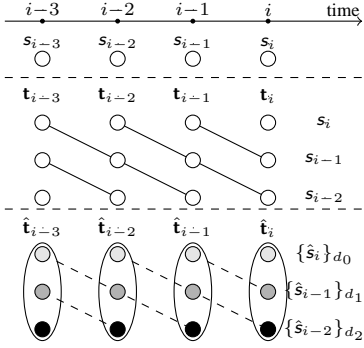
\begin{figure}
\centering
%
%
%
%
%

\begin{tikzpicture}[scale=0.5]

\draw (6,1) ellipse (.5 and 1.3);
\draw  (4,1) ellipse (.5 and 1.3);
\draw  (2,1) ellipse (.5 and 1.3);
\draw (0,1) ellipse (.5 and 1.3);

\draw[dashed] (0,2)--(4,0);
\draw[dashed] (2,2)--(6,0);
\draw[dashed] (4,2)--(6,1);
\draw[dashed] (0,1)--(2,0);

\foreach \t in {0,2,4,6}{
\draw[fill=black!100!white] (\t,0) circle (.2cm);
}
\draw (7.5,0)--(7.5,0) node{$\scriptstyle\{\hat{\rvs}_{i-2}\}_{d_2}$};
\foreach \t in {0,2,4,6}{
\draw[fill=black!30!white] (\t,1) circle (.2cm);
}
\draw (7.5,1)--(7.5,1) node {$\scriptstyle\{\hat{\rvs}_{i-1}\}_{d_1}$};

\foreach \t in {0,2,4,6}{
\draw[fill=black!10!white] (\t,2) circle (.2cm);
}
\draw (7.5,2)--(7.5,2) node {$\scriptstyle\{\hat{\rvs}_{i}\}_{d_0}$};

\draw (6,2.7)--(6,2.7) node {$\scriptstyle \hat{\rvbt}_{i}$};
\draw (4,2.7)--(4,2.7) node {$\scriptstyle\hat{\rvbt}_{i-1}$};
\draw (2,2.7)--(2,2.7) node {$\scriptstyle\hat{\rvbt}_{i-2}$};
\draw (0,2.7)--(0,2.7) node {$\scriptstyle\hat{\rvbt}_{i-3}$};

\draw[dashed] (-1,3.2)--(8.5,3.2);

\draw[-] (0,2+3.5)--(4,0+3.5);
\draw[-] (2,2+3.5)--(6,0+3.5);
\draw[-] (4,2+3.5)--(6,1+3.5);
\draw[-] (0,1+3.5)--(2,0+3.5);

\foreach \t in {0,2,4,6}{
\draw[fill=white] (\t,3.5) circle (.2cm);
}
\draw (7.5,3.5)--(7.5,3.5) node {$\scriptstyle{\rvs}_{i-2}$};
\foreach \t in {0,2,4,6}{
\draw[fill=white] (\t,4.5) circle (.2cm);
}
\draw (7.5,4.5)--(7.5,4.5) node {$\scriptstyle{\rvs}_{i-1}$};

\foreach \t in {0,2,4,6}{
\draw[fill=white] (\t,5.5) circle (.2cm);
}
\draw (7.5,5.5)--(7.5,5.5) node{$\scriptstyle{\rvs}_{i}$};

\draw  (6,2.7+3.5)--(6,2.7+3.5) node {$\scriptstyle{\rvbt}_{i}$};
\draw  (4,2.7+3.5)--(4,2.7+3.5) node {$\scriptstyle{\rvbt}_{i-1}$};
\draw  (2,2.7+3.5)--(2,2.7+3.5) node{$\scriptstyle{\rvbt}_{i-2}$};
\draw  (0,2.7+3.5)--(0,2.7+3.5) node{$\scriptstyle{\rvbt}_{i-3}$};

\draw[dashed] (-1,3.5+3.2)--(8.5,3.5+3.2);

\draw (6,7.6)--(6,7.6) node{$\scriptstyle{\rvs}_{i}$};
\draw (4,7.6)--(4,7.6) node{$\scriptstyle{\rvs}_{i-1}$};
\draw (2,7.6)--(2,7.6) node{$\scriptstyle{\rvs}_{i-2}$};
\draw (0,7.6)--(0,7.6) node{$\scriptstyle{\rvs}_{i-3}$};

\foreach \t in {0,2,4,6}{
\draw[fill=white] (\t,7.2) circle (.2cm);
}

\foreach \t in {0,2,4,6}{
\draw[fill=black] (\t,8.2) circle (.05cm);
}

\draw (8,8.5) -- (8,8.5) node {$\color{black}\scriptstyle\textrm{time}$};
\draw[->] (-1,8.2)--(8.5,8.2);
\draw (6,8.5) -- (6,8.5) node {$\color{black}\scriptstyle{i}$};
\draw (4,8.5) -- (4,8.5) node {$\color{black}\scriptstyle{i-1}$};
\draw (2,8.5) -- (2,8.5) node {$\color{black}\scriptstyle{i-2}$};
\draw (0,8.5) -- (0,8.5) node {$\color{black}\scriptstyle{i-3}$};

\end{tikzpicture}
\caption{Schematic of the Gaussian sources with sliding window recovery constraints for $K=2$. The source $\rvs_i$, drawn as white circles, are independent sources and $\rvbt_i$ is defined as a collection of $K+1=3$ most recent sources. The source symbols along the diagonal lines are the same. The decoder at time $i$ recovers $\rvs_{i}$, $\rvs_{i-1}$ and $\rvs_{i-2}$ within distortions $d_0$, $d_1$ and $d_2$, respectively where $d_0\le d_{1}\le d_{2}$. In figure the colour density of the circle represents the amount of reconstruction distortion. }
\label{fig:Gauss_Slide}
\end{figure}

The {\em lossy rate-recovery function} denoted by $R(B, W, \bd)$ is the minimum rate required to satisfy these constraints.

\begin{remark}
One motivation for considering the above setup is that the decoder might be interested in computing a function of the last $K+1$ source sequences at each time e.g.,, $\rvv_i = \sum_{j=0}^K \al^j \rvs_{i-j}$. A robust coding scheme, when the coefficient $\al$ is not known to the encoder is to communicate $\rvs^n_{i-j}$ with distortion $d_{j}$ at time $i$ to the decoder.
\end{remark}
\begin{thm}
For the proposed Gaussian source model with {a non-decreasing distortion }vector $\bd = (d_0,\ldots, d_K)$ with $0 < d_i \le 1$,  the lossy rate-recovery function is given by
\begin{multline}
R(B,W,\bd) = \frac{1}{2}\log \bigg(\frac{1}{d_0}\bigg) + \\ \frac{1}{W+1}\sum_{k=1}^{\min\{K-W,B\}}\frac{1}{2}\log \bigg(\frac{1}{d_{W+k}}\bigg). \label{eq:det-rate-Gaussian}
\end{multline}$\hfill\Box$
\label{thm:gauss-rate}
\end{thm}

The proof of Theorem~\ref{thm:gauss-rate} is provided in Section~\ref{sec:Gauss}. The coding scheme for the proposed model involves using a successive refinement codebook for each sequence $\rvs_i^n$ to produce $B+1$ layers and carefully assigning the sequence of layered codewords to each channel packet. A simple quantize and binning scheme in general does not achieve the rate-recovery function in Theorem~\ref{thm:gauss-rate}. 
A numerical comparison of the lossy rate-recovery function with other schemes is presented in Section~\ref{sec:Gauss}. 

This completes the statement of the main results in this paper.

%
%
%
%
%

\section{General Upper and Lower Bounds on Lossless Rate-Recovery Function}
\label{sec:THM1}
In this section we present the proof of Theorem~\ref{thm:genUB_LB}. In particular, we show that the rate-recovery function satisfies the following lower bound.
\begin{align}
R\ge R^{-}(B,W) =H(\rvs_{1}|\rvs_{0}) + \frac{1}{W+1} I(\rvs_{B}, \rvs_{B+W+1}|\rvs_{0})\label{eq:lower_rep}.
\end{align} which is inspired by a connection to a
multi-terminal source coding problem introduced in Section~\ref{subsec:gen}. Based on this connection, the proof of the lower bound in general form in \eqref{eq:lower_rep} is presented in Section~\ref{subsec:genLB}.  
Then by proposing a coding scheme based on random binning, we show in Section~\ref{sec:UBLB} that the following rate is achievable. 
\begin{align}
R\ge R^{+}(B,W) =H(\rvs_{1}|\rvs_{0}) + \frac{1}{W+1} I(\rvs_{B}, \rvs_{B+1}|\rvs_{0}).
\end{align}

\subsection{Connection to Multi-terminal Source Coding Problem}
\label{subsec:gen}
We first present a multi-terminal source coding setup which captures the tension inherent in the streaming setup. We focus on the special case when $B=1$ and $W=1$. At any given time $j$ the encoder output $\rvf_j$  must satisfy two objectives simultaneously:  1) if $j$ is outside the error propagation period then the decoder should use $\rvf_j$ and the past sequences to reconstruct $\rvs_j^n$; 2) if $j$ is within the recovery period then $\rvf_j$ must only help in the recovery of a future source sequence. 
\begin{figure}
\begin{center}
\vspace{1em}
\begin{tikzpicture}

\def \a {0}{
\draw [white] (-1,-\a-.8) -- (-1,-\a-.8) node {$\color{black}\rvs^n_{j},\rvs^n_{j+1}$};
\fill [fill=white, draw=gray!50!black] (0,-.3-\a) -- (0,-2-\a) -- (1,-2-\a) -- (1,-.3-\a) -- (0,-.3-\a);
\draw [white] (.5,-\a-1.15) -- (.5,-\a-1.15) node {$\color{black}\scriptstyle\textrm{Encoder}$};
\draw [->] (-.9,-1.15) -- (0,-1.15);
\draw [->] (1,-.4) -- (3,-.4);
\fill [fill=white, draw=gray!50!black] (3,-.8-\a) -- (3,-\a) -- (4,-\a) -- (4,-.8-\a) -- (3,-.8-\a);
\draw [white] (3.5,-\a-0.25) -- (3.5,-\a-0.25) node {$\color{black}\scriptstyle\textrm{Decoder}$};
\draw [white] (3.5,-\a-0.55) -- (3.5,-\a-0.55) node {$\color{black}\scriptstyle{1}$};
\draw [->] (4,-.4) -- (5,-.4);
\draw [white] (5.5,-\a-0.4) -- (5.5,-\a-0.4) node {$\color{black}\hat{\rvs}^n_{j}$};
\draw [white] (2,-\a-0.15) -- (2,-\a-0.15) node {$\color{black}\rvf_{j}$};
\draw [->] (3.5,-\a+1) -- (3.5,-\a);
\draw [white] (4,-\a+.5) -- (4,-\a+.5) node {$\color{black}\rvs^n_{j-1}$};
}
\def \a {1.5}{
\draw [->] (1,-.4-\a) -- (3,-.4-\a);
\fill [fill=white, draw=gray!50!black] (3,-.8-\a) -- (3,-\a) -- (4,-\a) -- (4,-.8-\a) -- (3,-.8-\a);
\draw [white] (3.5,-\a-0.25) -- (3.5,-\a-0.25) node {$\color{black}\scriptstyle\textrm{Decoder}$};
\draw [white] (3.5,-\a-0.55) -- (3.5,-\a-0.55) node {$\color{black}\scriptstyle{2}$};
\draw [->] (4,-.4-\a) -- (5,-.4-\a);
\draw [white] (5.5,-\a-0.4) -- (5.5,-\a-0.4) node {$\color{black}\hat{\rvs}^n_{j+1}$};
\draw [white] (2,-\a-0.15) -- (2,-\a-0.15) node {$\color{black}\rvf_{j+1}$};
\draw [->] (3.5,-\a-1.8) -- (3.5,-\a-.8);
\draw [white] (4.2,-\a-1.3) -- (4.2,-\a-1.3) node {$\color{black}\rvs^n_{j-2}$};
}

\draw [-] (2.5,-.4) -- (2.5,-1.7);
\draw [->] (2.5,-1.7) -- (3,-1.7);

\end{tikzpicture}
\caption{Multi-terminal problem setup associated with our proposed streaming setup when $W=B=1$. The erasure at time $t=j-1$ leads to two virtual decoders with different side information as shown.}
\label{fig:multiterminal}
\end{center}
\end{figure}
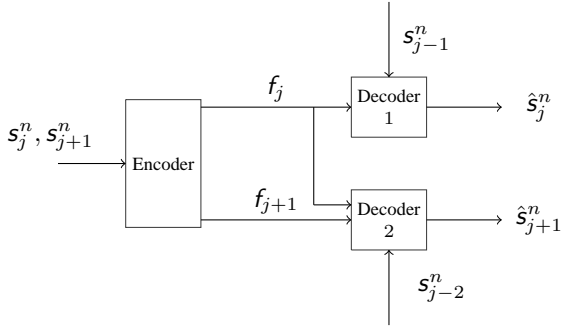

Fig.~\ref{fig:multiterminal} illustrates the multi-terminal source coding problem with one encoder and two  decoders that  captures these constraints.  The sequences $(\rvs_{j}^n, \rvs_{j+1}^n)$ are revealed to the encoder and produces outputs $\rvf_{j}$ and $\rvf_{j+1}$.
Decoder $1$  needs to recover $\rvs_{j}^n$ given $\rvf_{j}$ and $\rvs_{j-1}^n$ while decoder $2$ needs to recover $\rvs_{j+1}^n$
given $\rvs_{j-2}^n$ and $(\rvf_{j}, \rvf_{j+1})$. Thus decoder $1$ corresponds to the steady state  of the system when there is no loss while decoder $2$ corresponds to the recovery immediately after an erasure when $B=1$ and $W=1$. 
We note in advance that the multi-terminal source coding setup does not directly correspond to providing genie-aided side information in the streaming setup. In particular this setup does not account for the fact that the encoder has access to all previous source sequences and the decoders have access to past channel outputs. Nevertheless the main steps of the lower bound developed in the multi-terminal setup are then generalized rather naturally in the formal proof of the lower bound in the next sub-section.

For the above multi-terminal problem, we establish a  lower bound on the sum rate as follows:
\begin{align}
& n(R_1+R_2) \notag \\
&\ge H(\rvf_{j}, \rvf_{j+1}) \notag\\ &\ge H(\rvf_{j},\rvf_{j+1}|\rvs_{j-2}^n)\notag\\
&= H(\rvf_{j}, \rvf_{j+1}, \rvs_{j+1}^n | \rvs_{j-2}^n)- H(\rvs_{j+1}^n | \rvf_{j}, \rvf_{j+1}, \rvs_{j-2}^n)\notag\\
& = H(\rvf_{j}, \rvs_{j+1}^n | \rvs_{j-2}^n) + H(\rvf_{j+1} | \rvf_{j}, \rvs_{j-2}^n, \rvs_{j+1}^n)\notag\\ &\qquad - H(\rvs_{j+1}^n | \rvf_{j}, \rvf_{j+1}, \rvs_{j-2}^n)\label{eq:Aux1}\\
&\ge H(\rvf_{j}, \rvs_{j+1}^n|\rvs_{j-2}^n) - n\eps_n \label{eq:Fano_S3k}\\
&= H(\rvs_{j+1}^n|\rvs_{j-2}^n) + H(\rvf_{j}|\rvs_{j+1}^n, \rvs_{j-2}^n)-n\eps_n\notag\\
&\ge H(\rvs_{j+1}^n|\rvs_{j-2}^n) \!+\! H(\rvf_{j}|\rvs_{j+1}^n, \rvs_{j-1}^n, \rvs_{j-2}^n)-n\eps_n\label{eq:Cond_Red}\\
&\ge H(\rvs_{j+1}^n|\rvs_{j-2}^n) \!+\! H(\rvs^n_{j}|\rvs_{j+1}^n, \rvs_{j-1}^n, \rvs_{j-2}^n)-2n\eps_n\label{eq:Fano_S3k2}\\
&= H(\rvs_{j+1}^n|\rvs_{j-2}^n) \!+\! H(\rvs^n_{j}|\rvs_{j+1}^n, \rvs_{j-1}^n)-2n\eps_n\label{eq:Markov_S3k2}\\
&= nH(\rvs_3|\rvs_0) + nH(\rvs_1 |\rvs_2, \rvs_0) -2n\eps_n \label{eq:3k_der}
\end{align}
where~\eqref{eq:Aux1} follows from the chain rule of entropy, \eqref{eq:Fano_S3k} follows from the fact that $\rvs_{j+1}^n$ must be recovered from $\{\rvf_{j}, \rvf_{j+1}, \rvs_{j-2}^n\}$ at decoder $2$ hence Fano's inequality applies and~\eqref{eq:Cond_Red} follows from the fact that conditioning reduces entropy.
~Eq.~\eqref{eq:Fano_S3k2} follows from Fano's inequality applied to decoder $1$ and~\eqref{eq:Markov_S3k2} follows from the Markov chain associated with the source process.  Finally~\eqref{eq:3k_der} follows from the fact that the source process is memoryless. Dividing throughout by $n$ in~\eqref{eq:3k_der}  and taking $n\rightarrow\infty$ yields 
\begin{equation}
R_1 +R_2 \ge H(\rvs_1|\rvs_0,\rvs_2) + H(\rvs_3|\rvs_0).\label{eq:SumRateLower}
\end{equation}

{\bf Tightness of Lower Bound:} As a side remark, we note that the sum-rate lower bound in \eqref{eq:SumRateLower} can be achieved if Decoder 1 is further revealed $\rvs_{j+1}^n$. Note that the lower bound \eqref{eq:SumRateLower} also applies in this case since the Fano's Inequality applied to decoder 1 in \eqref{eq:Fano_S3k2} has $\rvs_{j+1}^n$ in the conditioning. We claim that $R_1 = H(\rvs_j|\rvs_{j+1},\rvs_{j-1})$ and $R_2 = H(\rvs_{j+1}|\rvs_{j-2})$ are achievable.  The encoder can achieve $R_1$ by  random binning of source $\rvs_{j}^n$ with $\{\rvs_{j-1}^n, \rvs_{j+1}^n\}$ as decoder 1's side information and achieve $R_2$ by random binning of source $\rvs_{j+1}^n$ with $\rvs^n_{j-2}$ as decoder 2's side information. Thus revealing the additional side information of $\rvs^n_{j+1}$ to decoder 1, makes the link connecting $\rvf_{j}$ to decoder 2 unnecessary.

 Also note that the setup in Fig.~\ref{fig:multiterminal} reduces to the source coding problem in~\cite{zigZag} if we set  $\rvs_{j-2}^{n} = \phi$.  It is also a successive refinement source coding problem with different side information at the decoders and special distortion constraints at each of the decoders. 
However to the best of our knowledge the multi-terminal problem in Fig.~\ref{fig:multiterminal} has  not been  addressed in the literature nor has the connection to our proposed streaming setup been considered in earlier works.

In the streaming setup, the symmetric rate i.e., $R_1=R_2=R$ is of interest. Setting this in \eqref{eq:SumRateLower} we obtain: 
\begin{equation}
R \ge \frac{1}{2}H(\rvs_1|\rvs_0,\rvs_2) + \frac{1}{2}H(\rvs_3|\rvs_0).\label{eq:R-MT-LB}
\end{equation}

It can be easily shown that the expression in \eqref{eq:R-MT-LB} and the right hand side of the general lower bound in \eqref{eq:genLB} for $B=W=1$ are the equivalent using a simple calculation.
\begin{align}
&R^-(B=1,W=1) \notag\\
&= H(\rvs_1|\rvs_0)+ \frac{1}{2}I(\rvs_1;\rvs_{3}|\rvs_0)\notag\\
&= H(\rvs_1|\rvs_0) + \frac{1}{2}H(\rvs_3|\rvs_0) -\frac{1}{2}H(\rvs_3 | \rvs_0, \rvs_1)\notag\\
&=\frac{1}{2}H(\rvs_1, \rvs_2 | \rvs_0) + \frac{1}{2}H(\rvs_3 | \rvs_0) - \frac{1}{2}H(\rvs_3 | \rvs_1)\label{eq:MC_Prop1}\\
& = \frac{1}{2}H(\rvs_2 | \rvs_0) + \frac{1}{2}H(\rvs_1|\rvs_0,\rvs_2) + \frac{1}{2}H(\rvs_3 | \rvs_0) - \frac{1}{2}H(\rvs_3 | \rvs_1)\label{eq:Aux2}\\
&= \frac{1}{2}H(\rvs_1|\rvs_0,\rvs_2) + \frac{1}{2}H(\rvs_3|\rvs_0)\label{eq:MC_Prop2}
\end{align}
where  the first term in \eqref{eq:MC_Prop1} follows from the Markov Chain property $\rvs_0 \rightarrow\rvs_1 \rightarrow \rvs_2$, the last term in \eqref{eq:MC_Prop1} follows from the Markov Chain property $\rvs_1 \rightarrow \rvs_2 \rightarrow \rvs_3$ and \eqref{eq:MC_Prop2} follows from the fact that the source model is stationary, thus the first and last term in \eqref{eq:Aux2} are the same.

As noted before the above proof  does not directly apply to the streaming setup as it does not take into account that the decoders have access to all the past encoder outputs, and that the encoder has access to all the past source sequences. We next provide a formal proof of the lower bound that shows that this additional information does not help.

\subsection{Lower Bound on Lossless Rate-Recovery Function}
\label{subsec:genLB}

For any sequence of $(n,2^{nR})$ codes we show that there is a sequence $\eps_n$ that vanishes as $n \rightarrow \infty$ such that
\begin{align}
R \ge H(\rvs_1|\rvs_0) + \frac{1}{W+1} I(\rvs_{B+W+1}; \rvs_B| \rvs_0) - \eps_n \label{eq:rate_UB}.
\end{align}

 We consider that a burst erasure of length $B$ spans the interval $[t-B-W, t-W-1]$ for some $t\ge B+W$. It suffices to lower bound the rate  for this erasure pattern. By considering the interval $[t-W, t]$, following the burst erasure we have the following.
\begin{align}
(W+1)nR &\ge H([\rvf]_{t-W}^{t})\notag\\
&\ge H([\rvf]_{t-W}^{t} | [\rvf]_{0}^{t-B-W-1}, \rvs^n_{-1}) \label{eq:GC1}
\end{align} where \eqref{eq:GC1} follows from the fact that conditioning reduces the entropy. 
By definition, the source sequence $\rvs_{t}^n$ must be recovered from $\{[\rvf]_{0}^{t-B-W-1}, [\rvf]_{t-W}^t, \rvs^n_{-1}\}$ Applying  Fano's inequality we have that 
\begin{align}
H(\rvs_{t}^n| [\rvf]_0^{t-B-W-1}, [\rvf]_{t-W}^t, \rvs_{-1}^n) \le n\eps_n. \label{eq:Fano2}
\end{align} Therefore we have 
\begin{align}
&H([\rvf]_{t-W}^{t}~|~[\rvf]_{0}^{t-B-W-1},\rvs_{-1}^n) \notag\\ &= H(\rvs^n_{t}, [\rvf]_{t-W}^{t}~|~[\rvf]_{0}^{t-B-W-1},\rvs_{-1}^n) \notag\\ &\quad- H(\rvs^n_{t} | [\rvf]_{0}^{t-B-W-1}, [\rvf]_{t-W}^{t},\rvs_{-1}^n)\label{eq:Aux3}\\
&\ge H(\rvs_{t}^n~|~[\rvf]_0^{t-B-W-1},\rvs_{-1}^n) \notag\\ &\quad+  H([\rvf]_{t-W}^{t}~|~\rvs_{t}^n, [\rvf]_{0}^{t-B-W-1},\rvs_{-1}^n) - n\eps_n.\label{eq:FanoApp}
\end{align} where \eqref{eq:Aux3} and the first two terms of \eqref{eq:FanoApp} follow from the application of chain rule and the last term in \eqref{eq:FanoApp} follows form~\eqref{eq:Fano2}.
Now we bound each of the two terms in~\eqref{eq:FanoApp}. First we note that:
\begin{align}
&H(\rvs_{t}^n|[\rvf]_0^{t-B-W-1},\rvs_{-1}^n) \notag\\
&\qquad \ge H(\rvs_{t}^n|[\rvf]_0^{t-B-W-1}, \rvs_{t-B-W-1}^n,\rvs_{-1}^n)\label{eq:Aux4}\\
&\qquad= H(\rvs_{t}^n|\rvs_{t-B-W-1}^n) \label{eq:Aux5}\\
&\qquad= H(\rvs_{B+W+1}^n|\rvs_{0}^n)\label{eq:Aux6} \\
&\qquad= nH(\rvs_{B+W+1}|\rvs_0),\label{eq:Markov2}
\end{align} where \eqref{eq:Aux4} follows from the fact that conditioning reduces entropy and \eqref{eq:Aux5} follows from the Markov relation $$(\rvs_{-1}^n,[\rvf]_0^{t-B-W-1}) \rightarrow \rvs_{t-B-W-1}^n \rightarrow \rvs_{t}^n.$$ Eq.~\eqref{eq:Aux6} and \eqref{eq:Markov2} follow from the stationary and memoryless  source model.

Furthermore the second term in~\eqref{eq:FanoApp} can be lower bounded using the following series of inequalities.
\begin{align}
&H\left([\rvf]_{t-W}^{t}~|~\rvs_t^n,[\rvf]_{0}^{t-B-W-1},\rvs_{-1}^n\right) \notag\\
&\ge H\left([\rvf]_{t-W}^{t-1}~\big|~\rvs_{t}^n,[\rvf]_{0}^{t-W-1},\rvs_{-1}^n\right) \label{eq:cond_f3k}\\
&= H\left([\rvf]_{t-W}^{t-1}, \rvs^n_{t-W}, \ldots,\rvs_{t-1}^n |\rvs_{t}^n,[\rvf]_{0}^{t-W-1},\rvs_{-1}^n\right) \notag\\
&\qquad - H\left(\rvs_{t-W}^n, \ldots,\rvs_{t-1}^n\big|\rvs_{t}^n,[\rvf]_{0}^{t-1},\rvs_{-1}^n\right)\notag\\
&\ge H\left([\rvf]_{t-W}^{t-1}, \rvs_{t-W}^n,  \ldots,\rvs_{t-1}^n |\rvs_{t}^n,[\rvf]_{0}^{t-W-1},\rvs_{-1}^n\right) - Wn\eps_n \label{eq:FanoApp2}\\
&\ge H\left(\rvs_{t-W}^n, \ldots,\rvs_{t-1}^n\big|\rvs_{t}^n,[\rvf]_{0}^{t-W-1},\rvs_{-1}^n\right)- Wn\eps_n \notag\\
&\ge H\left(\rvs_{t-W}^n, \ldots,\rvs_{t-1}^n\big|\rvs_{t}^n,[\rvf]_{0}^{t-W-1}, \rvs_{t-W-1}^n,\rvs_{-1}^n\right) - Wn\eps_n \notag\\
& = H\left(\rvs_{t-W}^n, \rvs_{t-W+1}^n, \ldots,\rvs_{t-1}^n\big|\rvs_{t}^n, \rvs_{t-W-1}^n \right) - Wn\eps_n\label{eq:CondMarkov}\\
&= n H(\rvs_{B+1},\rvs_{B+2},\ldots,\rvs_{B+W}|\rvs_B,\rvs_{B+W+1})- Wn\eps_n\label{eq:Aux7}\\
&= n H(\rvs_{B+1}, \rvs_{B+2}, \ldots, \rvs_{B+W},\rvs_{B+W+1} |\rvs_B) \notag\\ &\quad - n H(\rvs_{B+W+1} |\rvs_B) - Wn\eps_n\notag\\
&= n (W+1)H(\rvs_1 |\rvs_0)- n H(\rvs_{B+W+1} |\rvs_B) - Wn\eps_n \label{eq:term2}
\end{align}
Note that in \eqref{eq:cond_f3k}, in order to lower bound the entropy term, we reveal the codewords $[\rvf]_{t-B-W}^{t-W-1}$ which is not originally available at the decoder and exploit the fact that conditioning reduces the entropy. This step in deriving the lower bound may not be necessarily tight, however it is the best lower bound we have for the general problem. Also~\eqref{eq:FanoApp2} follows from the fact that according to the problem setup $\{\rvs_{t-W}^n, \ldots,\rvs_{t-1}^n\}$ must be decoded when $\rvs_{-1}^n$ and all the channel codewords before time $t$, i.e. $[\rvf]_0^{t-1}$,  are available at the decoder, Hence Fano's inequality again applies. The expression above~\eqref{eq:CondMarkov} also follows from conditioning reduces entropy.
Eq.~\eqref{eq:CondMarkov}  follows from the fact that 
\begin{align}
(\rvs_{-1}^n,[\rvf]_0^{t-W-1})\rightarrow \rvs_{t-W-1}^n\rightarrow (\rvs_{t-W}^n, \ldots,\rvs_{t-1}^n).
\end{align} 
Eq.~\eqref{eq:Aux7} and \eqref{eq:term2} follow from memoryless and stationarity of the source sequences.
Combining~\eqref{eq:FanoApp},~\eqref{eq:Markov2} and~\eqref{eq:term2} we have that
\begin{multline}
H\left([\rvf]_{t-W}^{t}~\big|~[\rvf]_{0}^{t-B-W-1},\rvs_{-1}^n\right) 
 \ge n H(\rvs_{B+W+1} |\rvs_0) +  \\ n (W+1)H(\rvs_1 |\rvs_0)- n H(\rvs_{B+W+1} |\rvs_B)- (W+1)n\eps_n
\label{eq:LB_term}
\end{multline} Finally from \eqref{eq:LB_term} and~\eqref{eq:GC1} we have that,
\begin{align}
&nR \ge n H(\rvs_1 |\rvs_0) + \notag\\ 
&\quad\frac{n}{W+1} \left[H(\rvs_{B+W+1} |\rvs_0) -  H(\rvs_{B+W+1} |\rvs_B)- (W+1)\eps_n\right]\notag\\
&= n H(\rvs_1 |\rvs_0) + \notag\\ 
& \frac{n}{W+1} \left[H(\rvs_{B+W+1} |\rvs_0) -  H(\rvs_{B+W+1} |\rvs_B, \rvs_0)- (W+1)\eps_n\right]\notag\\
& =  n H(\rvs_1 |\rvs_0) + \frac{n}{W+1} I(\rvs_{B+W+1} ; \rvs_{B} |\rvs_0)- n\eps_n  \label{eq:LB_bound}
\end{align} 
where the second step above follows from the Markov condition $\rvs_0 \rightarrow \rvs_B \rightarrow \rvs_{B+W+1}$.
As we take $n\rightarrow \infty$ we recover~\eqref{eq:rate_UB}.
This completes the proof of the lower bound in Theorem~\ref{thm:genUB_LB}.

We remark that the derived lower bound holds for any $t\ge B+W$. Therefore, the lower bound~\eqref{eq:rate_UB} on lossless rate-recovery function also holds for finite-horizon rate-recovery function whenever $\mathcal{L} \ge B+W$.

Finally we note that in our setup we are assuming a peak rate constraint on $\rvf_t$. If we assume the average rate constraint across $\rvf_t$ the lower bound still applies with minor modifications in the proof. 

\subsection{Upper Bound on Lossless Rate-Recovery Function}
\label{sec:UBLB}
In this section we establish the achievability of $R^+(B,W)$ in Theorem~\ref{thm:genUB_LB} using a  binning based scheme.  At each time the encoding function $\rvf_i$ in~\eqref{eq:f-enc}  is  the bin-index of  a Slepian-Wolf codebook~\cite{Cover95, slepianWolf:73}. Following a burst  erasure in $[j,j+B-1]$, the decoder collects $\rvf_{j+B},\ldots, \rvf_{j+W+B}$ and attempts to jointly recover all the underlying sources  at $t=j+W+B$.  Using 
 Corollary~\ref{corol:genUB} it suffices to show that
\begin{equation}
R^+  = \frac{1}{W+1} H(\rvs_{B+1}, \ldots, \rvs_{B+W+1}|\rvs_0) + \eps\label{eq:SW2} 
\end{equation}
is achievable for any arbitrary $\eps > 0$.

We use a  codebook $\cC$ which is generated by randomly partitioning the set of all typical sequences
~$T_\eps^n(\rvs)$ into $2^{nR^+}$ bins. 
The partitions are revealed to the  decoder ahead of time. 

Upon observing $\rvs_i^n$ the encoder declares an error if $\rvs_i^n \notin T_\eps^n(\rvs)$. Otherwise it finds the bin to which $\rvs_i^n$ belongs to and sends the corresponding bin index $\rvf_i$. We separately consider two possible scenarios at the decoder.

First suppose that the sequence $\rvs_{i-1}^n$ has already been recovered. Then the destination attempts to recover $\rvs_i^n$ from  $(\rvf_i, \rvs_{i-1}^n)$. This succeeds with high probability if $R^+ > H(\rvs_1 |\rvs_0)$, which is guaranteed via~\eqref{eq:SW2}. If we define probability of the error event $\mathcal{E}_{i}\triangleq \{\hat{\rvs}^n_{i}\neq \rvs^n_{i}\}$ conditioned on the correct recovery of $\rvs^n_{i-1}$, i.e. $ \overline{\mathcal{E}}_{i-1}$, as follows
\begin{align}
P^{(n)}_{e,1} \triangleq P(\mathcal{E}_{i} | \overline{\mathcal{E}}_{i-1}) \label{eq:eps1}
\end{align} then for the rates satisfying $R^+ > H(\rvs_1 |\rvs_0)$ and in particular for $R^{+}$ in \eqref{eq:SW2}, it is guaranteed that
\begin{align}
\lim_{n\to \infty} P^{(n)}_{e,1} = 0\label{eq:lim1}.
\end{align} 

Next consider the case where 
$\rvs_i^n$ is the first sequence to be recovered after the burst erasure. In particular the burst erasure spans the interval ${[ i-B'-W,i-W-1]}$ for some $B' \le B$. The decoder thus has access to $\rvs_{i-B'-W-1}^n$, before the start of the burst erasure. Upon receiving $\rvf_{i-W},\ldots, \rvf_i$ the destination simultaneously attempts to recover $(\rvs_{i-W}^n,\ldots, \rvs_i^n)$ given $(\rvs_{i-B'-W-1}^n, \rvf_{i-W},\ldots, \rvf_i)$. This succeeds with high probability if,
\begin{align}
(W+1)nR &= \sum_{j=i-W}^i H(\rvf_j) \\
&> nH(\rvs_{i-W},\ldots, \rvs_i|\rvs_{i-B'-W-1})\\
&=nH(\rvs_{B'+1},\ldots, \rvs_{B'+W+1}|\rvs_0)\label{eq:s-tinv}
\end{align}where~\eqref{eq:s-tinv} follows from the fact that the sequence of variables $\rvs_i$ is a stationary process. Whenever $B' \le B$ it immediately follows that~\eqref{eq:s-tinv}  is also guaranteed by~\eqref{eq:SW2}. Define $P^{(n)}_{e,2}$ as the probability of error in $\rvs^n_{i}$  given $(\rvs_{i-B-W-1}^n, \rvf_{i-W},\ldots, \rvf_i)$, i.e. 
\begin{align}
P_{e,2}^{(n)}\triangleq P(\mathcal{E}_{i} | \overline{\mathcal{E}}_{i-B-W-1}) \label{eq:eps2}.
\end{align}  For rate satisfying \eqref{eq:s-tinv}, which is satisfied through \eqref{eq:SW2}, it is guaranteed that  
\begin{align}
\lim_{n\to \infty} P^{(n)}_{e,2} = 0\label{eq:lim2}.
\end{align}

{\bf Analysis of the Streaming Decoder:} As described in problem setup, the decoder is interested in recovering all the source sequences outside the error propagation window with vanishing probability of error. Assume a communication duration of $\mathcal{L}$ and a single burst erasure of length $0<B'\le B$ spanning the interval $[j, j+B'-1]$, for $0\le j\le \mathcal{L}$. The decoder fails if at least one source sequences outside the error propagation window is erroneously recovered, i.e. $\hat{\rvs}^n_{i} \neq \rvs^n_{i}$ for some $i\in [0,j-1]\cup[j+B'+W+1,\mathcal{L}]$. For this particular channel erasure pattern, the probability of decoder's failure, denoted by $P^{(n)}_{\textrm{F}}$, can be bounded  as follows.  
\begin{align}
P^{(n)}_{\textrm{F}} &\le  \sum_{k=0}^{j-1}P(\mathcal{E}_{k} | \overline{\mathcal{E}}_{0}, \overline{\mathcal{E}}_{1}, \ldots ,\overline{\mathcal{E}}_{k-1}) +\notag\\ & P(\mathcal{E}_{j+B'+W+1} | \overline{\mathcal{E}}_{0}, \ldots,\overline{\mathcal{E}}_{j-1}) +\notag\\
&  \sum_{k=j+B'+W+2}^{\mathcal{L}}\!\!\!\! P(\mathcal{E}_{k}| \overline{\mathcal{E}}_{0}, \ldots, \overline{\mathcal{E}}_{j-1}, \overline{\mathcal{E}}_{j+B'+W+1}, \ldots,\overline{\mathcal{E}}_{k-1} ) \label{eq:out1}\\
& = (\mathcal{L}-B'-W) P^{(n)}_{e,1} + P^{(n)}_{e,2}\le \mathcal{L}  P^{(n)}_{e,1} + P^{(n)}_{e,2}\label{eq:out2}
\end{align} where  $P^{(n)}_{e,1}$ and $P^{(n)}_{e,2}$ are defined in \eqref{eq:eps1} and \eqref{eq:eps2}. Eq.~\eqref{eq:out2} follows from the fact that, because of the Markov property of the source model, all the terms in the first and the last summation in \eqref{eq:out1} are the same and equal to $P_{e,1}^{(n)}$. 

According to \eqref{eq:lim1} and \eqref{eq:lim2}, for any rate satisfying \eqref{eq:SW2} and for any $\mathcal{L}$, $n$ can be chosen large enough such that the upper bound on $P^{(n)}_{\textrm{F}}$ in \eqref{eq:out2} approaches zero. Thus the decoder fails with vanishing probability for any fixed $\mathcal{L}$. This in turn establishes the upper bound on $R(B,W)$, when ${\mathcal L} \rightarrow\infty$.
This completes the justification of the upper bound.

%
%
%
%
%

\section{Symmetric Sources: Proof of Corollary~\ref{thm:binning}}
\label{sec:Symmetric}
\begin{figure*}
\begin{center}
\begin{tikzpicture}

\def \a {0}{
\draw [white] (-1.5,-\a-.4) -- (-1.5,-\a-.4) node {$\color{black}\rvs_{j}^{n}$};
\fill [fill=white, draw=gray!50!black] (0,-.8-\a) -- (0,-\a) -- (1,-\a) -- (1,-.8-\a) -- (0,-.8-\a);
\draw [white] (.5,-\a-0.25) -- (.5,-\a-0.25) node {$\color{black}\scriptstyle\textrm{Encoder}$};
\draw [white] (.5,-\a-0.55) -- (.5,-\a-0.55) node {$\color{black}\scriptstyle{j}$};
\draw [->] (-1.3,-.4) -- (0,-.4);
\draw [->] (1,-.4) -- (3,-.4);
\fill [fill=white, draw=gray!50!black] (3,-.8-\a) -- (3,-\a) -- (4,-\a) -- (4,-.8-\a) -- (3,-.8-\a);
\draw [white] (3.5,-\a-0.25) -- (3.5,-\a-0.25) node {$\color{black}\scriptstyle\textrm{Decoder}$};
\draw [white] (3.5,-\a-0.55) -- (3.5,-\a-0.55) node {$\color{black}\scriptstyle{j}$};
\draw [->] (4,-.4) -- (5,-.4);
\draw [white] (5.5,-\a-0.4) -- (5.5,-\a-0.4) node {$\color{black}\hat{\rvs}^n_{j}$};
\draw [white] (2,-\a-0.15) -- (2,-\a-0.15) node {$\color{black}\rvf_{j}$};
\draw [->] (3.5,-\a+1) -- (3.5,-\a);
\draw [white] (4,-\a+.5) -- (4,-\a+.5) node {$\color{black}\rvs^n_{j-1}$};
}
\def \a {1.5}{
\draw [white] (-1.5,-\a-.4) -- (-1.5,-\a-.4) node {$\color{black}\rvs_{j+1}^{n}$};
\fill [fill=white, draw=gray!50!black] (0,-.8-\a) -- (0,-\a) -- (1,-\a) -- (1,-.8-\a) -- (0,-.8-\a);
\draw [white] (.5,-\a-0.25) -- (.5,-\a-0.25) node {$\color{black}\scriptstyle\textrm{Encoder}$};
\draw [white] (.5,-\a-0.55) -- (.5,-\a-0.55) node {$\color{black}\scriptstyle{j+1}$};
\draw [->] (-1,-.4-\a) -- (0,-.4-\a);
\draw [->] (1,-.4-\a) -- (3,-.4-\a);
\fill [fill=white, draw=gray!50!black] (3,-.8-\a) -- (3,-\a) -- (4,-\a) -- (4,-.8-\a) -- (3,-.8-\a);
\draw [white] (3.5,-\a-0.25) -- (3.5,-\a-0.25) node {$\color{black}\scriptstyle\textrm{Decoder}$};
\draw [white] (3.5,-\a-0.55) -- (3.5,-\a-0.55) node {$\color{black}\scriptstyle{j+1}$};
\draw [->] (4,-.4-\a) -- (5,-.4-\a);
\draw [white] (5.5,-\a-0.4) -- (5.5,-\a-0.4) node {$\color{black}\hat{\rvs}^n_{j+1}$};
\draw [white] (2,-\a-0.15) -- (2,-\a-0.15) node {$\color{black}\rvf_{j+1}$};
\draw [->] (3.5,-\a-1.8) -- (3.5,-\a-.8);
\draw [white] (4.2,-\a-1.3) -- (4.2,-\a-1.3) node {$\color{black}\rvs^n_{j-B-1}$};
}

\draw [-] (2.5,-.4) -- (2.5,-1.7);
\draw [->] (2.5,-1.7) -- (3,-1.7);
\draw [white] (2,-4) -- (2,-4) node {$\color{black}\textrm{(a)}$};

\def \a {0}{
\draw [white] (6.5,-\a-.4) -- (6.5,-\a-.4) node {$\color{black}\rvs_{j}^{n}$};
\fill [fill=white, draw=gray!50!black] (8,-.8-\a) -- (8,-\a) -- (9,-\a) -- (9,-.8-\a) -- (8,-.8-\a);
\draw [white] (8.5,-\a-0.25) -- (8.5,-\a-0.25) node {$\color{black}\scriptstyle\textrm{Encoder}$};
\draw [white] (8.5,-\a-0.55) -- (8.5,-\a-0.55) node {$\color{black}\scriptstyle{j}$};
\draw [->] (6.7,-.4) -- (8,-.4);
\draw [->] (9,-.4) -- (11,-.4);
\fill [fill=white, draw=gray!50!black] (11,-.8-\a) -- (11,-\a) -- (12,-\a) -- (12,-.8-\a) -- (11,-.8-\a);
\draw [white] (11.5,-\a-0.25) -- (11.5,-\a-0.25) node {$\color{black}\scriptstyle\textrm{Decoder}$};
\draw [white] (11.5,-\a-0.55) -- (11.5,-\a-0.55) node {$\color{black}\scriptstyle{j}$};
\draw [->] (12,-.4) -- (13,-.4);
\draw [white] (13.5,-\a-0.4) -- (13.5,-\a-0.4) node {$\color{black}\hat{\rvs}^n_{j}$};
\draw [white] (10,-\a-0.15) -- (10,-\a-0.15) node {$\color{black}\rvf_{j}$};
\draw [->] (11.5,-\a+1) -- (11.5,-\a);
\draw [white] (12,-\a+.5) -- (12,-\a+.5) node {$\color{black}\rvs^n_{j+1}$};
}
\def \a {1.5}{
\draw [white] (6.5,-\a-.4) -- (6.5,-\a-.4) node {$\color{black}\rvs_{j+1}^{n}$};
\fill [fill=white, draw=gray!50!black] (8,-.8-\a) -- (8,-\a) -- (9,-\a) -- (9,-.8-\a) -- (8,-.8-\a);
\draw [white] (8.5,-\a-0.25) -- (8.5,-\a-0.25) node {$\color{black}\scriptstyle\textrm{Encoder}$};
\draw [white] (8.5,-\a-0.55) -- (8.5,-\a-0.55) node {$\color{black}\scriptstyle{j+1}$};
\draw [->] (7,-.4-\a) -- (8,-.4-\a);
\draw [->] (9,-.4-\a) -- (11,-.4-\a);
\fill [fill=white, draw=gray!50!black] (11,-.8-\a) -- (11,-\a) -- (12,-\a) -- (12,-.8-\a) -- (11,-.8-\a);
\draw [white] (11.5,-\a-0.25) -- (11.5,-\a-0.25) node {$\color{black}\scriptstyle\textrm{Decoder}$};
\draw [white] (11.5,-\a-0.55) -- (11.5,-\a-0.55) node {$\color{black}\scriptstyle{j+1}$};
\draw [->] (12,-.4-\a) -- (13,-.4-\a);
\draw [white] (13.5,-\a-0.4) -- (13.5,-\a-0.4) node {$\color{black}\hat{\rvs}^n_{j+1}$};
\draw [white] (10,-\a-0.15) -- (10,-\a-0.15) node {$\color{black}\rvf_{j+1}$};
\draw [->] (11.5,-\a-1.8) -- (11.5,-\a-.8);
\draw [white] (12.2,-\a-1.3) -- (12.2,-\a-1.3) node {$\color{black}\rvs^n_{j-B-1}$};
}

\draw [-] (10.5,-.4) -- (10.5,-1.7);
\draw [->] (10.5,-1.7) -- (11,-1.7);
\draw [white] (10,-4) -- (10,-4) node {$\color{black}\textrm{(b)}$};

\draw [dotted] (6,1.5) -- (6,-4);

\end{tikzpicture}
\end{center}
\caption{Connection between the streaming problem and the multi-terminal source coding problem. 
The setup on the right is identical to the setup on the left, except with the side information sequence $\rvs_{j-1}^n$ replaced with $\rvs_{j+1}^n$. However the rate
region for both problems are identical for symmetric Markov sources.}
\label{fig:zigZag}
\end{figure*}
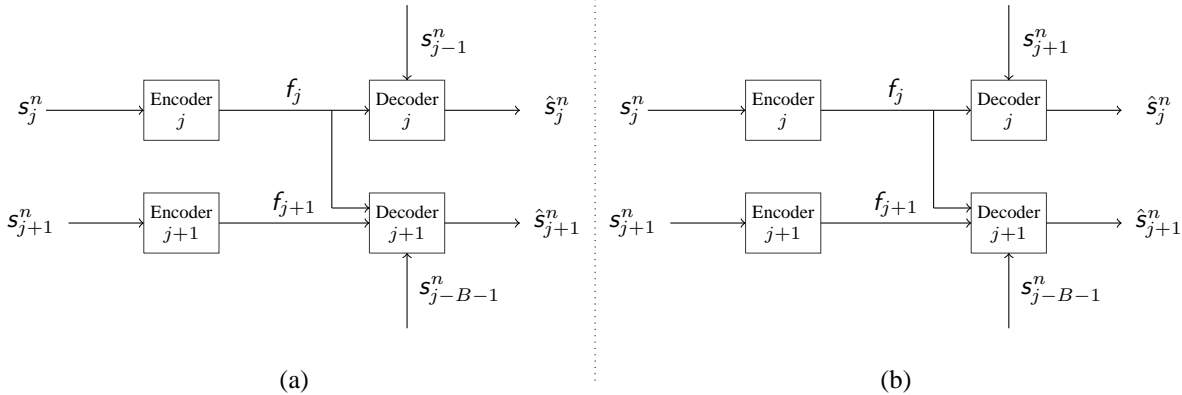

In this section we establish that the lossless rate-recovery function for symmetric Markov sources restricted to class of memoryless encoders is given by
\begin{align}
R(B,W) = \frac{1}{W+1}H(\rvs_{B+1}, \ldots, \rvs_{B+W+1} | \rvs_{0}).
\end{align} 

The achievability follows from Theorem~\ref{thm:genUB_LB} and Corollary~\ref{corol:genUB}. We thus only need to prove the converse to improve upon the general lower bound in~\eqref{eq:genLB}. 
The lower bound for the special case when $W=0$ follows directly from~\eqref{eq:genLB} and thus we only need to show the lower bound for $W \ge 1$. 
For simplicity in exposition we illustrate the case when $W=1$. Then we need to show that 
\begin{align}
R(B,W=1) \ge \frac{1}{2} H(\rvs_{B+1}, \rvs_{B+2}|\rvs_0)\label{eq:rate-rec-binning-2}
\end{align}
The proof for general $W > 1$ will follow along similar lines and will be sketched thereafter.

Assume that a burst erasure spans time indices ${j-B, \ldots, j-1}$. The decoder must recover 
\begin{equation}
\label{eq:decoder1}
{\hat\rvs}_{j+1}^n = \cG_{j+1}\left([\rvf]_0^{j-B-1}, \rvf_{j}, \rvf_{j+1}, \rvs_{-1}^n\right).
\end{equation} 

Furthermore if there is no erasure until time $j$ then \begin{align}\hat{\rvs}_j^n = \cG_j\left([\rvf]_0^{j}, \rvs_{-1}^n\right)\label{eq:decoder0}\end{align} must hold. 
Our aim is to 
establish the following lower bound on the sum-rate.
\begin{equation}
\label{eq:Rj_lb}
2R \ge H(\rvs_{j+1}|\rvs_j) + H(\rvs_{j}|\rvs_{j-B-1}).
\end{equation}
The lower bound~\eqref{eq:rate-rec-binning-2} then  follows since \begin{align}
R&\ge \frac{1}{2}(H(\rvs_{j+1}|\rvs_j) + H(\rvs_{j}|\rvs_{j-B-1}))\notag\\
&=\frac{1}{2}(H(\rvs_{j+1}|\rvs_j,\rvs_{j-B-1}) + H(\rvs_{j}|\rvs_{j-B-1}))\label{eq:Q1}\\
&=\frac{1}{2}H(\rvs_{j+1},\rvs_{j}|\rvs_{j-B-1}) =\frac{1}{2}H(\rvs_{B+1},\rvs_{B+2}|\rvs_0),\label{eq:inter}
\end{align}
where \eqref{eq:Q1} follows from the Markov chain property $\rvs_{j-B-1} \rightarrow \rvs_{j} \rightarrow \rvs_{j+1}$, and the last step in \eqref{eq:inter} follows from stationarity of the source model.

To establish~\eqref{eq:Rj_lb} we make a connection to a multi-terminal source coding problem in Fig.~\ref{fig:zigZag}(a). We accomplish this in several steps as outlined below.

\subsection{Multi-Terminal Source Coding}
Consider the multi-terminal source coding problem with side information illustrated in Fig.~\ref{fig:zigZag}(a). 
In this setup there are four source sequences drawn i.i.d.\ from a joint distribution $p(\rvs_{j+1},\rvs_j, \rvs_{j-1},\rvs_{j-B-1})$. The two source sequences $\rvs_j^n$ and $\rvs_{j+1}^n$ are revealed to the encoders $j$  and $j+1$ respectively and the two sources
$\rvs_{j-1}^n$ and $\rvs_{j-B-1}^n$ are revealed to the decoders $j$ and $j+1$ respectively. The encoders operate independently and compress the source sequences to $\rvf_j$ and $\rvf_{j+1}$ at rates $R_j$
and $R_{j+1}$ respectively. Decoder $j$ has access to $(\rvf_j, \rvs_{j-1}^n)$ while decoder ${j+1}$ has access to $(\rvf_j, \rvf_{j+1}, \rvs_{j-B-1}^n)$. The two decoders are required to reconstruct
\begin{align}
\tilde{\rvs}_j^n &= \tilde{\cG}_j(\rvf_j, \rvs_{j-1}^n) \label{eq:decoder0a}\\
\tilde{\rvs}_{j+1}^n&=\tilde{\cG}_{j+1}(\rvf_j, \rvf_{j+1}, \rvs_{j-B-1}^n)\label{eq:decoder1a}\end{align} respectively such that $\Pr(\rvs_i^n \neq \tilde{\rvs}_i^n) \le \eps_n$ for ${i=\{ j, j+1\}}$.

Note that the multi-terminal source coding setup in Fig.~\ref{fig:zigZag}(a) is similar to the setup in Fig.~\ref{fig:multiterminal}, except that the encoders do not cooperate and $\rvf_{i} =  \cF_{i}(\rvs^n_{i})$, due to the memoryless property. We exploit this property to directly show that a lower bound on the multi-terminal source coding setup in Fig.~\ref{fig:zigZag}(a) also constitutes a lower bound on the rate of the original streaming problem.

{
\begin{lemma}
\label{lem:memoryless}
For the class of memoryless encoding functions, i.e. $\rvf_j = \cF_j(\rvs_j^n)$, 
the decoding functions 
$\hat{\rvs}_j^n = {\cG}_{j}([\rvf]_0^j,\rvs_{-1}^n)$ and $\hat{\rvs}_{j+1}^n = {\cG}_{j+1}([\rvf]_0^{j-2}, \rvf_j, \rvf_{j+1},\rvs_{-1}^n)$
can be replaced by the following decoding functions
\begin{align}
&\tilde{\rvs}_j^n = \tilde{\cG}_{j}(\rvf_j,\rvs_{j-1}^n) \label{eq:Fano2a}\\
&\tilde{\rvs}_{j+1}^n = \tilde{\cG}_{j+1}( \rvf_j ,\rvf_{j+1}, \rvs^n_{j-2},) \label{eq:Fano1a}
\end{align} such that
\begin{align}
&\Pr(\tilde{\rvs}_j^n \neq \rvs_j^n) \le \Pr(\hat{\rvs}_j^n \neq \rvs_j^n)\\
&\Pr(\tilde{\rvs}_{j+1}^n \neq \rvs_{j+1}^n) \le \Pr(\hat{\rvs}_{j+1}^n \neq \rvs_{j+1}^n)  \label{eq:ErPr}.
\end{align} \hfill$\Box$
\end{lemma}

\begin{proof}
Assume that the extra side-informations $\rvs^n_{j-1}$ 
is revealed to the decoder $j$.
Now define the maximum a posteriori probability (MAP) decoder as follow. 
\begin{align}
&\tilde{\rvs}_j^n = {\overline{\cG}}_{j}([\rvf]_0^j,\rvs_{-1}^n, \rvs^n_{j-1}) \triangleq \underset{\rvs_j^n}{\operatorname{argmax}} \quad   p(\rvs_j^n | [\rvf]_0^{j},\rvs_{-1}^n, \rvs^n_{j-1}) \label{eq:MAP1}
\end{align} where we dropped the subscript in conditional probability density for sake of simplicity.
It is known that the MAP decoder is optimal and minimizes the decoding error probability, therefore
\begin{align}
&\Pr(\tilde{\rvs}_j^n \neq \rvs_j^n) \le \Pr(\hat{\rvs}_j^n \neq \rvs_j^n)  \label{eq:ineq1}
\end{align}
Also note that 
\begin{align}
\tilde{\rvs}_j^n = {\overline{\cG}}_{j}([\rvf]_0^j,\rvs_{-1}^n, \rvs_{j-1}^n)& = \underset{\rvs_j^n}{\operatorname{argmax}} \quad   p(\rvs_j^n | [\rvf]_0^{j},\rvs_{-1}^n,\rvs_{j-1}^n) \label{eq:MAPG1O}\\
& =\underset{\rvs_j^n}{\operatorname{argmax}} \quad   p(\rvs_j^n | \rvf_{j},\rvs_{j-1}^n) \label{eq:MAPGG1}\\
&\defeq  \tilde{\cG}_{j}(\rvf_{j},\rvs_{j-1}^n) \label{eq:De1}
\end{align} where \eqref{eq:MAPGG1} follows form the following Markov property. 
\begin{align}
\left([\rvf]_0^{j-1}, \rvs_{-1}^n\right) \rightarrow (\rvf_{j}, \rvs_{j-1}^n) \rightarrow \rvs_{j}^n.
\end{align} 
It can be shown through similar steps that the decoder defined in \eqref{eq:Fano1a} exists with the error probability satisfying \eqref{eq:ErPr}.  
This completes the proof.
\end{proof}
}

The conditions in~\eqref{eq:Fano2a} and~\eqref{eq:Fano1a} show that any rate that is achievable in the streaming problem in Fig.~\ref{fig:setup} is also achieved in the multi-terminal source coding setup in Fig.~\ref{fig:zigZag}(a). Hence a lower bound to this source network also constitutes a lower bound to the original problem. In the next section we find a lower bound on the rate for the setup in Fig.~\ref{fig:zigZag}(a).

\subsection{Lower Bound for Multi-terminal Source Coding Problem}
In this section, we establish a lower bound on the sum-rate of the multi-terminal source coding setup in Fig.~\ref{fig:zigZag}(a) i.e., $R \ge \frac{1}{2}H(\rvs_{B+1}, \rvs_{B+2}|\rvs_0)$. To this end, we observe the equivalence between the setup in 
Fig.~\ref{fig:zigZag}(a) and Fig.~\ref{fig:zigZag}(b) as stated below.

\begin{lemma}
The set of all achievable rate-pairs $(R_j, R_{j+1})$ for the problem in Fig.~\ref{fig:zigZag}(a) is identical to the set of all achievable rate-pairs for the problem in Fig.~\ref{fig:zigZag}(b)
where the side information sequence $\rvs_{j-1}^n$ at decoder 1 is replaced by the side information sequence $\rvs_{j+1}^n$.
\label{lem:SI}
\end{lemma}

The proof of Lemma~\ref{lem:SI} follows by observing that the capacity region for the problem in Fig.~\ref{fig:zigZag}(a) depends on the joint distribution $p(\rvs_j, \rvs_{j+1}, \rvs_{j-1},\rvs_{j-B-1})$ only via the {\em marginal}
distributions $p(\rvs_j, \rvs_{j-1})$ and $p(\rvs_{j+1}, \rvs_j, \rvs_{j-B-1})$.  Indeed the decoding error at decoder $j$ depends on the former whereas the decoding error at decoder ${j+1}$ depends on the latter. When the source is symmetric,  the joint distributions $p(\rvs_j, \rvs_{j-1})$ and $p(\rvs_j, \rvs_{j+1})$ are identical and thus exchanging $\rvs_{j-1}^n$ with $\rvs_{j+1}^n$ does not change the error probability at decoder $j$ and leaves the functions at all other terminals unchanged. The formal proof is straightforward and will be omitted.

Thus it suffices to lower bound the achievable sum-rate for the problem in Fig.~\ref{fig:zigZag}(b). First note that 
\begin{align}
nR_{j+1}&= H(\rvf_{j+1})\notag\\
&\ge I(\rvf_{j+1} ;\rvs_{j+1}^n | \rvs^n_{j-B-1}, \rvf_{j})\notag\\
&=  H(\rvs_{j+1}^n|\rvs_{j-B-1}^n, \rvf_{j}) -H(\rvs_{j+1}^n|\rvs_{j-B-1}^n, \rvf_{j}, \rvf_{j+1})\notag\\
&\ge H(\rvs_{j+1}^n|\rvs_{j-B-1}^n, \rvf_{j}) - n\eps_n\label{eq:sum1}
\end{align}where \eqref{eq:sum1} follows by applying Fano's inequality for decoder $j+1$ in Fig.~\ref{fig:zigZag}(b) since $\rvs_{j+1}^n$ can be recovered from $(\rvs_{j-B-1}^n, \rvf_{j}, \rvf_{j+1})$. To bound $R_j$
\begin{align}
&nR_{j}= H(\rvf_{j})\notag\\
&\ge I(\rvf_{j};\rvs^n_{j}|\rvs_{j-B-1}^n)\nonumber \\
&= H(\rvs_{j}^n|\rvs_{j-B-1}^n)\notag - H(\rvs_{j}^n|\rvs_{j-B-1}^n, \rvf_{j})\nonumber \\
&\ge nH(\rvs_{j}|\rvs_{j-B-1})- H(\rvs_{j}^n|\rvs_{j-B-1}^n, \rvf_{j}) \notag\\ &\quad +H(\rvs_{j}^n|\rvs_{j-B-1}^n,\rvs_{j+1}^n, \rvf_{j})-n\eps_n\label{eq:FanoApp_Sj}\\
&=nH(\rvs_{j}|\rvs_{j-B-1})  - I(\rvs_{j}^n; \rvs_{j+1}^n |\rvs_{j-B-1}^n, \rvf_{j})-n\eps_n\nonumber\\
&=nH(\rvs_{j}|\rvs_{j-B-1})-H(\rvs_{j+1}^n |\rvs_{j-B-1}^n, \rvf_{j}) \notag\\ &\quad + H( \rvs_{j+1}^n |\rvs_{j-B-1}^n,\rvs_{j}^n, \rvf_{j})-n\eps_n\nonumber\\
&=nH(\rvs_{j}|\rvs_{j-B-1})-H(\rvs_{j+1}^n |\rvs_{j-B-1}^n, \rvf_{j}) \notag\\ &\quad + nH(\rvs_{j+1}|\rvs_j) -n\eps_n\label{eq:sum2}
\end{align}
where~\eqref{eq:FanoApp_Sj} follows by applying Fano's inequality for decoder $j$ in Fig.~\ref{fig:zigZag}(b) since $\rvs_j^n$ can be recovered from $(\rvs_{j+1}^n, \rvf_j)$ and hence $H(\rvs_{j}^n|\rvs_{j-B-1}^n,\rvs_{j+1}^n, \rvf_{j}) \le n\eps_n$ holds and~\eqref{eq:sum2} follows from the Markov relation ${\rvs_{j+1}^n \rightarrow \rvs_j^n \rightarrow (\rvf_j, \rvs_{j-B-1}^n)}$. By summing~\eqref{eq:sum1} and~\eqref{eq:sum2} and using $R_j=R_{j+1}=R$, we have 
\begin{align}
R_{j}+R_{j+1} &\ge H(\rvs_{j+1}|\rvs_{j}) + H(\rvs_{j}|\rvs_{j-B-1}) \\
&= H(\rvs_{j},\rvs_{j+1}|\rvs_{j-B-1})\label{eq:Q2}.
\end{align}
which is equivalent to~\eqref{eq:Rj_lb}. 
\begin{remark}
\label{rem:suc}
One way to interpret the lower bound in \eqref{eq:Q2} is by observing that the decoder $j+1$ in Fig.~\ref{fig:zigZag}(b) is able to recover not only $\rvs_{j+1}^n$ but also $\rvs_{j}^n$. In particular, the decoder $j+1$ first recovers ${\rvs}^n_{j+1}$. Then, similar to decoder $j$, it also recovers $\rvs_{j}^n$ from $\rvf_{j}$ and ${\rvs}^n_{j+1}$ as side information. Hence, by only considering decoder $j+1$ and following standard source coding argument, the lower bound on the sum-rate satisfies \eqref{eq:Q2}.
 
\end{remark}
\subsection{Extension to Arbitrary $W>1$}

To extend the result for arbitrary $W$, we use the following result which is a natural generalization of Lemma \ref{lem:memoryless}. 
{
\begin{lemma}
\label{lem:genW} 
Consider memoryless encoding functions $\rvf_k = \cF_k(\rvs_k^n)$ for $k\in\{j,\ldots,j+W\}$. Any set of decoding functions 
\begin{align}
&\hat{\rvs}_k^n = {\cG}_k([\rvf]_0^k,\rvs_{-1}^n) \qquad k\in\{j,\ldots,j+W-1\}\\
&\hat{\rvs}_{j+W}^n = {\cG}_{j+W}([\rvf]_0^{j-B-1}, [\rvf]_j^{j+W},\rvs_{-1}^n)
\end{align} 
can be replaced by a new set of decoding functions as
\begin{align}
&\tilde{\rvs}_k^n = \tilde{\cG}_k(\rvf_k,\rvs_{k-1}^n) \qquad k\in\{j,\ldots,j+W-1\}\\
&\tilde{\rvs}_{j+W}^n = \tilde{\cG}_{j+W}(\rvs^n_{j-B-1}, [\rvf]_j^{j+W})
\end{align} where 
\begin{align}
\Pr(\tilde{\rvs}_{l}^n\neq {\rvs}_{l}^n) \le \Pr(\hat{\rvs}_{l}^n\neq {\rvs}_{l}^n) \qquad j\le l\le j+W.
\end{align} \hfill$\Box$
\end{lemma}

The proof is an immediate extension of Lemma~\ref{lem:memoryless} and is excluded here.}
The lemma suggests a natural multi-terminal problem for establishing the lower bound with $W+1$ encoders and decoders. 
For concreteness we discuss the case when $W = 2$. 
Consider three encoders $t \in \{j, j+1, j+2\}$. Encoder $t$ observes  $\rvs_t^n$ and compresses it into an index $\rvf_t \in [1,2^{nR_t}]$.  
$\rvs_{t-1}^n$ for $t \in \{j, j+1\}$ are revealed to the corresponding decoders and $\rvs_{j-B-1}^n$ is revealed to the decoder $j+2$.  
Using an argument analogous to Lemma~\ref{lem:SI} the rate region is equivalent to the case when $\rvs_{j+1}^n$ and $\rvs_{j+2}^n$ are instead revealed to decoders $j$ and $j+1$ respectively. 
For this new setup we can argue that decoder $j+2$ can always reconstruct $(\rvs_j^n, \rvs_{j+1}^n, \rvs_{j+2}^n)$ given $(\rvs_{j-B-1}^n, \rvf_j, \rvf_{j+1},\rvf_{j+2})$. In particular, following the same argument in Remark \ref{rem:suc}, the decoder $j+2$ first recovers $\rvs_{j+2}^n$, then using $\{\rvf_{j+1}, \rvs^n_{j+2}\}$ recovers $\rvs_{j+1}^n$ and finally using $\{\rvf_{j}, \rvs^n_{j+1}\}$ recovers $\rvs_{j}^n$. And hence if we only consider decoder ${j+2}$ with side information $\rvs_{j-B-1}^n$ the sum-rate must satisfy:
\begin{align}
3R = R_j+R_{j+1} + R_{j+2} \ge H(\rvs_j, \rvs_{j+1},\rvs_{j+2}|\rvs_{j-B-1}) \label{eq:lower3}.
\end{align} 
Using 
Lemma~\ref{lem:genW} for $W=2$ it follows that the proposed lower bound also continues to hold for the original streaming problem. This completes the proof. The extension to any arbitrary $W$ is completely analogous.

%
%
%

\section{Lossy Rate-Recovery for Gauss-Markov Sources}
\label{sec:GM}
We establish lower and upper bounds on the lossy rate-recovery function of Gauss-Markov sources when an immediate recovery following the burst erasure is required i.e., ${W=0}$. For the single burst erasure case, the proof of  the lower bound in Prop.~\ref{prop:GML} is presented in Section~\ref{sec:LowerGM} whereas the proof of the upper bound in Prop.~\ref{prop:GMAch} is presented in Section~\ref{sec:UpperGM}. The proof of Prop.~\ref{prop:GM-ME} for the multiple burst erasures case is presented in Section~\ref{sec:GM_ME}. Finally the proof of Corollary~\ref{corol:HR}, which establishes the lossy rate-recovery function in the high resolution regime is presented in Section~\ref{sec:HR}.  

\subsection{Lower Bound: Single Burst Erasure}
\label{sec:LowerGM}
Consider any rate $R$ code that satisfies an average distortion of $D$ as stated in~\eqref{eq:D-def}.
For each $i \ge 0$ we have
\begin{align}
nR & \ge H(\rvf_{i})\notag\\
& \ge H(\rvf_{i}|[\rvf]_{0}^{i-B-1}, \rvs^n_{-1})\label{eq:g1} \\
& = I(\rvs^n_{i} ; \rvf_{i} |  [\rvf]_{0}^{i-B-1}, \rvs^n_{-1}) + H(\rvf_{i} | \rvs^n_{i} , [\rvf]_{0}^{i-B-1}, \rvs^n_{-1})
\notag\\
&\ge h(\rvs^n_{i} |  [\rvf]_{0}^{i-B-1} , \rvs^n_{-1}) - h(\rvs^n_{i} | \rvf_{i} ,  [\rvf]_{0}^{i-B-1}, \rvs^n_{-1}) \label{eq:g2}
\end{align} where \eqref{eq:g1} follows from the fact that conditioning reduces the entropy. 

We now present an upper bound for the second term and a lower bound for the first term in \eqref{eq:g2}. We first establish an upper bound for the second term in \eqref{eq:g2}. Suppose that the burst erasure occurs in the interval ${[i-B, i-1]}$. The reconstruction sequence $\hat{\rvs}_i^n$ must be a function of $(\rvf_{i} ,  [\rvf]_{0}^{i-B-1}, \rvs^n_{-1})$. Thus we have

\begin{align}
h(\rvs_{i}^n|[\rvf]_{0}^{i-B-1},\rvf_{i}, \rvs^n_{-1})& = h(\rvs_{i}^n-\hat{\rvs}_{i}^n~|~[\rvf]_{0}^{i-B-1},\rvf_{i}, \rvs^n_{-1})\notag\\
&\le h(\rvs_{i}^n-\hat{\rvs}_{i}^n) \notag\\
&\le \frac{n}{2} \log(2\pi e D) \label{eq:3},
\end{align} where the last step uses the fact that the expected average distortion between $\rvs^n_i$ and $\hat{\rvs}^n_i$ is no greater than $D$, and applies standard arguments~\cite[Ch.~13]{coverThomas}.

To lower bound the first term in \eqref{eq:g2}, we successively use the Gauss-Markov relation~\eqref{eq:GM-Def} to express:
\begin{align}
\rvs_{i} =\rho^{(B+1)} \rvs_{i-B-1} + \tilde{n} \label{eq:GM-Successive}
\end{align} for each $i \ge B$ and $\tilde{n}\sim \cN(0,1-\rho^{2(B+1)})$ is independent of $\rvs_{i-B-1}$. 
Using the Entropy Power Inequality~\cite{coverThomas} we have
\begin{multline}
2^{\frac{2}{n}h(\rvs^n_{i}|[\rvf]_{0}^{i-B-1}, \rvs^n_{-1})}\ge \\ 2^{\frac{2}{n}h(\rho^{B+1} \rvs^n_{i-B-1}|[\rvf]_{0}^{i-B-1}, \rvs^n_{-1})}+ 2^{\frac{2}{n}h(\tilde{n}^n)}
\end{multline} This further reduces to
\begin{multline}
h(\rvs^n_{i}~|~[\rvf]_{0}^{i-B-1}, \rvs^n_{-1})  \ge \\ \frac{n}{2} \log \left( \rho^{2(B+1)}2^{\frac{2}{n}h(\rvs^n_{i-B-1}|[\rvf]_{0}^{i-B-1}, \rvs^n_{-1})} \!+\!2 \pi e (1\!-\!\rho^{2({B+1})}) \right).\label{eq:31}
\end{multline} 
It remains to lower bound the entropy term in the right hand side of~\eqref{eq:31}. We show the following in Appendix~\ref{ApA}. 
\begin{lemma}
\label{lem:1}
For any $k\ge 0$ 
\begin{align}
2^{\frac{2}{n}h(\rvs^n_{k}|[\rvf]_{0}^{k}, \rvs_{-1}^n)} \ge  \frac{2 \pi e(1-\rho^2)}{2^{2R}-\rho^2}\left(1-\left(\frac{\rho^2}{2^{2R}}\right)^k\right)  \label{eq:lem1}
\end{align}
$\hfill\Box$
\end{lemma} 

Upon substituting,~\eqref{eq:lem1}, \eqref{eq:31}, and \eqref{eq:3} into \eqref{eq:g2} we obtain that for each ${i \ge B+1}$

\begin{multline}
R \ge  \frac{1}{2} \log \bigg[ \frac{\rho^{2(B+1)}(1-\rho^2)}{D(2^{2R}-\rho^2)}\left(1-\left(\frac{\rho^2}{2^{2R}}\right)^{i-B-1}\right) \\ + \frac{1-\rho^{2(B+1)}}{D} \bigg].
\end{multline} 
{Selecting the largest value of $i$, i.e.  $\mathcal{L}$, yields the tightest lower bound. As mentioned earlier, we are interested in infinite horizon when $\mathcal{L}\to \infty$,
which yields the tightest lower bound, we have}
\begin{align}
R \ge \frac{1}{2} &\log \left( \frac{\rho^{2(B+1)}(1-\rho^2)}{D(2^{2R}-\rho^2)} + \frac{1-\rho^{2(B+1)}}{D} \right)~\label{eq:R-LB-i-infty}
\end{align}
Rearranging~\eqref{eq:R-LB-i-infty} we have that
\begin{align}
D2^{4R} - (D\rho^2 +1 -\rho^{2(B+1)}) 2^{2R} + \rho^2 (1-\rho^{2B}) \ge 0 \label{eq:quad}
\end{align} 

Since~\eqref{eq:quad} is a quadratic equation in $2^{2R},$ it can be easily solved. Keeping the root that yields $R>0$ results in the lower bound in \eqref{eq:thm-1} in Prop.~\ref{prop:GML}. This completes the proof.  

\begin{remark}
Upon examining the proof of the lower bound of Prop.~\ref{prop:GML}, we note that it applies to any source process that satisfies~\eqref{eq:GM-Def} and where the additive noise is i.i.d.\ $\cN(0, 1-\rho^2)$. We do not use the fact that the source process is itself a Gaussian process. 
\end{remark}

\subsection{Coding Scheme: Single Burst Erasure}
\label{sec:UpperGM}
The achievable rate is based on quantization and binning. 
For each $i\ge0$, we consider the test channel
\begin{align}
\rvu_{i} = \rvs_{i} + \rvz_{i},\label{eq:testchannel}
\end{align} where $\rvz_{i} \sim \cN(0,\sigma^2_{z})$ is independent Gaussian noise.
At time $i$ we sample a total of $2^{n{(I(\rvu_i;\rvs_i)+\eps)}}$ codeword sequences i.i.d.\ from $\cN(0, 1+\sigma_z^2)$. 
The codebook at each time is partitioned into $2^{nR}$ bins. The encoder finds the codeword  sequence~$\rvu_{i}^n$ typical with the source sequence $\rvs_{i}^n$ and transmits the bin index $\rvf_{i}$ assigned to $\rvu_{i}^n$.

The decoder, upon receiving $\rvf_i$  attempts to decode  $\rvu^n_{i}$ at time $i$, using all the previously recovered codewords $\{\rvu^n_j:0\le j \le i-1, g_{j}\neq \star\}$ and the source sequence $\rvs_{-1}^n$ as side information.  The reconstruction sequence $\hat{\rvs}_{i}^n$  is the minimum mean square error (MMSE) estimate of $\rvs_{i}^n$ given $\rvu_i^n$ and the past sequences.  The coding scheme presented here is based on binning, similar to lossless case discussed in Section~\ref{sec:UBLB}. The main difference in the analysis is that, unlike the lossless case, neither the recovered sequences $\rvu^n_{i}$ nor reconstructed source sequences $\hat{\rvs}^n_{i}$ inherit the Markov property of the original source sequences $\rvs^n_{i}$. Therefore, unlike the lossless case, the decoder does not reset following a burst erasure, once the error propagation is completed. Since the effect of a burst erasure persists throughout, the analysis of achievable rate is significantly more involved.
  
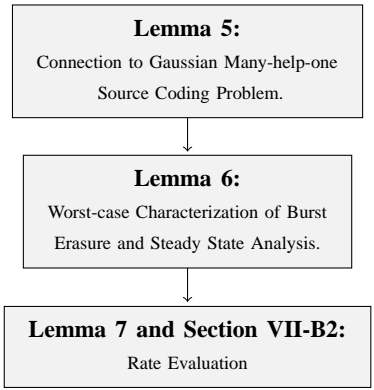
\begin{figure}
\begin{center}
\vspace{1em}
%
%
%
%
%
\begin{tikzpicture}


\draw [->] (0,-.5) -- (0,-1.2);
\draw [->] (0,-2.5) -- (0,-3.2);

\node[draw, fill=gray!10!white] at (0,0) {\begin{tabular}{c} \small{{\bf  Lemma~\ref{lem:New}:}} \\
\scriptsize{Connection to Gaussian Many-help-one }\\
\scriptsize{ Source Coding Problem.}
\end{tabular}};

\node[draw, fill=gray!10!white] at (0,-2) {\begin{tabular}{c} \small{{\bf  Lemma~\ref{lem:3Step}:}} \\
\scriptsize{Worst-case Characterization of Burst}\\
\scriptsize{Erasure and Steady State Analysis.}
\end{tabular}};

\node[draw, fill=gray!10!white] at (0,-3.8) {\begin{tabular}{c} \scriptsize{{\bf  \small{Lemma~\ref{claim:1} and Section~\ref{sec:NE}}:}}\\
\scriptsize{Rate Evaluation}
\end{tabular}};

\end{tikzpicture}
\caption{Flowchart summarizing the  proof steps of Prop.~\ref{prop:GMAch}.}
\label{fig:Chart}
\end{center}
\end{figure}

Fig.~\ref{fig:Chart} summarizes the main steps in proving Prop.~\ref{prop:GMAch}. In particular, in Lemma~\ref{lem:New}, we first derive necessary parametric  rate constraints associated with every possible erasure pattern. Second, through the Lemma~\ref{lem:3Step}, we characterize the  the worst-case erasure pattern that dominates the rate and distortion constraints.  Finally in Lemma~\ref{claim:1} and Section~\ref{sec:NE}, we evaluate the achievable rate to complete the proof of Prop.~\ref{prop:GMAch}.

\subsubsection{Analysis of Achievable Rate}
Given a collection of random variables ${\mathcal V}$, we let the MMSE estimate of $\rvs_i$ be denoted by $\hat{\rvs}_{i}({\mathcal V})$, and its associated estimation error is denoted by $\sigma^2_{i}({\mathcal V})$, i.e.,
\begin{align}
\hat{\rvs}_{i}({\mathcal V}) &\!= E\left[ \rvs_{i}~|~{\mathcal V} \right]\label{eq:s_est-def}\\
\sigma^2_{i}({\mathcal V}) &\!= E[\left(\rvs_{i} -\hat{\rvs}_{i}({\mathcal V})\right)^2] \label{eq:sig-def}.
\end{align}    

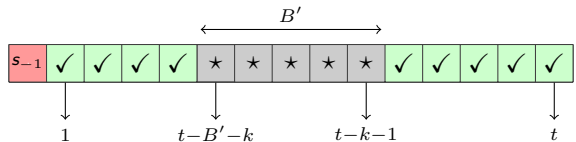
\begin{figure}
\begin{center}
\vspace{1em}
\begin{tikzpicture}
\draw[step=.6cm,color=gray!50!black] (-.5,0) grid (14*.5,.5);

\fill[color=gray!40!white] (3*.5,0) rectangle (5*.5,.5);
\draw[step=.5cm,color=gray!50!black] (-.5,0) grid (14*.5,.5);
\draw [<->] (4*.5+.05,.7)--(9*.5-.05 ,.7);
\draw [white] (6.5*.5,.9) -- (6.5*.5,.9) node {$\color{black}\scriptstyle  B'$};
\fill[color=gray!40!white] (5*.5,0) rectangle (9*.5,.5);

\foreach \t in {0}{
\fill[color=red!40!white] (\t*.5-.5,0) rectangle (\t*.5,.5);
\draw [white] (\t*.5-.5/2,0.25) -- (\t*.5-.5/2,0.25) node {$\color{black}\scriptstyle \rvs_{\scriptscriptstyle{-1}}$};
}

\foreach \t in {1,2,3,4,10,11,12,13,14}{
\fill[color=green!20!white] (\t*.5-.5,0) rectangle (\t*.5,.5);
\draw [white] (\t*.5-.5/2,0.25) -- (\t*.5-.5/2,0.25) node {$\color{black}\checkmark$};
}
\foreach \t in {5,6,7,8,9}{
\draw [white] (\t*.5-.5/2,0.25) -- (\t*.5-.5/2,0.25) node {$\color{black}\star$};
}

\foreach \t in {1}{
\draw [->] (\t*.5-.5/2,0) -- (\t*.5-.5/2,-0.5);
\draw [white] (\t*.5-.5/2,-.7) -- (\t*.5-.5/2,-.7) node {$\color{black} \scriptstyle{1}$};
}

\foreach \t in {5}{
\draw [->] (\t*.5-.5/2,0) -- (\t*.5-.5/2,-0.5);
\draw [white] (\t*.5-.5/2,-.7) -- (\t*.5-.5/2,-.7) node {$\color{black} \scriptstyle{t-B'-k}$};
}
\foreach \t in {9}{
\draw [->] (\t*.5-.5/2,0) -- (\t*.5-.5/2,-0.5);
\draw [white] (\t*.5-.5/2,-.7) -- (\t*.5-.5/2,-.7) node {$\color{black} \scriptstyle{t-k-1}$};
}
\foreach \t in {14}{
\draw [->] (\t*.5-.5/2,0) -- (\t*.5-.5/2,-0.5);
\draw [white] (\t*.5-.5/2,-.7) -- (\t*.5-.5/2,-.7) node {$\color{black} \scriptstyle{t}$};
}
\draw[step=.5cm,color=gray!50!black] (-.5,0) grid (14*.5,.5);
\end{tikzpicture}
\caption{{Schematic of single burst erasure channel model. The channel inputs in the interval $[t-B'-k, t-k-1]$ is erased for some $0\le B'\le B$ and $k \in [0,t-B']$. The rest are available at the decoder, as shown by check mark in the figure.}}
\label{fig:Channel_Single}
\end{center}
\end{figure}
We begin with a parametric characterization of the achievable rate. 

\begin{lemma}
\label{lem:New}
A rate-distortion pair $(R,D)$ is achievable, if for every $t\ge0$, $B'\in[0,B]$ and $k \in[0,t-B']$  we have 
\begin{align}
R  \ge \lambda_t(k,B') \triangleq I(\rvs_{t}; \rvu_{t}~|~[\rvu]_{0}^{t-B'-k-1}, [\rvu]_{t-k}^{t-1}, \rvs_{-1})  \label{1},
\end{align} and the test-channel~\eqref{eq:testchannel} satisfies
\begin{align}
\gamma_t(k,B') &\triangleq E\left[ \left(\rvs_{t} - \hat{\rvs}_{t}([\rvu]_{0}^{t-B'-k-1}, [\rvu]_{t-k}^{t},\rvs_{-1}) \right)^2\right]\notag\\
& =  \sigma^2_{t}([\rvu]_{0}^{t-B'-k-1}, [\rvu]_{t-k}^{t},\rvs_{-1}) \le D  \label{2}.
\end{align}
where $\sigma^2_t(\cdot)$ and $\hat{\rvs}_t(\cdot)$ are defined in~\eqref{eq:sig-def} and~\eqref{eq:s_est-def} respectively.
$\hfill\Box$
\end{lemma}
\begin{proof} 
Consider the decoder at any time $t\ge 0$ outside the error propagation window. 
Assume that a single burst erasure of length $B'\in[0,B]$ spans the interval $[t-B'-k, t-k-1]$ for some $k \in [0,t-B']$ i.e., 
\begin{align}
\rvg_{j} =\begin{cases}
\star, &j\in\{t-B'-k, \ldots, t-k-1\} \\
\rvf_j, &\text{ else}.
\end{cases} \label{eq:EP}
\end{align} 
The schematic of the erasure channel is illustrated in Fig.~\ref{fig:Channel_Single}. Notice that $k=0$ represents the case of the most recent burst erasure spanning the interval $[t-B'-1, t-1]$. The decoder is interested in first successfully recovering $\rvu^n_{t}$ and then reconstructing $\rvs_{t}^n$ within distortion $D$ by performing MMSE estimation of $\rvs_{t}^n$ from all the previously recovered sequences $\rvu_{i}^{n}$ where $i\le t$ and $\rvg_{i} \neq \star$. The decoder succeeds with high probability if the rate constraint satisfies~\eqref{1} (see e.g.,~\cite{tavildar:10}) and the distortion constraint satisfies~\eqref{2}. If these constraints hold for all the possible triplets $(t,B',k)$, the decoder is guaranteed to succeed in reproducing any source sequence within desired distortion $D$. 

Finally in the streaming setup, we can follow the argument similar to that in Section~\ref{sec:UBLB} to argue that the decoder succeeds in the entire horizon of $\cL$  provided we select the source length $n$ to be sufficiently large.
The formal proof is omitted here. 
\end{proof}

As a result of Lemma~\ref{lem:New}, in order to compute the achievable rate, we need to characterize the worst case values of  $(t,k,B')$ that simultaneously maximize $\lambda_t(k,B)$ and $\g_t(k,B)$. We present such a characterization next.

\begin{lemma}
\label{lem:3Step}
The functions $\lambda_t(k,B)$ and $\gamma_t(k,B)$ satisfy the following properties: \begin{enumerate}
\item  For all $t\ge B'$ and $k\in[0,t-B']$, $\lambda_t(k,B') \le \lambda_t(0,B')$ and $\gamma_t(k,B') \le \gamma_t(0,B')$, i.e. the worst-case erasure pattern contains the burst erasure in the interval $[t-B, t-1]$. 
\item  For all $t\ge B$ and $0\le B' \le B$, $\lambda_t(0,B') \le \lambda_t(0,B)$ and $\gamma_t(0,B') \le \gamma_t(0,B)$, i.e. the worse-case erasure pattern includes maximum burst length. 
\item  For a fixed $B$, the functions $\lambda_t(0,B)$ and $\gamma_t(0,B)$ are both  increasing with respect to $t,$ for $t\ge B$, i.e. the worse-case erasure pattern happens in steady state (i.e., $t\rightarrow\infty$) of the system.
\item  For all $t<B$, $0\le B' \le t$ and $k\in[0,t-B']$, $\lambda_t(k,B') \le \lambda_B(0,B)$ and $\gamma_t(k,B') \le \gamma_B(0,B)$ i.e., the burst erasure spanning $[0,B-1]$ dominates all burst erasures that terminate before time ${B-1}$.
\end{enumerate}

$\hfill\Box$
\end{lemma}

\begin{proof} 
Before establishing the proof, we state two inequalities which are established in Appendix~\ref{App:hStep2}. 
For each $k \in [1:t-B']$ we have that:
\begin{align}
&h( \rvu_{t} | [\rvu]_{0}^{t-B'-k-1}, [\rvu]_{t-k}^{t-1}, \rvs_{-1})\notag\\
&\qquad \qquad \le h( \rvu_{t} | [\rvu]_{0}^{t-B'-k}, [\rvu]_{t-k+1}^{t-1}, \rvs_{-1})\label{eq:hStep2}, \\
&h( \rvs_{t} | [\rvu]_{0}^{t-B'-k-1}, [\rvu]_{t-k}^{t}, \rvs_{-1})\notag\\
&\qquad \qquad \le h( \rvs_{t} | [\rvu]_{0}^{t-B'-k}, [\rvu]_{t-k+1}^{t}, \rvs_{-1})\label{eq:hStep2D}.
\end{align}
The above inequalities state that the conditional differential entropy of $\rvu_{t}$ and $\rvs_{t}$ is reduced if the variable $\rvu_{t-B'-k}$ is replaced by $\rvu_{t-k}$ in the conditioning and the remaining variables remain unchanged.
 Fig.~\ref{fig:LemmaFig} provides a schematic interpretation of the above inequalities. The proof in Appendix~\ref{App:hStep2} exploits the specific structure of the Gaussian test channel~\eqref{eq:testchannel} and Gaussian sources to establish these inequalities.

In the remainder of the proof, we establish each of the four properties separately.

1) We show that both $\lambda_t(k,B')$ and $\gamma_t(k,B')$ are decreasing functions of $k$ for $k \in [1:t-B']$. 

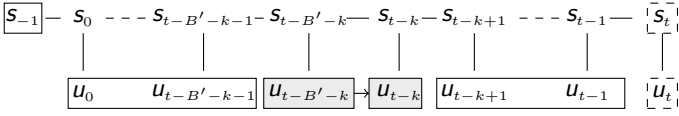
\begin{figure}
\begin{center}
\vspace{1em}
\begin{tikzpicture}

\draw[-] (.8,-.2) -- (.8,-.7);
\draw[-] (2.4,-.2) -- (2.4,-.7);
\draw[-] (3.8,-.2) -- (3.8,-.7);
\draw[-] (5,-.2) -- (5,-.7);
\draw[-] (6,-.2) -- (6,-.7);
\draw[-] (7.5,-.2) -- (7.5,-.7);
\draw[-] (8.5,-.2) -- (8.5,-.7);

\draw[-] (.3,0) -- (.5,0);
\draw[ dashed] (1.1,0) -- (1.7,0);
\draw[-] (3.08,0) -- (3.2,0);
\draw[-] (4.3,0) -- (4.6,0);
\draw[-] (5.3,0) -- (5.5,0);
\draw[dashed] (6.6,0) -- (7.2,0);
\draw[-] (7.8,0) -- (8.1,0);

\draw[->] (4.3,-1) -- (4.6,-1);

\draw [fill= gray!15!white](3.2,-1.2) rectangle (4.4,-.8);
\draw [fill= gray!15!white](4.6,-1.2) rectangle (5.3,-.8);

\draw [-](5.5,-1.2) rectangle (8,-.8);
\draw [-](.6,-1.2) rectangle (3.1,-.8);
\draw [-](-.25,-.2) rectangle (.25,.2);
\draw [dashed](8.3,-1.2) rectangle (8.7,-.8);
\draw [dashed](8.3,-.2) rectangle (8.7,.2);

\draw [white] (0,0) -- (0,0) node {$\color{black} \rvs_{\scriptscriptstyle{-1}}$};
\draw [white] (.8,0) -- (.8,0) node {$\color{black} \rvs_{\scriptscriptstyle{0}}$};
\draw [white] (2.4,0) -- (2.4,0) node {$\color{black} \rvs_{\scriptscriptstyle{t-B'-k-1}}$};
\draw [white] (3.8,0) -- (3.8,0) node {$\color{black} \rvs_{\scriptscriptstyle{t-B'-k}}$};
\draw [white] (5,0) -- (5,0) node {$\color{black} \rvs_{\scriptscriptstyle{t-k}}$};
\draw [white] (6,0) -- (6,0) node {$\color{black} \rvs_{\scriptscriptstyle{t-k+1}}$};
\draw [white] (7.5,0) -- (7.5,0) node {$\color{black} \rvs_{\scriptscriptstyle{t-1}}$};
\draw [white] (8.5,0) -- (8.5,0) node {$\color{black} \rvs_{\scriptscriptstyle{t}}$};

\draw [white] (.8,-1) -- (.8,-1) node {$\color{black} \rvu_{\scriptscriptstyle{0}}$};
\draw [white] (2.4,-1) -- (2.4,-1) node {$\color{black} \rvu_{\scriptscriptstyle{t-B'-k-1}}$};
\draw [white] (3.8,-1) -- (3.8,-1) node {$\color{black} \rvu_{\scriptscriptstyle{t-B'-k}}$};
\draw [white] (5,-1) -- (5,-1) node {$\color{black} \rvu_{\scriptscriptstyle{t-k}}$};
\draw [white] (6,-1) -- (6,-1) node {$\color{black} \rvu_{\scriptscriptstyle{t-k+1}}$};
\draw [white] (7.5,-1) -- (7.5,-1) node {$\color{black} \rvu_{\scriptscriptstyle{t-1}}$};
\draw [white] (8.5,-1) -- (8.5,-1) node {$\color{black} \rvu_{\scriptscriptstyle{t}}$};

\end{tikzpicture}
\caption{Replacing $\rvu_{t-B'-k}$ by $\rvu_{t-k}$ improves the estimate of $\rvs_{t}$ and $\rvu_{t}$.}
\label{fig:LemmaFig}
\end{center}
\end{figure}

\begin{align}
\lambda_t(k,B') &= I(\rvs_{t}; \rvu_{t} | [\rvu]_{0}^{t-B'-k-1}, [\rvu]_{t-k}^{t-1}, \rvs_{-1}) \notag\\
& = h( \rvu_{t} | [\rvu]_{0}^{t-B'-k-1}, [\rvu]_{t-k}^{t-1}, \rvs_{-1}) - h(\rvu_{t} | \rvs_{t})\notag\\
& \le h( \rvu_{t} | [\rvu]_{0}^{t-B'-k}, [\rvu]_{t-k+1}^{t-1}, \rvs_{-1}) - h(\rvu_{t} | \rvs_{t}) \label{6}\\
& = I(\rvs_{t}; \rvu_{t} | [\rvu]_{0}^{t-B'-k}, [\rvu]_{t-k+1}^{t-1}, \rvs_{-1}) \notag\\
& = \lambda_t(k-1,B')\label{eq:step1-r}
\end{align} where 
\eqref{6} follows from using~\eqref{eq:hStep2}. In a similar fashion since $$\gamma_t(k,B') = \sigma^2_t\left([\rvu]_{0}^{t-B'-k}, [\rvu]_{t-k+1}^{t}, \rvs_{-1}\right)$$ is the MMSE estimation error of $\rvs_t$ given $\left( [\rvu]_{0}^{t-B'-k}, [\rvu]_{t-k+1}^{t}, \rvs_{-1}\right)$, we have
\begin{align}
\frac{1}{2}\log\left(2\pi e \cdot\gamma_t(k,B')\right) &= h(\rvs_{t}| [\rvu]_{0}^{t-B'-k-1}, [\rvu]_{t-k}^{t}, \rvs_{-1}) \notag\\
& \le h( \rvs_{t} | [\rvu]_{0}^{t-B'-k}, [\rvu]_{t-k+1}^{t}, \rvs_{-1}) \label{7}\\
& =\frac{1}{2}\log\left(2\pi e \cdot\gamma_t(k-1,B')\right)\label{eq:step1-d}
\end{align} where \eqref{7} follows from using~\eqref{eq:hStep2D}. 
Since $f(x) = \frac{1}{2}\log(2\pi e x)$ is a monotonically increasing function it follows that ${\gamma_t(k,B') \le \gamma_t(k-1,B')}$. By recursively applying~\eqref{eq:step1-r} and~\eqref{eq:step1-d} until $k=1$, the proof of property (1) is complete.

2) We next show that the worst case erasure pattern also has the longest burst. This follows intuitively since the decoder can just ignore some of the symbols received over the channel. Thus any rate achieved with the longest burst is also achieved for the shorter burst. The formal justification is as follows. For any $B' \le B$ we have,
{\allowdisplaybreaks{\begin{align}
\lambda_t(0,B') &= I(\rvs_{t}; \rvu_{t} | [\rvu]_{0}^{t-B'-1}, \rvs_{-1}) \notag\\
&= h( \rvu_{t} | [\rvu]_{0}^{t-B'-1}, \rvs_{-1}) - h( \rvu_{t} | \rvs_{t}) \label{3}\\
&= h( \rvu_{t} | [\rvu]_{0}^{t-B-1}, [\rvu]_{t-B}^{t-B'-1}, \rvs_{-1}) - h( \rvu_{t} | \rvs_{t}) \notag\\
& \le h( \rvu_{t} | [\rvu]_{0}^{t-B-1}, \rvs_{-1}) - h( \rvu_{t} | \rvs_{t}) \label{4}\\
& = I(\rvs_{t}; \rvu_{t} | [\rvu]_{0}^{t-B-1},  \rvs_{-1})\label{5}\\
& = \lambda_t(0,B) \label{03}
\end{align}}} where \eqref{3} and \eqref{5} follows from the Markov chain property
\begin{align}
\rvu_{t} \rightarrow \rvs_{t}\rightarrow \{[\rvu]_{0}^{t-j-1}, \rvs_{-1}\}, \quad j\in\{B,B'\}
\end{align} and \eqref{4} follows from the fact that conditioning reduces differential entropy. 
In a similar fashion the inequality $\gamma_t(0,B') \le \gamma_t(0,B)$  follows from the fact that the estimation error can only be reduced by having more observations.

3) We show that both $\lambda_t(0,B)$ and $\gamma_t(0,B)$ are increasing functions with respect to $t$. Intuitively as $t$ increases the effect of having $\rvs_{-1}$ at the decoder vanishes and hence the required rate increases.
 Consider 
\begin{align}
\lambda_{t+1}(0,B) &= I(\rvs_{t+1}; \rvu_{t+1} | [\rvu]_{0}^{t-B}, \rvs_{-1})\notag\\
& = h(\rvu_{t+1} | [\rvu]_{0}^{t-B}, \rvs_{-1}) - h(\rvu_{t+1} | \rvs_{t+1})\notag\\
&=  h(\rvu_{t+1} | [\rvu]_{0}^{t-B}, \rvs_{-1}) -  h(\rvu_{t} | \rvs_{t})\label{eq:tinv_1}\\
& \ge  h(\rvu_{t+1} | [\rvu]_{0}^{t-B}, \rvs_{-1}, \rvs_{0}) -  h(\rvu_{t} | \rvs_{t})\label{8}\\
& =   h(\rvu_{t+1} | [\rvu]_{1}^{t-B},  \rvs_{0}) -  h(\rvu_{t} | \rvs_{t})\label{9}\\
& =   h(\rvu_{t} | [\rvu]_{0}^{t-B-1},  \rvs_{-1}) -  h(\rvu_{t} | \rvs_{t})\label{10}\\
& =  I(\rvs_{t}; \rvu_{t} | [\rvu]_{0}^{t-B-1}, \rvs_{-1}) \notag\\
& = \lambda_t(0,B)\label{eq:la_inc_1}
\end{align} where~\eqref{eq:tinv_1} and \eqref{10} follow from time-invariant property of the source model and the test channel,
\eqref{8} follows from the fact that conditioning reduces differential entropy and \eqref{9} uses the following Markov chain property
\begin{align}
\{\rvu_{0}, \rvs_{-1}\}  \rightarrow  \{[\rvu]_{1}^{t-B}, \rvs_{0}\}  \rightarrow \rvu_{t+1}.
\end{align}
Similarly, 
\begin{align}
\frac{1}{2}\log\left(2\pi e \cdot\gamma_{t+1}(0,B)\right) &= h(\rvs_{t+1}| [\rvu]_{0}^{t-B}, \rvu_{t+1}, \rvs_{-1}) \notag\\
&\ge h(\rvs_{t+1}| [\rvu]_{0}^{t-B}, \rvu_{t+1}, \rvs_{0},\rvs_{-1}) \notag\\
&= h(\rvs_{t+1}| [\rvu]_{1}^{t-B}, \rvu_{t+1}, \rvs_{0}) \label{11}\\
& = h(\rvs_{t}| [\rvu]_{0}^{t-B-1}, \rvu_{t}, \rvs_{-1})\notag\\
& =\frac{1}{2}\log\left(2\pi e \cdot\gamma_t(0,B)\right)\label{eq:d_inc_1}
\end{align} where \eqref{11} follows from the following Markov chain property
 \begin{align}
\{\rvu_{0}, \rvs_{-1}\}  \rightarrow  \{[\rvu]_{1}^{t-B}, \rvu_{t+1}, \rvs_{0}\}  \rightarrow \rvs_{t+1}.
\end{align} 
Since~\eqref{eq:la_inc_1} and~\eqref{eq:d_inc_1} hold for every $t \ge B$ the proof of property (3) is complete.

4) Note that for $t<B$ we have $0\le B' \le t$ and thus we can write 
\begin{align}
\lambda_{t}(k,B') &\le \lambda_{t}(0,B')\label{wh1}\\
&\le \lambda_{t}(0,t) \label{wh}\\
&= h(\rvu_{t} | \rvs_{-1}) -  h(\rvu_{t} | \rvs_{t}) \notag\\
&=h(\rvu_{t} | \rvs_{-1}) -  h(\rvu_{B} | \rvs_{B}) \notag\\
&=h(\rvu_{B} | \rvs_{B-t-1}) -  h(\rvu_{B} | \rvs_{B})\notag\\
& \le  h(\rvu_{B} | \rvs_{-1}) -  h(\rvu_{B} | \rvs_{B}) \label{022}\\
&  =\lambda_{B}(0,B)
\end{align} where \eqref{wh1} follows from part 1 of the lemma, \eqref{wh} is based on the fact that the worse-case erasure pattern contains most possible erasures and follows from the similar steps used in deriving \eqref{03} and using the fact that if $t<B$, the burst erasure length is at most $t$. Eq.~\eqref{022} follows from the fact that whenever $t<B$ the relation $\rvs_{-1}\rightarrow \rvs_{B-t-1} \rightarrow \rvu_B$ holds since $t<B$ is assumed. In a similar fashion we can show that
$\gamma_{t}(k,B') \le \gamma_{B}(0,B)$.

This completes the proof of lemma~\ref{lem:3Step}.
\end{proof}

Following the four parts of Lemma~\ref{lem:3Step}, it follows that the worst-case erasure pattern happens at steady state i.e. $t\to \infty$ when there is a burst of length $B$ which spans $[t-B,t-1]$. According to this and Lemma~\ref{lem:New}, any pair $(R,D)$ is achievable if 
\begin{align}
&R  \ge \lim _{t\to \infty} \lambda_t(0,B)\label{imp1}\\
&D \ge \lim _{t\to \infty} \gamma_t(0,B)\label{imp2} 
\end{align}

{\begin{lemma}
\label{claim:1}
Consider $\rvu_{i} = \rvs_{i} + \rvz_{i}$ and suppose the noise variance $\sigma^2_{z}$ satisfies
\begin{align}
\Gamma(B, \sigma^2_{z}) &\triangleq \lim_{t \to \infty} E\left[\left(\rvs_{t} - \hat{\rvs}_{t}([\rvu]_{0}^{t-B-1}, \rvu_{t}) \right)^2\right] \\
& = \lim_{t \to \infty} \sigma^2_{t} \left([\rvu]_{0}^{t-B-1}, \rvu_{t}\right)\le D. \label{eq:DC}
\end{align} 
The following rate is achievable:
\begin{align}
R = \Lambda(B,\sigma^2_{z}) \triangleq \lim_{t\to \infty} I(\rvs_{t} ; \rvu_{t} | [\rvu]_{0}^{t-B-1}). \label{eq:RateC}
\end{align}
$\hfill\Box$
\end{lemma}
\begin{proof}
It suffices to show that any test channel satisfying \eqref{eq:DC} also implies \eqref{imp2}
and any rate satisfying \eqref{eq:RateC} implies  \eqref{imp1}. These relations can be established in a straightforward manner as shown below. 
\begin{align}
R= \Lambda(B,\sigma^2_{z}) &= \lim_{t\to \infty} I(\rvs_{t} ; \rvu_{t} | [\rvu]_{0}^{t-B-1}) \notag\\
&= \lim_{t\to \infty} \left( h( \rvu_{t} | [\rvu]_{0}^{t-B-1}) - h( \rvu_{t} | \rvs_{t})\right)\\
&\ge \lim_{t\to \infty} \left( h( \rvu_{t} | [\rvu]_{0}^{t-B-1}, \rvs_{-1}) - h( \rvu_{t} | \rvs_{t})\right)\\
& =  \lim _{t\to \infty} \lambda_t(0,B)
\end{align} and 
\begin{align}
D\ge \Gamma(B,\sigma^2_{z}) &= \lim_{t \to \infty} E\left[\left(\rvs_{t} - \hat{\rvs}_{t}([\rvu]_{0}^{t-B-1}, \rvu_{t}) \right)^2\right] \\
& \ge \lim_{t \to \infty} E\left[\left(\rvs_{t} - \hat{\rvs}_{t}([\rvu]_{0}^{t-B-1}, \rvu_{t}, \rvs_{-1}) \right)^2\right]\\
&= \lim _{t\to \infty} \gamma_t(0,B) 
\end{align} 
\end{proof}

{ We conclude that $\Gamma(B, \sigma^2_{z}) =D$, the rate $R^{+}_{\textrm{GM-SE}}(B,D) = \Lambda(B,\sigma^2_{z})$ is achievable.}

\subsubsection{Numerical Evaluation} 
\label{sec:NE}
We derive an expression for  numerically evaluating the noise variance $\sigma_z^2$ in~\eqref{eq:testchannel} and also establish~\eqref{eq:GM-LB-R} and~\eqref{eq:GM-LB-RR}.

To this end it is helpful to consider the following single-variable discrete-time Kalman filter for $i \in[0, t-B-1]$,
\begin{align}
&\rvs_i=\rho \rvs_{i-1}+ \rvn_{i}, \qquad \rvn_{i} \sim N(0,1-\rho^2) \label{eq:Kalman1}\\
&\rvu_i= \rvs_{i}+ \rvz_{i},\qquad \rvz_{i} \sim N(0,\sigma^{2}_z) \label{eq:Kalman2}.
\end{align} 
Note that $\rvs_{i}$  can be viewed as the state of the system updated according a Gauss-Markov model and $\rvu_{i}$ as the output of the system at each time $i$, which is a noisy version of the state $\rvs_{i}$. Consider the system in steady state i.e. $t\to \infty$. The MMSE estimation error of $\rvs_{t-B}$ given {\em all} the previous outputs up to time $t-B-1$ i.e. $[\rvu]_{0}^{t-B-1}$ is expressed as (see, e.g., ~\cite[Example V.B.2]{poor94}):
\begin{align}
&\Sigma(\sigma^{2}_{z}) \triangleq \lim_{t\to \infty}\sigma^2_{t-B}([\rvu]_{0}^{t-B-1})\\
&~=\frac{1}{2}\sqrt{(1-\sigma^2_{z})^2(1-\rho^2)^2+4\sigma^2_{z}(1-\rho^2) } + \frac{1-\rho^2}{2}(1-\sigma^2_{z}) \label{eq:Sigma}
\end{align} Also using the orthogonality principle for MMSE estimation we have
\begin{align}
[\rvu]_{0}^{t-B-1} \rightarrow \hat{\rvs}_{t-B}([\rvu]_{0}^{t-B-1}) \rightarrow \rvs_{t-B}  \rightarrow \rvs_{t} \label{h4}
\end{align}

Thus we can express
\begin{align}
\rvs_{t-B} = \hat{\rvs}_{t-B} ([\rvu]_{0}^{t-B-1}) + \hat{\rve}\label{h2}
\end{align} where the noise $\hat{\rve} \sim \cN(0, \Sigma\left(\sigma^{2}_{z} \right))$ is independent of the observation set $[\rvu]_{0}^{t-B-1}$. 
Equivalently we can express (see e.g. \cite{willskyWornel03})
\begin{align}
\hat{\rvs}_{t-B}([\rvu]_{0}^{t-B-1}) = \tilde{\alpha}\rvs_{t-B} + \tilde{\rve} \label{eq:s-hat-rev}
\end{align}
where
\begin{align}
&\tilde{\alpha} \triangleq 1-\Sigma\left(\sigma^{2}_{z} \right)
\end{align}
and $\tilde{\rve} \sim \cN \left(0, \Sigma\left(\sigma^{2}_{z} \right) \left(1-\Sigma\left(\sigma^{2}_{z} \right)\right)\right)$ is independent of $\rvs_{t-B}$.
Thus we have
\begin{align}
\Lambda(B, \sigma^2_{z}) &= \lim_{t\to \infty}  I(\rvs_{t} ; \rvu_{t} | [\rvu]_{0}^{t-B-1})\notag\\
&= \lim_{t\to \infty}  I(\rvs_{t} ; \rvu_{t} | \hat{\rvs}_{t-B}([\rvu]_{0}^{t-B-1}))\notag\\
&=\lim_{t\to \infty}  I(\rvs_{t} ; \rvu_{t} | \tilde{\alpha}\rvs_{t-B} + \tilde{\rve})\notag\\
&=\lim_{t\to \infty}  I(\rvs_{t} ; \rvu_{t} | \rvs_{t-B} + {\rve})\notag\\
&=  I(\rvs_{t} ; \rvu_{t} | \tilde{\rvs}_{t-B})
\end{align}
where  we have used~\eqref{eq:s-hat-rev} and ${{\rve} \sim \cN(0, \Sigma\left(\sigma^{2}_{z} \right)/(1-\Sigma\left(\sigma^{2}_{z} \right)))}$. This establishes~\eqref{eq:GM-LB-R} in Prop.~\ref{prop:GMAch}.
In a similar manner,
\begin{align}
\Gamma(B,\sigma^{2}_{z})  &= \lim_{t\to \infty }\sigma^{2}_{t} ([\rvu]_{0}^{t-B-1}, \rvu_{t})\notag\\
& = \lim_{t\to \infty }\sigma^{2}_{t} (\hat{\rvs}_{t-B}([\rvu]_{0}^{t-B-1}), \rvu_{t})\notag\\
& = \lim_{t\to \infty }\sigma^{2}_{t} (\tilde{\alpha}{\rvs}_{t-B} + \tilde{\rve}, , \rvu_{t})\notag\\
& = \lim_{t\to \infty }\sigma^{2}_{t} ({\rvs}_{t-B}+ {{\rve}} , \rvu_{t})\notag\\
&= \sigma^{2}_{t} (\tilde{\rvs}_{t-B}, \rvu_{t})\label{h3_2}
\end{align} which establishes~\eqref{eq:GM-LB-RR}. Furthermore since
\begin{align}
\rvs_{t} = \rho^{B} \rvs_{t-B} + \tilde{\rvn}\label{h1}
\end{align} where $\tilde{\rvn} \sim \cN(0, 1-\rho^{2B}),$
\begin{align}
\Gamma(B,\sigma^{2}_{z})  &=\sigma^{2}_{t} (\tilde{\rvs}_{t-B}, \rvu_{t})\\
&= \left[ \frac{1}{\sigma^2_{z}} + \frac{1}{1-\rho^{2B}\left(1-\Sigma(\sigma^2_{z})\right)}\right]^{-1}\label{Gam}
\end{align} where \eqref{Gam} follows from the application of MMSE estimator and using \eqref{h1}, \eqref{h2} and the definition of the test channel in \eqref{eq:testchannel}. 
Thus the noise $\sigma_z^2$ in the test channel~\eqref{eq:testchannel} is obtained by setting
\begin{align}\Gamma(B,\sigma^{2}_{z}) = D. \label{eq:noise-eq}\end{align} 
This completes the proof of Prop.~\ref{prop:GMAch}.

\begin{remark}
\label{rem:GenW}
When $W>0$,  the generalization of Lemma~\ref{lem:3Step} appears to involve a rate-region corresponding to the Berger-Tung inner bound~\cite{tung:78} and the analysis is considerably more involved. Furthermore hybrid schemes involving predictive coding and binning  may lead to an improved performance over the binning-only scheme. Thus the scope of this problem is well beyond the results in this paper.
\end{remark}


\subsection{Coding Scheme: Multiple Burst Erasures with Guard Intervals}
\label{sec:GM_ME}
We study the achievable rate using the quantize and binning scheme with test channel~\eqref{eq:testchannel} when the channel introduces multiple burst erasures each of length no greater than $B$ and with a guard interval of at-least $L$ symbols separating consecutive burst erasures. While the coding scheme is the same as the single burst erasure channel model and is based on quantize and binning and MMSE estimation at the decoder, characterizing the worst case erasure pattern of the channel is main challenge and requires some additional steps.

\subsubsection{Analysis of Achievable Rate}
We introduce the following notation in our analysis.
Let $\Omega_t$ denote the set of time indices up to time $t-1$ when the channel packets are not erased i.e., 
\begin{align}
\Omega_{t} = \{i: 0\le i\le t-1, \rvg_{i} \neq \star\}, \label{eq:Omeg-def}
\end{align} 
and let us define
\begin{align}
& \rvbs_{\Omega} = \{\rvs_{i}:i\in \Omega\},\\
& \rvbu_{\Omega} = \{\rvu_{i}:i\in \Omega\}. 
\end{align}
 
Given the erasure sequence $\Omega_t$, and given $\rvg_t = \rvf_t$, the decoder can reconstruct $\rvu_t^n$ provided that the test channel is selected such that the rate satisfies (see e.g.,~\cite{tavildar:10})
\begin{align}
R\ge \lambda_{t}(\Omega_t) \triangleq I(\rvs_{t} ; \rvu_{t} | \rvbu_{\Omega_t}, \rvs_{-1} ). \label{rate_c}
\end{align} and the distortion constraint satisfies
\begin{align}
\gamma_{t}(\Omega_t) &\triangleq E\left[\left(\rvs_{t} - \hat{\rvs}_{t}(\rvbu_{\Omega_t}, \rvu_t, \rvs_{-1})\right)^2\right] \notag\\
&= \sigma^2_{t} \left(\rvbu_{\Omega_t}, \rvu_t, \rvs_{-1}\right) \le D \label{tchannel_c}
\end{align}
for each $t \ge 0$ and each feasible set $\Omega_t$. Thus we are again required to characterize the $\Omega_t$ for each value of $t$ corresponding to the worst-case erasure pattern. The following two lemmas are useful towards this end.


 
\begin{lemma}
\label{lem:help2}
Consider the two sets $A, B\subseteq \mathbb{N}$ each of size $r$ as $A=\{a_{1},a_{2},\cdots,a_{r}\}$, $B=\{b_{1},b_{2},\cdots,b_{r}\}$ such that $1\le a_1<a_2<\cdots<a_r$ and $1\le b_1<b_2<\cdots<b_r$ and for any $i \in \{1,\ldots,r\}$, $a_{i} \le b_{i}$. Then the test channel~\eqref{eq:testchannel} satisfies the following: 
\begin{align}
&h(\rvs_{t}|\rvbu_{A},\rvs_{-1}) \ge h(\rvs_{t}|\rvbu_{B},\rvs_{-1}), \quad \forall t \ge b_r \label{eq:lem2-help2}\\
&h(\rvu_{t}|\rvbu_{A},\rvs_{-1}) \ge h(\rvu_{t}|\rvbu_{B},\rvs_{-1}), \quad \forall t > b_r \label{eq:lem1-help2}. 
\end{align}
$\hfill\Box$
\end{lemma}

The proof of Lemma~\ref{lem:help2} is available in Appendix~\ref{App:help2}.


\begin{figure}
        \centering
        \begin{subfigure}[b]{0.405\textwidth}
                \centering
%
%
%
%
%

\begin{tikzpicture}

\fill[color=gray!40!white] (8*.4,0) rectangle (9*.4,.4);
\fill[color=gray!40!white] (12*.4,0) rectangle (14*.4,.4);
\fill[color=gray!40!white] (17*.4,0) rectangle (19*.4,.4);
\fill[color=red!40!white] (0*.4,0) rectangle (1*.4,.4);
\draw[step=0.4cm,color=gray!50!black] (0,0) grid (20*.4,.4);

\foreach \t in {12,17}{
\draw [<->] (\t*.4+.05,.5)--(\t*.4+2*.4-.05 ,.5);
\draw [white] (\t*.4+.4,.6) -- (\t*.4+.4,.6) node {$\color{black}\scriptstyle  B$};
}

\foreach \t in {9,14}{
\draw [<->] (\t*.4+.05,.5)--(\t*.4+3*.4-.05 ,.5);
\draw [white] (\t*.4+1.5*.4,.6) -- (\t*.4+1.5*.4,.6) node {$\color{black}\scriptstyle  L$};
}

%
%

\foreach \t in {20}{
\draw [->] (\t*.4-.4/2,0) -- (\t*.4-.4/2,-0.5);
\draw [white] (\t*.4-.4/2,-0.7) -- (\t*.4-.4/2,-0.7) node {$\color{black}\rvs_{\scriptscriptstyle{t}}$};
}
\foreach \t in {1}{
\draw [->] (\t*.4-.4/2,0) -- (\t*.4-.4/2,-0.5);
\draw [white] (\t*.4-.4/2,-0.7) -- (\t*.4-.4/2,-0.7) node {$\color{black}\rvs_{\scriptscriptstyle{-1}}$};
}
%

\foreach \t in {9,13,14,18,19}{
\draw [white] (\t*.4-1*.4+.4/2,.4/2) -- (\t*.4-1*.4+.4/2,.4/2) node {$\color{black}\scriptstyle \star$};
}

\foreach \t in {2,3,4,5,6,7,8,10,11,12,15,16,17,20}{
\draw [white] (\t*.4-1*.4+.4/2,.4/2) -- (\t*.4-1*.4+.4/2,.4/2) node {$\color{black}\scriptstyle \checkmark$};
}

\foreach \t in {0,1,2,3,4,5,6,8,9,10,13,14,15}{
\draw [white] (\t*.4+.4+.4/2,-.2) -- (\t*.4+.4+.4/2,-.2) node {$\color{black}\scriptscriptstyle (\t)$};
}

\end{tikzpicture}

                \caption{$\Omega^{\star}_{18}(13)\!=\! \{0,1,2,3,4,5,6,8,9,10,13,14,15\}$}
                \label{fig:D1}
        \end{subfigure}%
        \qquad \qquad
        \begin{subfigure}[b]{0.405\textwidth}
                \centering
%
%
%
%
%

\begin{tikzpicture}

\fill[color=gray!40!white] (2*.4,0) rectangle (4*.4,.4);
\fill[color=gray!40!white] (7*.4,0) rectangle (9*.4,.4);
\fill[color=gray!40!white] (12*.4,0) rectangle (14*.4,.4);
\fill[color=gray!40!white] (17*.4,0) rectangle (19*.4,.4);
\fill[color=red!40!white] (0*.4,0) rectangle (1*.4,.4);
\draw[step=0.4cm,color=gray!50!black] (0,0) grid (20*.4,.4);

\foreach \t in {2,7,12,17}{
\draw [<->] (\t*.4+.05,.5)--(\t*.4+2*.4-.05 ,.5);
\draw [white] (\t*.4+.4,.6) -- (\t*.4+.4,.6) node {$\color{black}\scriptstyle  B$};
}

\foreach \t in {4,9,14}{
\draw [<->] (\t*.4+.05,.5)--(\t*.4+3*.4-.05 ,.5);
\draw [white] (\t*.4+1.5*.4,.6) -- (\t*.4+1.5*.4,.6) node {$\color{black}\scriptstyle  L$};
}

%
%

\foreach \t in {20}{
\draw [->] (\t*.4-.4/2,0) -- (\t*.4-.4/2,-0.5);
\draw [white] (\t*.4-.4/2,-0.7) -- (\t*.4-.4/2,-0.7) node {$\color{black}\rvs_{\scriptscriptstyle{t}}$};
}
\foreach \t in {1}{
\draw [->] (\t*.4-.4/2,0) -- (\t*.4-.4/2,-0.5);
\draw [white] (\t*.4-.4/2,-0.7) -- (\t*.4-.4/2,-0.7) node {$\color{black}\rvs_{\scriptscriptstyle{-1}}$};
}
%

\foreach \t in {3,4,8,9,13,14,18,19}{
\draw [white] (\t*.4-1*.4+.4/2,.4/2) -- (\t*.4-1*.4+.4/2,.4/2) node {$\color{black}\scriptstyle \star$};
}

\foreach \t in {2,5,6,7,10,11,12,15,16,17,20}{
\draw [white] (\t*.4-1*.4+.4/2,.4/2) -- (\t*.4-1*.4+.4/2,.4/2) node {$\color{black}\scriptstyle \checkmark$};
}

\foreach \t in {0,3,4,5,8,9,10,13,14,15}{
\draw [white] (\t*.4+.4+.4/2,-.2) -- (\t*.4+.4+.4/2,-.2) node {$\color{black}\scriptscriptstyle (\t)$};
}

\end{tikzpicture}

                \caption{$\Omega^{\star}_{18}= \! \{0,3,4,5,8,9,10,13,14,15\}$ }
                \label{fig:D2}
        \end{subfigure}
        \caption{Schematic of the erasure patterns in Lemma~\ref{lem:GM-ME} for $t=18$, $L=3$ and $B=2$. Fig.~\ref{fig:D1} illustrates, $\Omega^{\star}_{t}(\theta)$ in part~1 of Lemma~\ref{lem:GM-ME}. The non-erased symbols are denoted by check-marks. Fig.~\ref{fig:D2} illustrates  $\Omega^{\star}_{t}$ as stated in part~2 of Lemma~\ref{lem:GM-ME}.}
        \label{fig:D1D2}
\end{figure}
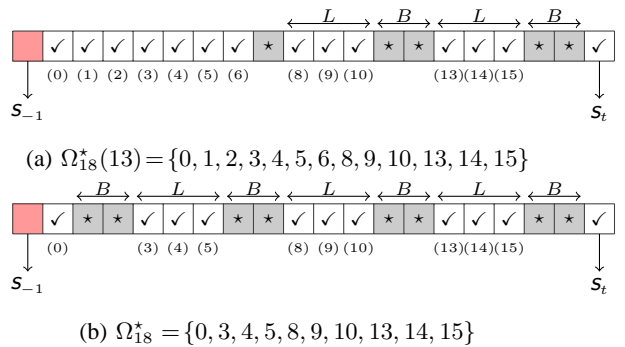


\begin{lemma}
\label{lem:GM-ME}
Assume that at time $t$, $\rvg_t =\rvf_t$ and let $\Omega_t$ be as defined in~\eqref{eq:Omeg-def} .
\begin{enumerate}
\item Among all feasible sets $\Omega_t$  of size  $|\Omega_t|=\theta$, 
$\lambda_{t}(\Omega_t)$ and $\gamma_{t}(\Omega_t)$ are maximized by a set $\Omega^{\star}_t(\theta)$
where all the erasures happen in the closest possible locations to time $t$. 
  

\item For each fixed $t$, both $\lambda_{t}(\Omega^\star_t(\theta))$ and $\gamma_t(\Omega_t^\star(\theta))$ are maximized by the minimum possible value of $\theta$. Equivalently, the maximizing set, denoted by $\Omega^\star_t$, corresponds to the erasure pattern with maximum number of erasures. 
\item Both $\lambda_{t}(\Omega^{\star}_t)$ and $\gamma_{t}(\Omega^{\star}_t)$ are increasing functions with respect to $t$. 
 
\end{enumerate}
$\hfill\Box$
\end{lemma}

The proof of Lemma~\ref{lem:GM-ME} is presented in Appendix~\ref{App:Lemma7}. We present an example  in Fig.~\ref{fig:D1D2} to illustrate Lemma~\ref{lem:GM-ME}. We assume $t=18$. The total number of possible erasures up to time $t=18$ is restricted to be $5$, or equivalently the number of non-erased packets is $\theta = 13$ in Fig~\ref{fig:D1}. The set $\Omega^{\star}_{18}(13)$ indicates the set of  non-erased indices associated with the worst case erasure pattern.  Based on part 2 of  Lemma~\ref{lem:GM-ME}, Fig.~\ref{fig:D2} shows the worst case erasure pattern for time $t=18$, which includes the maximum possible erasures.

Following  the three steps in Lemma~\ref{lem:GM-ME} a rate-distortion pair $(R,D)$ is achievable if
\begin{align}
&R\ge \lim_{t\to \infty} \lambda_{t}(\Omega^{\star}_t) \label{eq:limRate}\\
&D\ge \lim_{t\to \infty} \gamma_{t}(\Omega^{\star}_t) \label{eq:limtest}
\end{align}

{\begin{lemma}
\label{Claim:2}
Any test channel noise $\sigma^2_{z}$ satisfying~\eqref{eq:GM-ME-Ach} and~\eqref{eq:GM-ME-Ach-D} in Prop.~\ref{prop:GM-ME}, i.e. 
\begin{align}
&R\ge I(\rvs_{t};\rvu_{t} | \tilde{\rvs}_{t-L-B} ,  [\rvu]_{t-L-B+1}^{t-B-1})\\
&D\ge \sigma^{2}_{t}( \tilde{\rvs}_{t-L-B}, [\rvu]_{t-L-B+1}^{t-B-1},\rvu_{t}) \label{eq:Q3}
\end{align} 
where $\tilde{\rvs}_{t-L-B} = {\rvs}_{t-L-B} +\rve$, where $\rve\sim \mathcal{N}(0, D/(1-D))$,  also satisfies~\eqref{eq:limRate} and~\eqref{eq:limtest}. 
\end{lemma}}
\begin{proof}
See Appendix~\ref{App:Claim2}.
\end{proof}

This completes the proof of Prop~\ref{prop:GM-ME}.

\subsubsection{Numerical Evaluation} We derive the expression for numerically evaluating $\sigma_z^2$. To this end, first note that the estimation error of estimating $\rvs_{t-B-1}$ from $\{\tilde{\rvs}_{t-L-B}, [\rvu]_{t-L-B+1}^{t-B-1}\}$ can be computed as follows.
\begin{align}
&\eta (\sigma^{2}_{z})\triangleq \sigma^{2}_{t-B-1}(\tilde{\rvs}_{t-L-B}, [\rvu]_{t-L-B+1}^{t-B-1}) \notag\\
&= E\left[\rvs^2_{t-B-1}\right] - E\left[\rvs_{t-B-1}{\rvU}\right] \left(E\left[{\rvU}^{T}{\rvU}\right]\right)^{-1} E\left[\rvs_{t-B-1}{\rvU}^{T}\right] \label{eq:MMSE1}\\
& = 1- A_1 (A_2)^{-1} A_1^T
\end{align}  where we define  $${\rvU}\triangleq \begin{bmatrix}  \rvu_{t-B-1} & \rvu_{t-B-2} & \ldots & \rvu_{t-L-B+1} & \tilde{\rvs}_{t-L-B}\end{bmatrix}$$ and $(.)^T$ denotes the transpose operation. Also note that $A_1$ and $A_2$ can be computed as follows.
\begin{align}
&A_1= (1 , \rho , \rho^2 , \cdots , \rho^{L-1})\\
& A_2= \begin{pmatrix}
1+\sigma^2_{z} & \rho &\cdots & \rho^{L-2} & \rho^{L-1}\\
\rho & 1+\sigma^2_{z} &\cdots & \rho^{L-3} & \rho^{L-2}\\
\vdots & \vdots & \ddots & \vdots & \vdots \\
\rho^{L-2}  &\rho^{L-3}  &\cdots & 1+\sigma^2_{z} & \rho\\
\rho^{L-1}  &\rho^{L-2}  &\cdots & \rho &1+\frac{D}{1-D}\\
\end{pmatrix}
\end{align} According to \eqref{eq:Q3} we can write 
\begin{align}
D &= \sigma^{2}_{t}( \tilde{\rvs}_{t-L-B}, [\rvu]_{t-L-B+1}^{t-B-1}, \rvu_{t}) \\
 &=  \sigma^{2}_{t}\left( \hat{\rvs}_{t-B-1}(\tilde{\rvs}_{t-L-B}, [\rvu]_{t-L-B+1}^{t-B-1}), \rvu_{t}\right)\notag\\
& = \left[\frac{1}{\sigma^2_{z}}+ \frac{1}{1-\rho^{2(B+1)}(1-\eta (\sigma^{2}_{z})) }\right]^{-1} \label{eq:solveD}
\end{align} Therefore by solving \eqref{eq:solveD} the expression for $\sigma^{2}_{z}$ can be obtained.  Finally the achievable rate is computed as:
\begin{align}
& R^{+}_{\textrm{GM-ME}}(L,B,D)  = I(\rvs_{t};\rvu_{t} | \tilde{\rvs}_{t-L-B} ,  [\rvu]_{t-L-B+1}^{t-B-1})\notag\\
 &  =  h(\rvs_{t} | \tilde{\rvs}_{t-L-B} ,  [\rvu]_{t-L-B+1}^{t-B-1}) - h(\rvs_{t} | \tilde{\rvs}_{t-L-B} ,  [\rvu]_{t-L-B+1}^{t-B-1}, \rvu_{t})\notag\\
 & = h\left(\rvs_{t} | \hat{\rvs}_{t-B-1}(\tilde{\rvs}_{t-L-B}, [\rvu]_{t-L-B+1}^{t-B-1})\right) - \frac{1}{2} \log( 2 \pi e D)\notag\\
 & = \frac{1}{2} \log\left( 2 \pi e  \left(1-\rho^{2(B+1)}(1-\eta (\sigma^{2}_{z})) \right)\right) - \frac{1}{2} \log( 2 \pi e D)\notag\\
 & = \frac{1}{2} \log\left(\frac{ 1-\rho^{2(B+1)}(1-\eta (\sigma^{2}_{z}))}{D}\right).\label{eq:repD}
\end{align}  

\section{High Resolution Asymptotic}
\label{sec:HR}

We investigate the behavior of the lossy rate-recovery functions for Gauss-Markov sources for single and multiple burst erasure channel models, i.e. $R_{\textrm{GM-SE}}(B,D)$ and $R_{\textrm{GM-ME}}(L,B,D)$, in the high resolution regime and establish Corollary~\ref{corol:HR}. The following inequalities can be readily verified.
\begin{multline}
R^{-}_{\textrm{GM-SE}}(B,D) \le R_{\textrm{GM-SE}}(B,D) \le  \\ R_{\textrm{GM-ME}}(L,B,D) \le R^{+}_{\textrm{GM-ME}}(L,B,D) \label{eq:HRineqs}
\end{multline} The first and the last inequalities in \eqref{eq:HRineqs} are by definition and the second inequality follows from the fact that the rate achievable for multiple erasure model is also achievable for single burst erasure as the decoder can simply ignore the available codewords in reconstructing the source sequences.  According to \eqref{eq:HRineqs}, it suffices to characterize the high resolution limit of $R^{-}_{\textrm{GM-SE}}(B,D)$ and $R^{+}_{\textrm{GM-ME}}(L,B,D)$ in Prop.~\ref{prop:GML} and Prop.~\ref{prop:GM-ME} respectively.

For the lower bound note that as $D \rightarrow 0$ the expression for $\Delta$ in~\eqref{eq:thm-1} satisfies
$$\Delta \triangleq (D\rho^2+1-\rho^{2(B+1)})^2 - 4D\rho^2(1-\rho^{2B}) \rightarrow (1-\rho^{2(B+1)})^2.$$

Upon direct substitution in~\eqref{eq:thm-1} we have that
\begin{align}
\lim_{D\to 0} \left\{R^{-}_{\textrm{GM-SE}}(B,D) - \frac{1}{2} \log \left(\frac{1-\rho^{2(B+1)}}{D}\right)\right\} = 0, \label{hrl}
\end{align} as required.

To establish the upper bound note that
according to Prop.~\ref{prop:GM-ME} we can write 
\begin{align}
&R^{+}_{\textrm{GM-ME}}(L,B,D) \notag\\
&= I(\rvs_t ;\rvu_t | \tilde{\rvs}_{t-L-B}, [\rvu]_{t-L-B+1}^{t-B-1})\notag\\
&= h(\rvs_t | \tilde{\rvs}_{t-L-B}, [\rvu]_{t-L-B+1}^{t-B-1}) - h(\rvs_t | \tilde{\rvs}_{t-L-B}, [\rvu]_{t-L-B+1}^{t-B-1}, \rvu_t)\notag\\
&= h(\rvs_t | \tilde{\rvs}_{t-L-B}, [\rvu]_{t-L-B+1}^{t-B-1}) - \frac{1}{2} \log (2 \pi e D) \label{eq:limcal}
\end{align} 
where the last term follows from the definition of $\tilde{\rvs}_{t-L-B}$ in Prop.~\ref{prop:GM-ME}. Also we have 
\begin{align}
h(\rvs_{t}|\rvs_{t-B-1}) \le  h(\rvs_t | \tilde{\rvs}_{t-L-B}, [\rvu]_{t-L-B+1}^{t-B-1})  \le h(\rvs_t | \rvu_{t-B-1}) \label{sandwich1}
\end{align} where the left hand side inequality in \eqref{sandwich1} follows from the following Markov property,
\begin{align}
 \{ \tilde{\rvs}_{t-L-B}, [\rvu]_{t-L-B+1}^{t-B-1}\}  \rightarrow \rvs_{t-B-1} \rightarrow \rvs_{t}
\end{align} and the fact that conditioning reduces the differential entropy. Also, the right hand side inequality in \eqref{sandwich1} follows from the latter fact. By computing the upper and lower bounds in \eqref{sandwich1} we have 
\begin{multline}
\frac{1}{2}\log\left(2 \pi e (1-\rho^{2(B+1)})\right) \le \\ h(\rvs_t | \tilde{\rvs}_{t-L-B}, [\rvu]_{t-L-B+1}^{t-B-1})  \le \\ \frac{1}{2}\log\left(2 \pi e \left(1-\frac{\rho^{2(B+1)}}{1+\sigma^2_{z}}\right)\right) \label{sandwich2}
\end{multline}
Now note that 
\begin{align}
D\ge \sigma^2_{t}(\hat{\rvs}_{t-L-B}, [\rvu]_{t-L-B+1}^{t-B-1}, \rvu_{t}) &\ge \sigma^{2}_{t}(\rvu_{t}, \rvs_{t-1})\label{hr1g}\\
&=\left(\frac{1}{\sigma^{2}_{z}} + \frac{1}{1-\rho^2}\right)^{-1}.
\end{align} which equivalently shows that if $D\rightarrow 0$ we have that $\sigma_z^2 \rightarrow 0$. By computing the limit of the upper and lower bounds in  \eqref{sandwich2} as $D\to 0$, we can see that 
\begin{align}
&\lim _{D\to 0}  \left\{h(\rvs_t | \tilde{\rvs}_{t-L-B}, [\rvu]_{t-L-B+1}^{t-B-1}) \right.\notag\\
&\left.\quad- \frac{1}{2}\log \left(2 \pi e (1-\rho^{2(B+1)})\right)\right\} =0 \label{sandwich3}
\end{align} Finally \eqref{sandwich3} and   \eqref{eq:limcal} results in 
\begin{align}
\lim_{D\to 0} \left\{R^{+}_{\textrm{GM-ME}}(L,B,D) - \frac{1}{2}\log\left(\frac{1-\rho^{2(B+1)}}{D}\right)\right\} =0 \label{hr2}
\end{align} as required. Equations \eqref{hrl}, \eqref{hr2} and \eqref{eq:HRineqs} establishes the results of Corollary~\ref{corol:HR}.

%
%
%
%
%

\section{Independent Gaussian Sources with Sliding Window Recovery: Proof of Theorem~\ref{thm:gauss-rate}}
\label{sec:Gauss}
In this section we study the memoryless Gaussian source model discussed in Section~\ref{subsec:GSW}.
The source sequences are drawn i.i.d.\ both in spatial and temporal dimension according to a unit-variance, zero-mean, Gaussian distribution $\mathcal{N}(0,1)$. The rate-$R$ causal encoder sequentially compresses the source sequences and sends the codewords through the burst erasure channel. The channel erases a single burst of maximum length  $B$ and  perfectly reveals the rest of the packets to the decoder. The decoder at each time $i$ reconstructs $K+1$ past source sequences, i.e. $(\rvs^n_{i}, \rvs^n_{i-1}, \ldots, \rvs^n_{i-K})$ within a vector distortion measure $\bd = (d_0, \ldots, d_{K})$. More recent source sequences are required to be reconstructed within less distortion, i.e.  $d_0 \le d_1\le \ldots\le  d_{K}$. The decoder however is not interested in reconstructing the source sequences during the error propagation window, i.e. during the burst erasure and a window of length $W$ after the burst erasure ends. 

For this setup, we establish the rate-recovery function stated in Theorem~\ref{thm:gauss-rate}. We do this by presenting the coding scheme in Section~\ref{subsec:Gauss-Coding} and the converse in Section~\ref{subsec:G-Converse}. We also study some baseline schemes and compare their performance with the rate-recovery function at the end of this section.  

\begin{remark}
Our coding scheme in this section builds upon the technique introduced in \cite{etezadiKhistiDCC:12} for lossless recovery of deterministic sources. The example involving deterministic sources in  \cite{etezadiKhistiDCC:12} established that the lower bound in Theorem~\ref{thm:genUB_LB} can be attained for a certain class of deterministic sources. The binning based scheme is  suboptimal in general. The present paper does not include this example, but the reader is encouraged to see~\cite{etezadiKhistiDCC:12}.
\end{remark}

\subsection{Sufficiency of $K=B+W$}
\label{subsec:Sufficiency}
In our analysis we only consider the case $K=B+W$. The coding scheme can be easily extended to a general $K$ as follows. If $K<B+W$, we can assume that the decoder, instead of recovering the source $\rvbt_{i}=(\rvs_{i}, \rvs_{i-1}, \ldots, \rvs_{i-K})^{T}$ at time $i$ within distortion $\bd$,  aims to recover the source $\rvbt'_{i}=(\rvs_{i}, ..., \rvs_{i-K'})^{T}$ within distortion $\bd'$ where $K'=B+W$ and 
\begin{align}
d'_{j}=\begin{cases}
d_{j} &\mbox{for } j\in\{0, 1, ..., K\}\\
1&\mbox{for } j\in\{K+1, ..., K'\},
\end{cases}
\end{align}
and thus this case is a special case of ${K=B+W}$. Note that  layers $K+1,\ldots, K'$ require zero rate as the source sequences have unit variance.

If ${K>B+W}$, for each $j\in\{B+W+1, \ldots, K\}$ the decoder is required to reconstruct $\rvs^n_{i-j}$ within distortion $d_{j}$. However we note the rate associated with these layers is again zero. In particular there are two possibilities during the recovery at time $i$. Either, $\hat{\rvbt}_{i-1}^n$ or, $\hat{\rvbt}_{i-B-W-1}^n$ are guaranteed to have been reconstructed. In the former case
$\{\hat{\rvs}^n_{i-j}\}_{d_{j-1}}$ is\footnote{The notation $\{\hat{\rvs}_{i}^n\}_{d}$ indicates the reconstruction of $\hat{\rvs}_{i}^n$ within average distortion $d$.} available from time $i-1$ and $d_{j-1}\le d_{j}$. In the latter case
$\{\hat{\rvs}^n_{i-j}\}_{d_{j-W-B-1}}$ is available from time $i-B-W-1$ and again $d_{j-W-B-1}\le d_{j}$.
 Thus the reconstruction of any layer $j\ge B+W$  does not require any additional rate and it again suffices to assume ${K=B+W}$.

\subsection{Coding Scheme}
\label{subsec:Gauss-Coding}
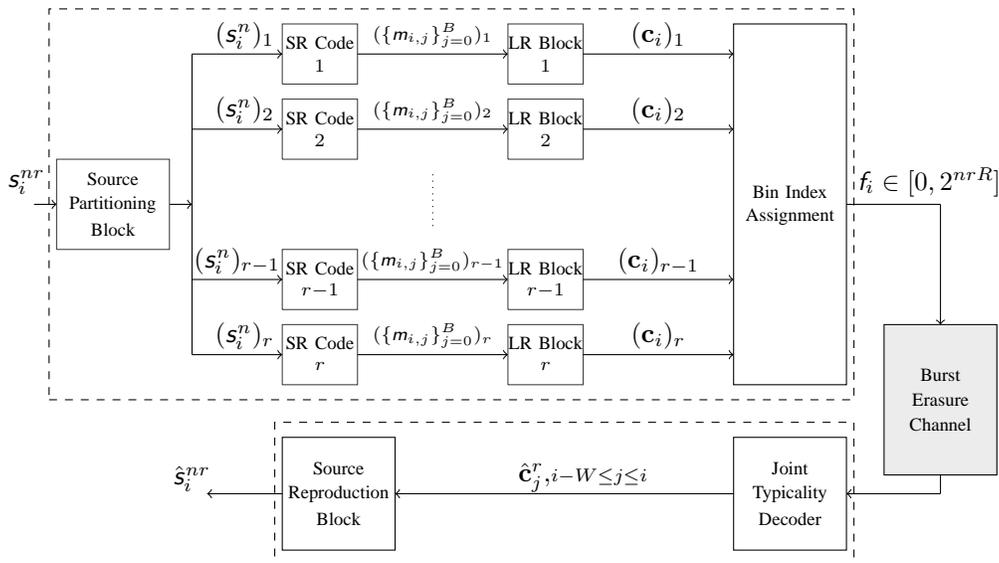
\begin{figure*}
\vspace{1em}
\begin{center}
%
%
%
%

%
\begin{tikzpicture}
\draw [->] (-1.5,-2.4) -- (-1.2,-2.4);
\draw [->] (-3.3,-2.4) -- (-3,-2.4);
\draw [white] (-3.4,-2.1) -- (-3.4,-2.1) node {$\color{black}\rvs_{i}^{nr}$};
\draw [white] (-1.2,-6.05) -- (-1.2,-6.05) node {$\color{black}\hat{\rvs}_{i}^{nr}$};
\fill [fill=white, draw=gray!50!black] (-3,-3) -- (-3,-1.8) -- (-1.5,-1.8) -- (-1.5,-3) -- (-3,-3);
\draw [white] (-2.25,-2.4+1/3) -- (-2.25,-2.4 + 1/3) node {$\color{black}\textrm{\scriptsize{Source}}$};
\draw [white] (-2.25,-2.4) -- (-2.25,-2.4) node {$\color{black}\textrm{\scriptsize{Partitioning}}$};
\draw [white] (-2.25,-2.4-1/3) -- (-2.25,-2.4 - 1/3) node {$\color{black}\textrm{\scriptsize{Block}}$};

\draw [-](6,-4.8) rectangle (7.5,0);
\draw [white] (6.75,-2.4+1/6) -- (6.75,-2.4+1/6) node {$\color{black}\textrm{\scriptsize{Bin Index}}$};
\draw [white] (6.75,-2.4-1/6) -- (6.75,-2.4 - 1/6) node {$\color{black}\textrm{\scriptsize{Assignment}}$};

\draw [-] (7.5,-2.4) -- (8.75,-2.4);
\draw [white] (8.6,-2.1) -- (8.6,-2.1) node {$\color{black}\rvf_{i}\in[0,2^{nrR}]$};

\draw [->] (8.75,-2.4) -- (8.75,-4);
\draw [fill= gray!15!white](8,-6) rectangle (9.5,-4);
\draw [white] (6.75,-6.25+1/3) -- (6.75,-6.25 + 1/3) node {$\color{black}\textrm{\scriptsize{Joint}}$};
\draw [white] (6.75,-6.25) -- (6.75,-6.25) node {$\color{black}\textrm{\scriptsize{Typicality}}$};
\draw [white] (6.75,-6.25-1/3) -- (6.75,-6.25 - 1/3) node {$\color{black}\textrm{\scriptsize{Decoder}}$};

\draw [white] (4,-6.03) -- (4,-6.03) node {$\color{black}\hat{\bc}^r_j, \scriptstyle{i-W\le j\le i}$};

\draw [dashed](-3.1,-5) rectangle (7.6,.2);

\draw [dashed](-0.1,-7.2) rectangle (7.6,-5.3);

\draw [-](6,-7) rectangle (7.5,-5.5);
\draw [white] (8.75,-5+1/3) -- (8.75,-5 + 1/3) node {$\color{black}\textrm{\scriptsize{Burst}}$};
\draw [white] (8.75,-5) -- (8.75,-5) node {$\color{black}\textrm{\scriptsize{Erasure}}$};
\draw [white] (8.75,-5-1/3) -- (8.75,-5 - 1/3) node {$\color{black}\textrm{\scriptsize{Channel}}$};

\draw [-] (8.75,-6) -- (8.75,-6.25);
\draw [->] (8.75,-6.25) -- (7.5,-6.25);

\draw [->] (6,-6.25) -- (1.5,-6.25);
\draw [<-] (-1,-6.25) --(0,-6.25);


\draw [-](0,-7) rectangle (1.5,-5.5);
\draw [white] (0.75,-6.25+1/3) -- (.75,-6.25 + 1/3) node {$\color{black}\textrm{\scriptsize{Source}}$};
\draw [white] (.75,-6.25) -- (.75,-6.25) node {$\color{black}\textrm{\scriptsize{Reproduction}}$};
\draw [white] (.75,-6.25-1/3) -- (.75,-6.25 - 1/3) node {$\color{black}\textrm{\scriptsize{Block}}$};

\def \a {0}{
\fill [fill=white, draw=gray!50!black] (0,-.8-\a) -- (0,-\a) -- (1,-\a) -- (1,-.8-\a) -- (0,-.8-\a);
\draw [white] (.5,-\a-0.25) -- (.5,-\a-0.25) node {$\color{black}\textrm{\scriptsize{SR Code}}$};
\draw [white] (.5,-\a-0.55) -- (.5,-\a-0.55) node {$\color{black}\scriptstyle{1}$};
\draw [->] (-1.2,-.4-\a) -- (0,-.4-\a);
\draw [->] (1,-.4-\a) -- (3,-.4-\a);
\fill [fill=white, draw=gray!50!black] (3,-.8-\a) -- (3,-\a) -- (4,-\a) -- (4,-.8-\a) -- (3,-.8-\a);
\draw [white] (3.5,-\a-0.25) -- (3.5,-\a-0.25) node {$\color{black}\textrm{\scriptsize{LR Block}}$};
\draw [white] (3.5,-\a-0.55) -- (3.5,-\a-0.55) node {$\color{black}\scriptstyle{1}$};
\draw [->] (4,-.4-\a) -- (6,-.4-\a);
\draw [white] (-.5,-\a-0.15) -- (-.5,-\a-0.15) node {$\color{black}(\rvs^{n}_{i})_{1}$};
\draw [white] (5,-\a-0.15) -- (5,-\a-0.15) node {$\color{black}(\bc_{i})_{1}$};
\draw [white] (2,-\a-0.15) -- (2,-\a-0.15) node {$\color{black}\scriptstyle(\{\rvm_{i,j}\}_{j=0}^{B})_{1}$};
}
\def \a {1}{
\fill [fill=white, draw=gray!50!black] (0,-.8-\a) -- (0,-\a) -- (1,-\a) -- (1,-.8-\a) -- (0,-.8-\a);
\draw [white] (.5,-\a-0.25) -- (.5,-\a-0.25) node {$\color{black}\textrm{\scriptsize{SR Code}}$};
\draw [white] (.5,-\a-0.55) -- (.5,-\a-0.55) node {$\color{black}\scriptstyle{2}$};
\draw [->] (-1.2,-.4-\a) -- (0,-.4-\a);
\draw [->] (1,-.4-\a) -- (3,-.4-\a);
\fill [fill=white, draw=gray!50!black] (3,-.8-\a) -- (3,-\a) -- (4,-\a) -- (4,-.8-\a) -- (3,-.8-\a);
\draw [white] (3.5,-\a-0.25) -- (3.5,-\a-0.25) node {$\color{black}\textrm{\scriptsize{LR Block}}$};
\draw [white] (3.5,-\a-0.55) -- (3.5,-\a-0.55) node {$\color{black}\scriptstyle{2}$};
\draw [->] (4,-.4-\a) -- (6,-.4-\a);
\draw [white] (-.5,-\a-0.15) -- (-.5,-\a-0.15) node {$\color{black}(\rvs^{n}_{i})_{2}$};
\draw [white] (5,-\a-0.15) -- (5,-\a-0.15) node {$\color{black}(\bc_{i})_{2}$};
\draw [white] (2,-\a-0.15) -- (2,-\a-0.15) node {$\color{black}\scriptstyle(\{\rvm_{i,j}\}_{j=0}^{B})_{2}$};
}

\def \a {3}{
\fill [fill=white, draw=gray!50!black] (0,-.8-\a) -- (0,-\a) -- (1,-\a) -- (1,-.8-\a) -- (0,-.8-\a);
\draw [white] (.5,-\a-0.25) -- (.5,-\a-0.25) node {$\color{black}\textrm{\scriptsize{SR Code}}$};
\draw [white] (.5,-\a-0.55) -- (.5,-\a-0.55) node {$\color{black}\scriptstyle{r-1}$};
\draw [->] (-1.2,-.4-\a) -- (0,-.4-\a);
\draw [->] (1,-.4-\a) -- (3,-.4-\a);
\fill [fill=white, draw=gray!50!black] (3,-.8-\a) -- (3,-\a) -- (4,-\a) -- (4,-.8-\a) -- (3,-.8-\a);
\draw [white] (3.5,-\a-0.25) -- (3.5,-\a-0.25) node {$\color{black}\textrm{\scriptsize{LR Block}}$};
\draw [white] (3.5,-\a-0.55) -- (3.5,-\a-0.55) node {$\color{black}\scriptstyle{r-1}$};
\draw [->] (4,-.4-\a) -- (6,-.4-\a);
\draw [white] (-.6,-\a-0.15) -- (-.6,-\a-0.15) node {$\color{black}(\rvs^{n}_{i})_{r-1}$};
\draw [white] (5,-\a-0.15) -- (5,-\a-0.15) node {$\color{black}(\bc_{i})_{r-1}$};
\draw [white] (2,-\a-0.15) -- (2,-\a-0.15) node {$\color{black}\scriptstyle(\{\rvm_{i,j}\}_{j=0}^{B})_{r-1}$};
}
\def \a {4}{
\fill [fill=white, draw=gray!50!black] (0,-.8-\a) -- (0,-\a) -- (1,-\a) -- (1,-.8-\a) -- (0,-.8-\a);
\draw [white] (.5,-\a-0.25) -- (.5,-\a-0.25) node {$\color{black}\textrm{\scriptsize{SR Code}}$};
\draw [white] (.5,-\a-0.55) -- (.5,-\a-0.55) node {$\color{black}\scriptstyle{r}$};
\draw [->] (-1.2,-.4-\a) -- (0,-.4-\a);
\draw [->] (1,-.4-\a) -- (3,-.4-\a);
\fill [fill=white, draw=gray!50!black] (3,-.8-\a) -- (3,-\a) -- (4,-\a) -- (4,-.8-\a) -- (3,-.8-\a);
\draw [white] (3.5,-\a-0.25) -- (3.5,-\a-0.25) node {$\color{black}\textrm{\scriptsize{LR Block}}$};
\draw [white] (3.5,-\a-0.55) -- (3.5,-\a-0.55) node {$\color{black}\scriptstyle{r}$};
\draw [->] (4,-.4-\a) -- (6,-.4-\a);
\draw [white] (-.5,-\a-0.15) -- (-.5,-\a-0.15) node {$\color{black}(\rvs^{n}_{i})_{r}$};
\draw [white] (5,-\a-0.15) -- (5,-\a-0.15) node {$\color{black}(\bc_{i})_{r}$};
\draw [white] (2,-\a-0.15) -- (2,-\a-0.15) node {$\color{black}\scriptstyle(\{\rvm_{i,j}\}_{j=0}^{B})_{r}$};
}

\draw [-] (-1.2,-.4) -- (-1.2,-4.4);
\draw [dotted] (2,-2) -- (2,-2.7);
%

\end{tikzpicture}
%
\caption{Schematic of encoder and decoder for i.i.d.\ Gaussian with sliding window recovery constraint. $\textrm{SR}$ and $\textrm{LR}$ indicate successive refinement and layer rearrangement (Sections \ref{subsec:sr} and \ref{subsec:lr}), respectively.}
\label{fig:Ach-GS}
\end{center}
\end{figure*}

Throughout our analysis, we assume the source sequences are of length $n\cdot r$ where both $n$ and $r$ will be assumed to be arbitrarily large. The block diagram of the scheme is shown in  Fig.~\ref{fig:Ach-GS}.  We partition $\rvs_i^{n\cdot r}$ into $r$ blocks each consisting of $n$ symbols $(\rvs_i^n)_l$. 
We then apply a successive refinement quantization codebook to each such block to generate ${B+1}$ indices $\left(\{\rvm_{i,j}\}_{j=0}^B\right)_l$ as discussed in section~\ref{subsec:sr}. Thereafter these indices are carefully rearranged in time to generate $(\bc_i)_l$ as discussed in Section~\ref{subsec:lr}.  At each time we thus have a length $r$ sequence  $\bc_i^r \defeq \left\{(\bc_i)_1,\ldots, (\bc_i)_r\right\}$   We transmit the  bin index of each  sequence  over the channel as in Section~\ref{sec:UBLB}. At the receiver the sequence $\hat{\bc}_i^r$ is first reconstructed by the inner decoder. Thereafter upon rearranging the refinement layers in each packet, the required reconstruction sequences are produced.  We provide the details of the encoding and decoding below. 

\subsubsection{Successive Refinement (SR) Encoder}
\label{subsec:sr}
\begin{figure}
\vspace{1em}
\begin{center}
\noindent
\resizebox{3in}{1.7in}{
\begin{tikzpicture}
\draw [white] (-1.2,-2.4) -- (-1.2,-2.4) node {$\color{black}\bs_{i}^{n}$};
\draw [-] (-1,-2.4) -- (-.7,-2.4);
\foreach \a in {0,1,3,4}{
\fill [fill=white, draw=gray!50!black] (0,-.8-\a) -- (0,-\a) -- (1,-\a) -- (1,-.8-\a) -- (0,-.8-\a);
\draw [white] (.5,-\a-0.25) -- (.5,-\a-0.25) node {$\color{black}\textrm{\scriptsize{Encoder}}$};
\draw [->] (-.7,-.4-\a) -- (0,-.4-\a);
\draw [->] (1,-.4-\a) -- (4.5,-.4-\a);
\fill [fill=white, draw=gray!50!black] (4.5,-.8-\a) -- (4.5,-\a) -- (5.5,-\a) -- (5.5,-.8-\a) -- (4.5,-.8-\a);
\draw [white] (5,-\a-0.25) -- (5,-\a-0.25) node {$\color{black}\textrm{\scriptsize{Decoder}}$};
\draw [->] (5.5,-.4-\a) -- (6,-.4-\a);
}

\foreach \a in {0}{
\draw [white] (.5,-\a-0.55) -- (.5,-\a-0.55) node {$\color{black}\scriptstyle{B}$};
\draw [white] (5,-\a-0.55) -- (5,-\a-0.55) node {$\color{black}\scriptstyle{B}$};
\draw [white] (6.9,-\a-0.4) -- (6.9,-\a-0.4) node {$\color{black}\{\hat{\bs}_{i}^n\}_{d_{B+W}}$};
\draw [white] (1.5,-\a-0.15) -- (1.5,-\a-0.15) node {$\color{black}\scriptstyle\rvm_{i,B}$};
\draw [white] (3.5,-\a-0.15) -- (3.5,-\a-0.15) node {$\color{black}\scriptstyle M_{i,B}$};
}

\def \a {1}{
\draw [white] (.5,-\a-0.55) -- (.5,-\a-0.55) node {$\color{black}\scriptstyle{B-1}$};
\draw [white] (5,-\a-0.55) -- (5,-\a-0.55) node {$\color{black}\scriptstyle{B-1}$};
\draw [white] (6.9,-\a-0.4) -- (6.9,-\a-0.4) node {$\color{black}\{\hat{\bs}_{i}^n\}_{d_{B+W-1}}$};
\draw [white] (1.5,-\a-0.15) -- (1.5,-\a-0.15) node {$\color{black}\scriptstyle\rvm_{i,B-1}$};
\draw [white] (3.5,-\a-0.12) -- (3.5,-\a-0.12) node {$\color{black}\scriptstyle M_{i,B-1}$};
\draw (3.5,-\a-0.35) ellipse (.05cm and .1cm);
}

\def \a {3}{
\draw [white] (.5,-\a-0.55) -- (.5,-\a-0.55) node {$\color{black}\scriptstyle{1}$};
\draw [white] (5,-\a-0.55) -- (5,-\a-0.55) node {$\color{black}\scriptstyle{1}$};
\draw [white] (6.9,-\a-0.4) -- (6.9,-\a-0.4) node {$\color{black}\{\hat{\bs}_{i}^n\}_{d_{W+1}}$};
\draw [white] (1.5,-\a-0.15) -- (1.5,-\a-0.15) node {$\color{black}\scriptstyle\rvm_{i,1}$};
\draw [white] (3.5,-\a-0.04) -- (3.5,-\a-0.04) node {$\color{black}\scriptstyle M_{i,1}$};
\draw (3.5,-\a-0.3) ellipse (.05cm and .15cm);
}

\def \a {4}{
\draw [white] (.5,-\a-0.55) -- (.5,-\a-0.55) node {$\color{black}\scriptstyle{0}$};
\draw [white] (5,-\a-0.55) -- (5,-\a-0.55) node {$\color{black}\scriptstyle{0}$};
\draw [white] (6.9,-\a-0.4) -- (6.9,-\a-0.4) node {$\color{black}\{\hat{\bs}_{i}^n\}_{d_{0}}$};
\draw [white] (1.5,-\a-0.15) -- (1.5,-\a-0.15) node {$\color{black}\scriptstyle\rvm_{i,0}$};
\draw [white] (3.5,-\a+0.05) -- (3.5,-\a+0.05) node {$\color{black}\scriptstyle M_{i,0}$};
\draw (3.5,-\a-0.25) ellipse (.05cm and .2cm);
}

\draw [dotted] (.5,-2-.6)--(.5,-2-.2);
\draw [dotted] (3.5,-2-.6)--(3.5,-2-.2);

\draw [dashed] (-.9, -5) rectangle (4.2,.2);

\draw [-] (-.7,-.4) -- (-.7,-4.4);
\draw [-] (2.8,-.4) -- (2.8,-4.1);
\draw [-] (2.6,-1.4) -- (2.6,-4.2);
\draw [-] (2.4,-3.4) -- (2.4,-4.3);

\draw [->] (2.8,-1.3) -- (4.5,-1.3);
\draw [->] (2.8,-3.2) -- (4.5,-3.2);
\draw [->] (2.6,-3.3) -- (4.5,-3.3);
\draw [->] (2.8,-4.1) -- (4.5,-4.1);
\draw [->] (2.6,-4.2) -- (4.5,-4.2);
\draw [->] (2.4,-4.3) -- (4.5,-4.3);

\end{tikzpicture}}
\caption{($B+1$)-layer coding scheme based on successive refinement (SR). Note that for each $k\in[0,B]$, $\rvm_{i,k}$ is of rate $\tilde{R}_{k}$ and $M_{i,k}$ is of rate $R_{k}$. The dashed box represents the SR code.}
\label{Suc-Ref}
\end{center}
\end{figure}
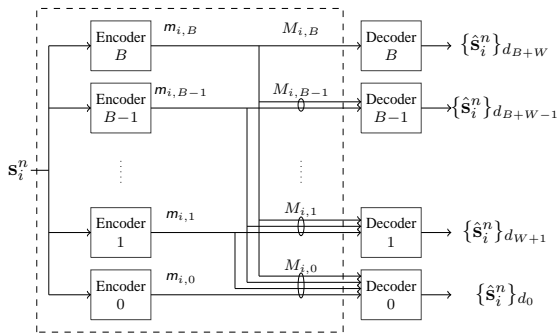
 
The encoder at time $i$, first partitions the source sequence $\rvs_{i}^{nr}$ into $r$ source sequences $(\rvs_{i}^{n})_{l}, l\in[1,r]$. As shown in Fig.~\ref{Suc-Ref}, we encode each source signal $(\rvs^n_{i})_{l}$ using a $(B+1)$-layer successive refinement codebook~\cite{equitzCover:91, rimoldi:94} to generate $(B+1)$ codewords whose indices are given by $\{(\rvm_{i,0})_{l}, (\rvm_{i,1})_{l}, \ldots,  (\rvm_{i,B})_{l}\}$ where $(\rvm_{i,j})_{l}\in \{1,2,\ldots,2^{n\tilde{R}_j}\}$ for $j\in\{0,1,\dots,B\}$ and
\begin{align}
&\tilde{R}_{j}=\begin{cases}
\frac{1}{2}\log(\frac{d_{W+1}}{d_{0}})& \mbox{for } j=0\\
\frac{1}{2}\log(\frac{d_{W+j+1}}{d_{W+j}})& \mbox{for } j\in\{1,2,\ldots,B-1\}\\
\frac{1}{2}\log(\frac{1}{d_{W+B}})& \mbox{for }  j=B,
\end{cases}
\end{align}

The $j$-th layer uses indices 
\begin{align}
(M_{i,j})_{l} \defeq \{(\rvm_{i,j})_{l},\ldots,(\rvm_{i,B})_{l}\}
\label{eq:Mdef}
\end{align} for reproduction and the associated rate with layer $j$ is given by:
\begin{align}
R_{j}=\begin{cases}
\sum_{k=0}^{B}\tilde{R}_{k}= \frac{1}{2}\log(\frac{1}{d_{0}})& \mbox{for } j=0\\
\sum_{k=j}^{B}\tilde{R}_{k}= \frac{1}{2}\log(\frac{1}{d_{W+j}})& \mbox{for } j\in\{1,2,\ldots,B\},
\end{cases}
\label{eq:Rdef}
\end{align}
and the corresponding distortion associated with layer $j$ equals $d_{0}$ for $j=0$ and $d_{W+j}$ for $j\in\{1,2,\ldots,B\}$.

From Fig.~\ref{Suc-Ref} it is clear that for any $i$ and $j\in \{0, \ldots, B\}$, the $j$-th layer $M_{i,j}$ is a subset of $j-1$-th layer $M_{i,j-1}$, i.e. $M_{i,j}\subseteq M_{i,j-1}$.

\subsubsection{Layer Rearrangement (LR) and Binning}
\label{subsec:lr}
In this stage the encoder rearranges the outputs of the SR blocks associated with different layers to produce an auxiliary set of sequences as follows\footnote{We suppress the index $l$ in~\eqref{eq:Mdef} for compactness.}.
\begin{align}
{\bc}_{i}\defeq\!
\begin{pmatrix}
{M_{i,0}}\\
{M_{i-1,1}}\\
{M_{i-2,2}}\\
\vdots\\
{M_{i-B,B}}
\end{pmatrix} \label{Slide1}
\end{align}

In the definition of~\eqref{Slide1} we note that $M_{i,0}$ consists of all the refinement layers associated with the source sequence at time $i$. It can be viewed as the ``innovation symbol" since it is independent of all past  symbols. It results in a distortion of $d_0$. The symbol $M_{i-1,1}$ consists of all refinement layers of the source sequence at time $i-1$, except the last layer and results in a distortion of $d_1$.  Recall that $M_{i-1,1} \subseteq M_{i-1,0}$. In a similar fashion $M_{i-B,B}$ is associated with the source sequence at time $i-B$ and results in a distortion of $d_B$.  Fig.~\ref{Layer} illustrates a schematic of these auxiliary codewords.

Note that as shown in  Fig.~\ref{Suc-Ref} the encoder at each time generates $r$ independent auxiliary codewords $(\bc_{i})_1,\ldots, (\bc_i)_r$.  Let $\bc_{i}^r$ be the set of all $r$ codewords.  In the final step, the encoder generates $\rvf_{i}$, the bin index associated with the codewords $\bc^r_{i}$ and transmit this through the channel. The bin indices are randomly and independently assigned to all the codewords beforehand and are revealed to both encoder and decoder. 

\begin{figure*}
\vspace{1em}
\begin{center}
\begin{tikzpicture}
\def \a {0}{
\draw [white] (\a,0) -- (\a,0) node {$\color{black}\scriptstyle{\bc_{i}}$};
\draw [fill=red!10!white]  (\a-.5,-.2-.5)  rectangle (\a+.5, .2-.5);
\draw [-] (\a-.5,-.25) -- (\a+.5,-.25);
\draw [white] (\a,-.5) -- (\a,-.5) node {$\color{black}\scriptstyle{M_{i,0}}$};
\draw [white] (\a,-1) -- (\a,-1) node {$\color{black}\scriptstyle{M_{i-1,1}}$};
\draw [white] (\a,-1.5) -- (\a,-1.5) node {$\color{black}\scriptstyle{M_{i-2,2}}$};
\draw [dashed] (\a,-2) -- (\a,-2.4);
\draw [white] (\a,-3) -- (\a,-3) node {$\color{black}\scriptstyle{M_{i-B,B}}$};
}

\def \a {-2}{
\draw [white] (\a,0) -- (\a,0) node {$\color{black}\scriptstyle{\bc_{i-1}}$};
\draw [fill=red!10!white]  (\a-.5,-.2-.5)  rectangle (\a+.5, .2-.5);
\draw [-] (\a-.5,-.25) -- (\a+.5,-.25);
\draw [white] (\a,-.5) -- (\a,-.5) node {$\color{black}\scriptstyle{M_{i-1,0}}$};
\draw [white] (\a,-1) -- (\a,-1) node {$\color{black}\scriptstyle{M_{i-2,1}}$};
\draw [white] (\a,-1.5) -- (\a,-1.5) node {$\color{black}\scriptstyle{M_{i-3,2}}$};
\draw [dashed] (\a,-2) -- (\a,-2.4);
\draw [white] (\a,-3) -- (\a,-3) node {$\color{black}\scriptstyle{M_{i-B-1,B}}$};
}

\def \a {-4}{
\draw [white] (\a,0) -- (\a,0) node {$\color{black}\scriptstyle{\bc_{i-2}}$};
\draw [fill=red!10!white]  (\a-.5,-.2-.5)  rectangle (\a+.5, .2-.5);
\draw [-] (\a-.5,-.25) -- (\a+.5,-.25);
\draw [white] (\a,-.5) -- (\a,-.5) node {$\color{black}\scriptstyle{M_{i-2,0}}$};
\draw [white] (\a,-1) -- (\a,-1) node {$\color{black}\scriptstyle{M_{i-3,1}}$};
\draw [white] (\a,-1.5) -- (\a,-1.5) node {$\color{black}\scriptstyle{M_{i-4,2}}$};
\draw [dashed] (\a,-2) -- (\a,-2.4);
\draw [white] (\a,-3) -- (\a,-3) node {$\color{black}\scriptstyle{M_{i-B-2,B}}$};
}

\def \a {-8}{
\draw [white] (\a,0) -- (\a,0) node {$\color{black}\scriptstyle{\bc_{i-W+1}}$};
\draw [fill=red!10!white]  (\a-.75,-.2-.5)  rectangle (\a+.75, .2-.5);
\draw [-] (\a-.5,-.25) -- (\a+.5,-.25);
\draw [white] (\a,-.5) -- (\a,-.5) node {$\color{black}\scriptstyle{M_{i-W+1,0}}$};
\draw [white] (\a,-1) -- (\a,-1) node {$\color{black}\scriptstyle{M_{i-W,1}}$};
\draw [white] (\a,-1.5) -- (\a,-1.5) node {$\color{black}\scriptstyle{M_{i-W-1,2}}$};
\draw [dashed] (\a,-2) -- (\a,-2.4);
\draw [white] (\a,-3) -- (\a,-3) node {$\color{black}\scriptstyle{M_{i-B-W+1,B}}$};
}

\def \a {-10}{
\draw [white] (\a,0) -- (\a,0) node {$\color{black}\scriptstyle{\bc_{i-W}}$};
\draw [fill=red!10!white]  (\a-.75,-.2-.5)  rectangle (\a+.75, .2-.5);
\draw [fill=red!10!white]  (\a-.75,-.2-1.5)  rectangle (\a+.75, .2-1.5);
\draw [fill=red!10!white]  (\a-.75,-.2-1)  rectangle (\a+.75, .2-1);
\draw [fill=red!10!white]  (\a-.75,-.2-3)  rectangle (\a+.75, .2-3);
\draw [-] (\a-.5,-.25) -- (\a+.5,-.25);
\draw [white] (\a,-.5) -- (\a,-.5) node {$\color{black}\scriptstyle{M_{i-W,0}}$};
\draw [white] (\a,-1) -- (\a,-1) node {$\color{black}\scriptstyle{M_{i-W-1,1}}$};
\draw [white] (\a,-1.5) -- (\a,-1.5) node {$\color{black}\scriptstyle{M_{i-W-2,2}}$};
\draw [dashed] (\a,-2) -- (\a,-2.4);
\draw [white] (\a,-3) -- (\a,-3) node {$\color{black}\scriptstyle{M_{i-B-W,B}}$};
}
\draw [ shift= {(-1,-.75)}, rotate=75] (0,0) circle (.25 and 1.6);
\draw [ shift= {(-2,-1)}, rotate=76] (0,0) circle (.25 and 2.7);
\draw [dashed] (-6-.5,-.25) -- (-6+.5,-.25);
\end{tikzpicture}
\caption{{Schematic of the auxiliary codewords defined in \eqref{Slide1}. The codewords are temporally correlated in a diagonal form depicted using ellipses. In particular, as shown in Fig.~\ref{Suc-Ref}, $M_{i-j,j} \subseteq M_{i-j,j-1}$. Based on this diagonal correlation structure, the codewords depicted in the boxes are sufficient to know all the codewords .}}
\label{Layer}
\end{center}
\end{figure*}
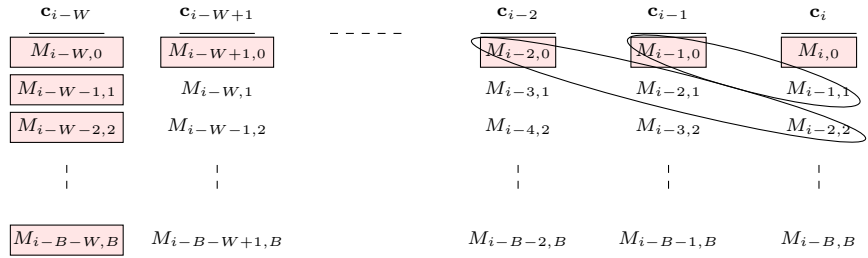

\subsubsection{Decoding and Rate Analysis}
{
To analyze the decoding process, first consider the simple case where the actual codewords $\bc^r_{i}$ defined in \eqref{Slide1}, and not the assigned bin indices, are transmitted through the channel. In this case, whenever the channel packet is not erased by the channel, the decoder has access to the codewords $\bc^r_{i}$. According to the problem setup, at anytime $i$ outside the error propagation window, when the decoder is interested in reconstructing the original source sequences, it has access to the past $W+1$ channel packets, i.e. $(\rvf_{i-W}, \ldots, \rvf_{i-1}, \rvf_{i})$. Therefore, the codewords $(\bc^r_{i}, \bc^r_{i-1},\ldots, \bc^r_{i-W})$ are known to the decoder. Now consider the following claim. 

\begin{claim}
\label{Claim}
The decoder at each time $i$ is able to reconstruct the source sequences within required distortion vector if either the sequences $(\bc^r_{i}, \bc^r_{i-1},\ldots, \bc^r_{i-W})$ or $(\hat{\rvbt}^{n\cdot r}_{i-1}, \bc^{r}_{i})$ is available.
$\hfill\Box$
\end{claim}

\begin{proof}
Fig.~\ref{Layer} shows a schematic of the codewords. First consider the case where $(\bc^r_{i}, \bc^r_{i-1},\ldots, \bc^r_{i-W})$ is known. According to \eqref{Slide1} the decoder also knows $(M^r_{i,0}, M^r_{i-1,0},\ldots, M^r_{i-W,0}).$  Therefore, according to SR structure depicted in Fig.~\ref{Suc-Ref}, the source sequences $(\rvs^{nr}_{i}, \rvs^{nr}_{i-1}, \ldots, \rvs^{nr}_{i-W})$ are each known within distortion $d_{0}$. This satisfies the original distortion constraint as $d_{0}\le d_{k}$ for each $k\in\{1, \ldots, W\}$. In addition, since $\bc_{i-W}$ is known,  according to \eqref{Slide1}, $(M^r_{i-W-1,1}, M^r_{i-W-2,2},\ldots, M^r_{i-B-W,B})$ is known and according to SR structure depicted in Fig.~\ref{Suc-Ref} the source sequences $(\rvs^{nr}_{i-W-1}, \rvs^{nr}_{i-W-2}, \ldots, \rvs^{nr}_{i-B-W})$ are known within distortion $(d_{W+1}, d_{W+2}, \ldots, d_{B+W})$ which satisfies the distortion constraint. Now consider the case where $\hat{\rvbt}_{i-1}^n$ and $\bc^r_{i}$ are available, i.e. $\rvbt^{nr}_{i-1}$ is already reconstructed within the required distortion vector, the decoder is able to reconstruct $\hat{\rvbt}^{nr}_{i}$ from $\hat{\rvbt}^{nr}_{i-1}$ and $\bc^r_{i}$. In particular, from $M^r_{i}$ the source sequence $\rvs^{nr}_{i}$ is reconstructed within distortion $d_0$. Also reconstruction of $\rvs^{nr}_{i-k}$ within distortion $d_{k-1}$ is already available from $\hat{\rvbt}_{i}$ for $k\in [1,B+W]$ which satisfies the distortion constraint as $d_{k-1}\le d_{k}$. 
\end{proof}

Thus we have shown that if actual codewords $\bc^r_{i}$ defined in \eqref{Slide1} are transmitted the required distortion constraints are satisfied. It can be verified from \eqref{Slide1} and~\eqref{eq:Rdef}
that the rate associated with the  $\bc_{i}^r$ is given by 
\begin{align}
R_{\cal C}= \sum_{k=0}^{B} R_{k}  = \frac{1}{2}\log\left(\frac{1}{d_0}\right) + \sum_{j=1}^B \frac{1}{2}\log\left(\frac{1}{d_{W+j}}\right) \label{eq:Rc}
\end{align}
Thus compared to the achievable rate~\eqref{eq:det-rate-Gaussian} in Theorem~\ref{thm:gauss-rate} we are missing the factor of $\frac{1}{W+1}$ in the second term. To reduce the rate, note that, as shown in Fig.~\ref{Layer} and based on definition of the auxiliary codewords in \eqref{Slide1}, there is a strong temporal correlation among the consecutive codewords.  We therefore bin the set of all sequences $\bc_i^r$ into $2^{nrR}$ bins as in Section~\ref{sec:UBLB}. The encoder, upon observing $\bc_{i}^r$, only transmits its bin index $\rvf_{i}$ through the channel.  We next describe the decoder and compute the minimum rate required to reconstruct  $\bc_i^r$.

%
%
Outside the error propagation window, one of the following cases can happen as discussed below. We claim that in either case the decoder is able to reconstruct $\bc_i^r$ as follows.

\begin{itemize}
\item In the first case, the decoder has already recovered  $\bc_{i-1}^{r}$ and attempts to recover $\bc_i^r$ given $(\rvf_i, \bc_{i-1}^r)$. This succeeds with high probability if
{\allowdisplaybreaks{\begin{align}
&nR \ge H(\bc_{i} | \bc_{i-1}) \\
&= H(M_{i,0}, M_{i-1,1}, \ldots, M_{i-B,B}|\bc_{i-1})\label{eq:M-sub-1}\\
&= H(M_{i,0}, M_{i-1,1}, \ldots, M_{i-B,B}| M_{i-1,0}, M_{i-2,1}, \ldots, \notag\\
&\qquad \ldots M_{i-B,B-1}, M_{i-B-1,B})\label{eq:M-sub-2}\\
&= H(M_{i,0}) \label{expl}\\
&= nR_{0} \label{SlideR2}
\end{align}}} where we use~\eqref{Slide1} in~\eqref{eq:M-sub-1} and~\eqref{eq:M-sub-2}, and the fact that layer $j$ is a subset of layer $j-1$ i.e., $M_{i-j, j} \subseteq M_{i-j, j-1}$ in~\eqref{expl}. Thus the reconstruction of $\bc_i^r$ follows since the choice of~\eqref{eq:det-rate-Gaussian} satisfies~\eqref{SlideR2}. Thus according to the second part of Claim~\ref{Claim}, the decoder is able to reconstruct $\hat{\rvbt}^{n\cdot r}_{i}$. 

\item In the second case we assume that the decoder has not yet successfully reconstructed $\bc_{i-1}^r$ but is required to reconstruct $\bc_i^r$. In this case $\bc_i^r$ is the first sequence to be recovered following the end of the error propagation window. Our proposed decoder uses $(\rvf_{i}, \rvf_{i-1},\ldots, \rvf_{i-W})$ to simultaneously reconstruct $( \bc_i^r,\ldots, \bc_{i-W}^r)$.
This succeeds with high probability provided:
\begin{align}
&n(W+1)R \notag\\
&\ge H(\bc_{i-W}, \bc_{i-W+1}, \ldots, \bc_{i})\notag\\
& = H(\bc_{i-W}, M_{i-W+1,0}, M_{i-W+2,0}, \ldots, M_{i,0})\label{Slide2}\\
& = H(\bc_{i-W}) + \sum_{k=1}^{W}H(M_{i-W+k,0})\label{Slide3}\\
&= H(M_{i-W,0},M_{i-W-1,1}, \ldots, M_{i-B-W, B}) \notag\\ &\qquad +\sum_{k=1}^{W}H(M_{i-W+k,0})\notag\\
& = n\sum_{k=1}^{B}R_{k} + n(W+1)R_{0} \label{SlideR1} 
\end{align}
where in \eqref{Slide2} we use the fact that 
the sub-symbols satisfy ${M}_{i,j+1} \subseteq {M}_{i,j}$ as illustrated in Fig.~\ref{Layer}.
In particular, in computing the rate in \eqref{Slide2} all the sub-symbols in $\bc_{i-W}$ and the sub-symbols $M_{j,0}$ for $j\in[i-W+1, i]$ need to be considered.
From \eqref{SlideR2}, \eqref{SlideR1} and \eqref{eq:Rdef}, the rate $R$ is achievable if
\begin{align}
R &\ge R_{0} + \frac{1}{W+1} \sum_{k=1}^{B} R_{k}\\
&=\frac{1}{2} \log\left( \frac{1}{d_0}\right) + \frac{1}{2(W+1)}\sum_{k=1}^{B} \log\left(\frac{1}{d_{W+k}}\right)\label{eq:R_gauss_layered_achiev}.
\end{align}
as required.
Thus, the rate constraint in \eqref{eq:R_gauss_layered_achiev} is sufficient for the decoder to recover the codewords $( \bc_i^r,\ldots, \bc_{i-W}^r)$ right after the error propagation window and to reconstruct $\hat{\rvbt}_{i}^{n\cdot r}$ according to Claim~\ref{Claim}.  
\end{itemize}
}

Thus, the rate constraint in \eqref{eq:R_gauss_layered_achiev} is sufficient for the decoder to succeed in reconstructing the source sequences within required distortion constraints at the anytime $i$ outside error propagation window.
This completes the justification of the upper bound in Theorem~\ref{thm:gauss-rate}.  

\subsection{Converse for Theorem~\ref{thm:gauss-rate}}
\label{subsec:G-Converse}
We need to show that for any sequence of codes that achieve a distortion tuple $(d_0,\ldots, d_{W+B})$ the rate is lower bounded by~\eqref{eq:R_gauss_layered_achiev}.
As in the proof of Theorem~\ref{thm:genUB_LB}, we consider a burst erasure  of length $B$ spanning the time interval $[t-B-W, t-W-1]$. Consider,

\begin{align}
(W+1)n R &\ge H([\rvf]_{t-W}^t)\notag\\
&\ge H([\rvf]_{t-W}^t| [\rvf]_{0}^{t-B-W-1}, \rvs_{-1}^n) \label{eq:Gauss-LB}
\end{align}
where the last step follows from the fact that conditioning reduces entropy. We need to lower bound the entropy term in \eqref{eq:Gauss-LB}. Consider
\begin{align}
&H([\rvf]_{t-W}^{t}|[\rvf]_{0}^{{t-B-W-1}}, \rvs_{-1}^n) \notag \\
&=I([\rvf]_{t-W}^{t};\rvbt_{t}^n|[\rvf]_{0}^{{t-B-W-1}}, \rvs_{-1}^n) + \notag\\ &\qquad H([\rvf]_{t-W}^{t}|[\rvf]_{0}^{{t-B-W-1}}, \rvbt_{t}^n, \rvs_{-1}^n)\\
&=h(\rvbt_{t}^n|[\rvf]_{0}^{{t-B-W-1}}, \rvs_{-1}^n) -h(\rvbt_{t}^n|[\rvf]_{0}^{{t-B-W-1}}, [\rvf]_{t-W}^{t}, \rvs_{-1}^n)  + \notag\\ &\qquad H([\rvf]_{t-W}^{t}|[\rvf]_{0}^{{t-B-W-1}}, \rvbt_{t}^n, \rvs_{-1}^n)\notag\\
&=h(\rvbt_{t}^n) -h(\rvbt_{t}^n|[\rvf]_{0}^{{t-B-W-1}}, [\rvf]_{t-W}^{t}, \rvs_{-1}^n) + \notag\\ &\qquad  H([\rvf]_{t-W}^{t}|[\rvf]_{0}^{{t-B-W-1}}, \rvbt_{t}^n, \rvs_{-1}^n) \label{eq:exp1}
\end{align}
where~\eqref{eq:exp1} follows since ${\rvbt_{t}^n = (\rvs_{t-B-W}^n,\ldots, \rvs_{t}^n)}$ is independent of $([\rvf]_0^{{t-B-W-1}}, \rvs_{-1}^n)$ as the source sequences $\rvs_i^n$ are generated i.i.d.\ . By expanding $\rvbt_{t}^n$ we have that
\begin{align}
h(\rvbt_{t}^n)  &= h(\rvs_{t-B-W}^n, \ldots, \rvs_{t-W-1}^n)  + h(\rvs_{t-W}^n, \ldots, \rvs_{t}^n), \label{eq:rel1}
\end{align}and
\begin{align}
&h(\rvbt_{t}^n|[\rvf]_{0}^{{t-B-W-1}}, [\rvf]_{t-W}^{t}, \rvs_{-1}^n) \notag\\&=h(\rvs_{t-B-W}^n, \ldots, \rvs_{t-W-1}^n|[\rvf]_{0}^{{t-B-W-1}}, [\rvf]_{t-W}^{t}, \rvs_{-1}^n) + \notag \\ &\quad h\left(\rvs_{t-W}^n, \ldots, \rvs_{t}^n|[\rvf]_{0}^{{t-B-W-1}}, [\rvf]_{t-W}^{t}  ,\rvs_{t-B-W}^n,\right.\notag\\
&\left.\quad\quad \quad  \ldots, \rvs_{t-W-1}^n, \rvs_{-1}^n\right) \label{eq:rel2}
\end{align}

We next establish the following claim whose proof is in Appendix~\ref{app:NEWapp}.

\begin{lemma}
\label{claim:converse}
The following two inequalities holds. 
\begin{multline}
h(\rvs_{t-B-W}^n, \ldots, \rvs_{t-W-1}^n)- \\ h(\rvs_{t-B-W}^n, \ldots, \rvs_{t-W-1}^n|[\rvf]_{0}^{{t-B-W-1}}, [\rvf]_{t-W}^{t}, \rvs_{-1}^n)\\
 \ge \sum_{i=1}^{B}\frac{n}{2}\log{(\frac{1}{d_{W+i}})}\label{eq:Gauss_LB_T1}
\end{multline}

\begin{align}
&h(\rvs_{t-W}^n, \ldots, \rvs_{t}^n)- \notag\\
&\quad h\left(\rvs_{t-W}^n, \ldots, \rvs_{t}^n| [\rvf]_{0}^{{t-B-W-1}},[\rvf]_{t-W}^{t},\rvs_{t-B-W}^n, \right.\notag\\
&\left.\quad \quad \ldots, \rvs_{t-W-1}^n, \rvs_{-1}^n\right)\notag\\ 
&\quad+ H([\rvf]_{t-W}^{t}|[\rvf]_{0}^{{t-B-W-1}}, \rvbt_{t}^n, \rvs_{-1}^n) \notag\\
&\quad \quad \ge \frac{n(W+1)}{2}\log(\frac{1}{d_0})
\label{eq:Gauss_LB_T2}
\end{align}
$\hfill\Box$
\end{lemma}

\begin{proof}
See Appendix~\ref{app:NEWapp}.
\end{proof}

From~\eqref{eq:exp1},~\eqref{eq:rel1},~\eqref{eq:rel2},~\eqref{eq:Gauss_LB_T1} and~\eqref{eq:Gauss_LB_T2}, we can write
\begin{multline}
H\left([\rvf]_{t-W}^{t} | [\rvf]_{0}^{t-B-W-1}, \rvs_{-1}^n\right)  \ge \\ \frac{n}{2}\sum_{i=1}^{B} \log \left(\frac{1}{ d_{W+i}}\right) + \frac{n(W+1)}{2}\log\left(\frac{1}{d_{0}}\right) \label{eq:gauss-Ent-LB}.
\end{multline}

Substituting~\eqref{eq:gauss-Ent-LB} into~\eqref{eq:Gauss-LB} 
and taking $n\rightarrow \infty$, we recover 
\begin{align}
R \ge \frac{1}{2}\log_{2}\left(\frac{1}{d_{0}}\right) + \frac{1}{2(W+1)}\sum_{j=1}^{B} \log_2 \left(\frac{1}{ d_{W+j}}\right).
\end{align} as required.

\subsection{Illustrative Suboptimal Schemes}
\label{subsec:Comparison}
We compare the optimal lossy rate-recovery function with  the following suboptimal schemes. 
\subsubsection{Still-Image Compression} In this scheme, the encoder ignores the decoder's memory and at time $i\ge 0$ encodes the source $\rvbt_{i}$  in a memoryless manner and sends the codewords through the channel. The rate associated with this scheme is 
\begin{align}
R_{\textrm{SI}}(\bd)&=I(\rvbt_{i};\hat{\rvbt}_{i})=\sum_{k=0}^{K}\frac{1}{2}\log \bigg(\frac{1}{d_{k}}\bigg)
\end{align} In this scheme, the decoder is able to recover the source whenever its codeword is available, i.e. at all the times except when the erasure happens.   
\subsubsection{Wyner-Ziv Compression with Delayed Side Information} At time $i$ the encoders assumes that $\rvbt_{i-B-1}$ is already reconstructed at the receiver within distortion $\bd$. With this assumption, it compresses the source $\rvbt_{i}$ according to Wyner-Ziv scheme and transmits the codewords through the channel. The rate of this scheme is 
\begin{align}
R_{\textrm{WZ}}(B,\bd)&=I(\rvbt_{i};\hat{\rvbt}_{i}|\hat{\rvbt}_{i-B-1})=\sum_{k=0}^{B}\frac{1}{2}\log \bigg(\frac{1}{d_{k}}\bigg)
\end{align} Note that, if at time $i$, $\hat{\rvbt}_{i-B-1}$ is not available, $\hat{\rvbt}_{i-1}$ is available and the decoder can still use it as side-information to construct $\hat{\rvbt}_{i}$ since $I(\rvbt_{i};\hat{\rvbt}_{i}|\hat{\rvbt}_{i-B-1})\ge I(\rvbt_{i};\hat{\rvbt}_{i}|\hat{\rvbt}_{i-1})$. 

As in the case of Still-Image Compression, the Wyner-Ziv scheme also enables the recovery of each source sequence except those with erased codewords.

\subsubsection{Predictive Coding plus FEC} This scheme consists of predictive coding followed by a Forward Error Correction (FEC) code to compensate the effect of packet losses of the channel. As the contribution of $B$ erased codewords need to be recovered using $W+1$ available codewords, the rate of this scheme can be computed as follows. 
\begin{align}
R_{\textrm{FEC}}(B,W,\bd)&=\frac{B+W+1}{W+1}I(\rvbt_{i};\hat{\rvbt}_{i}|\hat{\rvbt}_{i-1})\\&=\frac{B+W+1}{2(W+1)}\log \bigg(\frac{1}{d_{0}}\bigg)
\end{align} 

\subsubsection{GOP-Based Compression} 
This scheme consists of predictive coding where the synchronization sources (I-frames) are inserted periodically to prevent error propagation. The synchronization frames are transmitted with the rate $R_1=I(\rvbt_{i} ; \hat{\rvbt}_{i})$ and the rest of the frames are transmitted at the rate $R_2=I(\rvbt_{i} ; \hat{\rvbt}_{i} | \hat{\rvbt}_{i-1})$ using predictive coding. Whenever the erasure happens the decoder fails to recover the source sequences until the next synchronization source and then the decoder becomes synced to the encoder. In order to guarantee the recovery of the sources, 
the synchronization frames have to be inserted with the period of at most $W+1$. This results in the following average rate expression.

\begin{align}
R&= \frac{1}{(W+1)} I(\rvbt_{i} ; \hat{\rvbt}_{i})+ \frac{W}{(W+1)} I(\rvbt_{i} ; \hat{\rvbt}_{i} | \hat{\rvbt}_{i-1})\\
&= \frac{1}{2(W+1)} \sum_{k=0}^{K}\log (\frac{1}{d_{k}})+ \frac{W}{2(W+1)}\log(\frac{1}{d_{0}})
\end{align}

\begin{figure}
\centering
\includegraphics[scale=0.6]{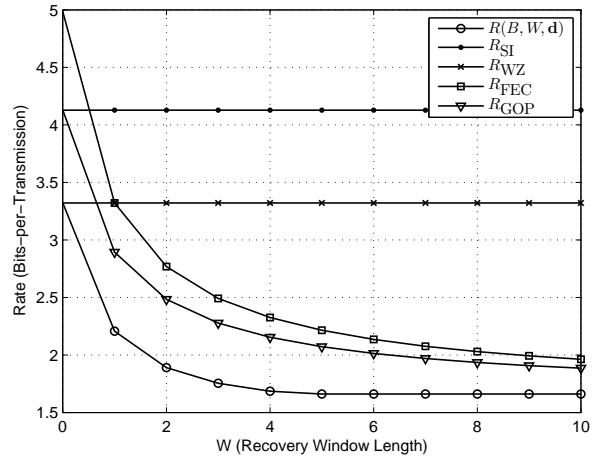}
\caption{Comparison of rate-recovery of suboptimal systems to minimum possible rate-recovery function for different recovery window length $W$. We assume $K=5$, ${B=2}$ and a distortion vector ${\bd=(0.1,0.25,0.4,0.55,0.7,.85)}$.  }
\label{fig:comparison}
\end{figure} 

In Fig.~\ref{fig:comparison},  we compare the result in Theorem~\ref{thm:gauss-rate} with the 
described schemes. 
It can be observed from Fig.~\ref{fig:comparison} that except when $W=0$ none of the other schemes are optimal. The \emph{Predictive Coding plus FEC} scheme, which is a natural {separation based scheme}  and the \emph{GOP-based} compression scheme are  suboptimal even for relatively large values of $W$.  Also note that the \emph{GOP-based} compression scheme reduces to \emph{Still-Image} compression for $W=0$.
 

\section{Conclusions}
\label{sec:Conclusion}

We presented a real-time streaming scenario where a sequence of source vectors must be sequentially encoded, and transmitted over a burst erasure channel.  The source vectors must be reconstructed with zero delay at the destination. However those sequences that occur during the erasure burst or a period of length $W$ following the burst need not be reconstructed. We assume that the source vectors are sampled i.i.d.\ across the spatial dimension and from a first-order, stationary, Markov process across the temporal dimension. We study the minimum achievable compression rate, which we define to be the rate-recovery function in our setup. 

For the case of discrete sources and lossless recovery, we establish upper and lower bounds on the rate-recovery function and observe that they coincide for the special cases when $W=0$ and $W\to \infty$. More generally both our upper and lower bound expressions can be expressed as the rate of  predictive coding plus another term that decreases at-least inversely with $W$. For the restricted  class of memoryless encoders and symmetric sources, we establish that a binning based scheme is optimal. For the case of Gauss-Markov sources and a quadratic distortion measure, we establish upper and lower bounds on the minimum rate when $W=0$ and observe that these bounds coincide in the high-resolution regime. The achievability is based on a quantize and binning scheme, but the analysis is a non-trivial extension of the lossless case as the reconstruction sequences at the destination do not form a Markov chain. We also study another setup involving independent Gaussian sources and a sliding-window reconstruction constraint where the rate-recovery function is attained using a successive refinement coding scheme. 

We believe that the present work can be extended in a number of directions. The focus in this paper has been on {causal encoders and zero-delay decoders. It is interesting to consider more general encoders with finite-lookahead and decoders with delay constraint. Some such extensions have been recently considered in~\cite{etezadi-khisti-14}. Secondly, the paper considers} only the case of burst-erasure channels. It may be interesting to consider channels that introduce both burst erasures and isolated erasures as considered recently in the channel coding context~\cite{badrinfo:13}. Thirdly, our present setup assumes that within the recovery period, a complete outage is declared, and no reconstruction is necessary. Finally our paper only addressed the case of lossless recovery for discrete sources. Extensions to lossy reconstruction, analogous to the case of Gaussian sources, will require characterization of the worst-case erasure sequence for a general source model, which appears challenging. This may be further generalized by considering partial recovery with a higher distortion during the recovery period. Such extensions will undoubtedly lead to a finer understanding of tradeoffs between compression rate and error propagation  in video transmission systems. 

\appendices
\section{Proof of Corollary~\ref{corol:genUB}}
\label{app:Cor1}
{
We want to show the following equality. 
\begin{align}
R^{+}(B,W) &= H(\rvs_{1}|\rvs_{0}) + \frac{1}{W+1}I(\rvs_{B};\rvs_{B+1} |\rvs_{0})\notag\\
&=\frac{1}{W+1}H(\rvs_{B+1}, \rvs_{B+2},\ldots, \rvs_{B+W+1}|\rvs_{0}) \label{eq:Cor1_h1}
\end{align}
}
According to the chain rule of entropies, the term in \eqref{eq:Cor1_h1} can be written as 
{\allowdisplaybreaks{\begin{align}
&H(\rvs_{B+1}, \rvs_{B+2},\ldots, \rvs_{B+W+1}|\rvs_{0})\\
&=H(\rvs_{B+1}|\rvs_{0})+ \sum_{k=1}^{W}H(\rvs_{B+k+1}|\rvs_{0}, \rvs_{B+1},\ldots,\rvs_{B+k})\notag\\
&=H(\rvs_{B+1}|\rvs_{0})+ W H(\rvs_{1}|\rvs_{0})\label{eq:cr1}\\
&=H(\rvs_{B+1}|\rvs_{0})-H(\rvs_{B+1}|\rvs_{B}, \rvs_{0})+H(\rvs_{B+1}|\rvs_{B}, \rvs_{0}) \notag\\ &\qquad+ W H(\rvs_{1}|\rvs_{0})\label{eq:cr2}\\
&=H(\rvs_{B+1}|\rvs_{0})-H(\rvs_{B+1}|\rvs_{B}, \rvs_{0})+H(\rvs_{B+1}|\rvs_{B}) \notag\\ &\qquad + W H(\rvs_{1}|\rvs_{0})\label{eq:cr3}\\
&=I(\rvs_{B+1};\rvs_{B}|\rvs_{0})+(W+1)H(\rvs_{1}|\rvs_{0})\label{eq:app-cor}\\
&=(W+1)R^{+}(B,W)
\end{align}}} where \eqref{eq:cr1} follows from the Markov property
\begin{align}
\left(\rvs_{0}, \rvs_{B+1}, \ldots,\rvs_{B+k-1}\right)\rightarrow \rvs_{B+k} \rightarrow  \rvs_{B+k+1} \label{eq:Mar}
\end{align} for any $k$ and from the stationarity of the sources which for each $k$ implies that
\begin{align}
H(\rvs_{B+k+1}|\rvs_{B+k})=H(\rvs_{1}|\rvs_{0}).
\end{align} Note that in \eqref{eq:cr2} we add and subtract the same term and \eqref{eq:cr3} also follows from the Markov property of \eqref{eq:Mar} for $k=0$.

%
%
\section{Proof of Lemma~\ref{lem:1}}
\label{ApA}
Define $q_{k} \triangleq 2^{\frac{2}{n}h(\rvs^n_{k}|\rvf^{k}, \rvs^n_{-1})}$. We need to show that 
\begin{align}
q_k &\ge \frac{2\pi e (1-\rho^2)}{2^{2R}-\rho^2}\left(1-\left(\frac{\rho^2}{2^{2R}}\right)^k\right) 
\end{align}
Consider the following entropy term.
\begin{align}
&h(\rvs^n_{k}|[\rvf]_{0}^{k}, \rvs^n_{-1}) = h(\rvs^n_{k}|[\rvf]_{0}^{k-1}, \rvs^n_{-1})- I(\rvf_{k};\rvs_{k}^n|[\rvf]_{0}^{k-1}, \rvs^n_{-1})\notag\\
&= h(\rvs^n_{k}|[\rvf]_{0}^{k-1}, \rvs^n_{-1})- \notag\\ &\qquad H(\rvf_{k}|[\rvf]_{0}^{k-1}, \rvs^n_{-1}) +  H(\rvf_{k}|\rvs^n_{k}, [\rvf]_{0}^{k-1}, \rvs^n_{-1})\notag\\
& \ge h(\rvs^n_{k}|[\rvf]_{0}^{k-1}, \rvs^n_{-1})- H(\rvf_{k})\label{eq:4.5}\\
& \ge \frac{n}{2}\log\left( \rho^{2}2^{\frac{2}{n}h(\rvs^n_{k-1}|[\rvf]_{0}^{k-1}, \rvs^n_{-1})}+ 2\pi e (1-\rho^{2})\right) - nR \label{eq:5}
\end{align}
where \eqref{eq:4.5} follows from the fact that conditioning reduces entropy and \eqref{eq:5} follows from the Entropy Power Inequality similar to \eqref{eq:31}. Thus
\begin{align}
q_k \ge \frac{\rho^2}{2^{2R}} q_{k-1} + \frac{2\pi e (1-\rho^2)}{2^{2R}}. \label{eq:5.1}
\end{align} 
By repeating the iteration in \eqref{eq:5.1}, we have
\begin{align}
q_k &\ge (\frac{\rho^{2}}{2^{2R}})^k q_{0} + \frac{2\pi e (1-\rho^2)}{2^{2R}} \sum_{l=0}^{k-1} (\frac{\rho^{2}}{2^{2R}})^{l}\label{eq:5.2} \\
&\ge \frac{2\pi e (1-\rho^2)}{2^{2R}-\rho^2}\left(1-\left(\frac{\rho^2}{2^{2R}}\right)^k\right) \label{eq:5.3}, 
\end{align} where \eqref{eq:5.3} follows from the fact $0 <\frac{\rho^2}{2^{2R}} < 1$ for any $\rho \in (0,1)$ and $R>0$. This completes the proof.

%
%
%
%
%
%
%

\section{Proof of~\eqref{eq:hStep2App} and~\eqref{eq:hStep2DApp}}
\label{App:hStep2}

We need to show~\eqref{eq:hStep2} and~\eqref{eq:hStep2D}, i.e. we need to establish the following two inequities for each $k \in [1:t-B']$
\begin{multline}
h( \rvu_{t} | [\rvu]_{0}^{t-B'-k-1}, [\rvu]_{t-k}^{t-1}, \rvs_{-1})\le \\ h( \rvu_{t} | [\rvu]_{0}^{t-B'-k}, [\rvu]_{t-k+1}^{t-1}, \rvs_{-1})\label{eq:hStep2App}\end{multline}\begin{multline}
h( \rvs_{t} | [\rvu]_{0}^{t-B'-k-1}, [\rvu]_{t-k}^{t}, \rvs_{-1})\le \\ h( \rvs_{t} | [\rvu]_{0}^{t-B'-k}, [\rvu]_{t-k+1}^{t}, \rvs_{-1})\label{eq:hStep2DApp}.
\end{multline}

We first establish the following Lemmas.

\begin{lemma}
\label{lem:AchGM}
Consider  random variables $\{\rvX_{0}, \rvX_{1}, \rvX_{2}, \rvY_{1} , \rvY_{2}\}$ that are jointly Gaussian, $\rvX_{k} \sim \cN(0,1), k\in\{0,1,2\}$,  $\rvX_0 \rightarrow \rvX_1 \rightarrow \rvX_2$ and that for $j\in\{1,2\}$ we have: 
\begin{align}
\rvX_{j} &\!= \rho_{j} \rvX_{j-1} + \rvN_{j}, \\
\rvY_{j} &\!= \rvX_{j} + \rvZ_{j}.
\end{align}
 Assume that $\rvZ_{j} \sim \cN(0,\sigma^2_{z})$ are independent of all random variables and likewise  $\rvN_{j} \sim \cN(0,1-\rho^{2}_{j})$ for $j\in\{1,2\}$ are also independent of all random variables.  The structure of correlation is sketched in Fig.~\ref{fig:five}. Then we have that:
\begin{align}
&\sigma^{2}_{\rvX_2}(\rvX_0,\rvY_{2}) \le \sigma^{2}_{\rvX_2}(\rvX_0,\rvY_{1}) \label{eqvar}
\end{align} where $\sigma^{2}_{\rvX_2}(\rvX_0,\rvY_{j})$ denotes the minimum mean square error of estimating $\rvX_2$ from $\{\rvX_0,\rvY_{j}\}$. 
\end{lemma} 
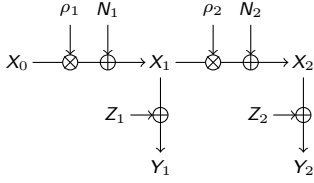
\begin{figure}
\begin{center}
\vspace{1em}
\begin{tikzpicture}
\foreach \t in{0,1.9}{
\draw[-] (0.2+\t,0) -- (.6+\t,0);
\draw[->] (0.8+\t,0) -- (1.7+\t,0);
\draw (0.7+\t,0) circle (.1);
\draw (1.2+\t,0) circle (.1);
\draw[-] (1.2+\t,-.1)--(1.2+\t,.1);
\draw[-] (.7 - .08+\t, -.08) -- (0.7+.08+\t,.08);
\draw[-] (.7 - .08+\t, +.08) -- (0.7+.08+\t,-.08);

\draw[->] (0.7+\t,0.5) -- (.7+\t,0.1);
\draw[->] (1.2+\t,0.5) -- (1.2+\t,0.1);

\draw[->] (1.9+\t,-.20) -- (1.9+\t,-1.2);
\draw (1.9+\t,-.7) circle (.1);
\draw[-] (1.9+\t-.1,-.7) -- (1.9+\t+.1,-.7) ;

\draw[->] (1.5+\t,-0.7) -- (1.8+\t,-0.7);
}
%
\draw [white] (0,0) -- (0,0) node {$\color{black}\scriptstyle \rvX_0$};
\draw [white] (1.9,0) -- (1.9,0) node {$\color{black}\scriptstyle \rvX_1$};
\draw [white] (3.8,0) -- (3.8,0) node {$\color{black}\scriptstyle \rvX_2$};
\draw [white] (1.9,-1.4) -- (1.9,-1.4) node {$\color{black}\scriptstyle \rvY_1$};
\draw [white] (3.8,-1.4) -- (3.8,-1.4) node {$\color{black}\scriptstyle \rvY_2$};

\draw [white] (1.3,-.7) -- (1.3,-.7) node {$\color{black}\scriptstyle \rvZ_1$};
\draw [white] (3.2,-.7) -- (3.2,-.7)  node {$\color{black}\scriptstyle \rvZ_2$};

\draw [white] (1.2,.7) -- (1.2,.7) node {$\color{black}\scriptstyle \rvN_1$};
\draw [white] (3.1,.7) -- (3.1,.7)  node {$\color{black}\scriptstyle \rvN_2$};

\draw [white] (.7,.7) -- (.7,.7) node {$\color{black}\scriptstyle \rho_1$};
\draw [white] (2.6,.7) -- (2.6,.7)  node {$\color{black}\scriptstyle \rho_2$};

\end{tikzpicture}
\caption{Relationship of the Variables for Lemma \ref{lem:AchGM}.}
\label{fig:five}
\end{center}
\end{figure}

\begin{proof}
By applying the standard relation for the MMSE estimation error we have (see e.g. \cite{willskyWornel03})
\begin{align}
&\sigma^{2}_{\rvX_2}(\rvX_0,\rvY_{1}) \notag\\
&= E[X_2^2] -\!  \notag\\ 
& \begin{pmatrix}  \! E[X_2 Y_1]\! & \!E[X_2 X_0]\! \end{pmatrix} \begin{pmatrix}
E[Y_1^2]  & E[X_0 Y_1] \\ E[X_0 Y_1] & E[X_0^2]
\end{pmatrix}^{-1} 
\begin{pmatrix}   E[X_2 Y_1] \\ E[X_2 X_0] \end{pmatrix} \\
&=1 - \rho_2^2 \begin{pmatrix}
 1 & \rho_{1} \end{pmatrix}\begin{pmatrix}
1+\sigma^2_{z} & \rho_{1} \\
 \rho_{1} &  1 \end{pmatrix}^{-1}\begin{pmatrix}
1\\
\rho_{1} \end{pmatrix}\\
&= 1- \frac{\rho_{1}^2\rho_{2}^2\sigma^{2}_{z}-\rho_{1}^2\rho_{2}^2+\rho_{2}^2}{1+\sigma^{2}_{z}-\rho^{2}_{1}}\label{eqlem11}
\end{align} 
where we use the fact that $E[X_0^2]=1$, $E[Y_1^2]=1 + \sigma_z^2$,
$E[X_0 Y_1] = \rho_1$, $E[X_2 X_0] = \rho_0 \rho_1$ and $E[X_2 Y_1]= \rho_2$. In a similar fashion it can be shown that:
\begin{align}
\sigma^{2}_{\rvX_2}(\rvX_0,\rvY_{2}) &= 1 -  \begin{pmatrix}
 1 & \rho_{1} \rho_{2} \end{pmatrix}\begin{pmatrix}
1+\sigma^2_{z} & \rho_{1}\rho_{2} \\
 \rho_{1}\rho_{2} &  1 \end{pmatrix}^{-1}\begin{pmatrix}
1\\
\rho_{1}\rho_{2} \end{pmatrix}\\
&= 1- \frac{\rho_{1}^2\rho_{2}^2\sigma^{2}_{z}-\rho_{1}^{2}\rho_{2}^2+1}{1+\sigma^{2}_{z}-\rho^{2}_{1}\rho^{2}_{2}} \label{eqlem22}
\end{align} 
To establish~\eqref{eqvar} we only need to show that, 
\begin{align}
\frac{\rho_{1}^2\rho_{2}^2\sigma^{2}_{z}-\rho_{1}^{2}\rho_{2}^2+1}{1+\sigma^{2}_{z}-\rho^{2}_{1}\rho^{2}_{2}}  \ge  \frac{\rho_{1}^2\rho_{2}^2\sigma^{2}_{z}-\rho_{1}^2\rho_{2}^2+\rho_{2}^2}{1+\sigma^{2}_{z}-\rho^{2}_{1}} \label{eqlem33}
\end{align} It is equivalent to showing 
\begin{align}
\frac{1+\sigma^{2}_{z}-\rho^{2}_{1}}{1+\sigma^{2}_{z}-\rho^{2}_{1}\rho^{2}_{2}} \ge  \frac{\rho_{1}^2\rho_{2}^2\sigma^{2}_{z}-\rho_{1}^2\rho_{2}^2+\rho_{2}^2}{\rho_{1}^2\rho_{2}^2\sigma^{2}_{z}-\rho_{1}^{2}\rho_{2}^2+1}   \label{eqlem44}
\end{align} or equivalently
\begin{align}
1- \frac{\rho^2_{1}(1-\rho^2_{2})}{1+\sigma^{2}_{z}-\rho^{2}_{1}\rho^{2}_{2}} \ge 
1- \frac{1-\rho^2_{2}}{\rho_{1}^2\rho_{2}^2\sigma^{2}_{z}-\rho_{1}^{2}\rho_{2}^2+1}  
\end{align} which is equivalent to showing
\begin{align}
\frac{\rho^2_{1}}{1+\sigma^{2}_{z}-\rho^{2}_{1}\rho^{2}_{2}} \le
\frac{1}{\rho_{1}^2\rho_{2}^2\sigma^{2}_{z}-\rho_{1}^{2}\rho_{2}^2+1}. \label{eq:immediate-verify}
\end{align}
However note that~\eqref{eq:immediate-verify} can be immediately verified since the left hand side
has the numerator smaller than the right hand side and the denominator greater than the right hand side whenever
$\rho_i^2 \in (0,1)$.
This completes the proof.
\end{proof}

\begin{lemma}
\label{lem:AKSuggest}
Consider the Gauss-Markov source model~\eqref{eq:GM-Def} and the test channel in Prop.~\ref{prop:GM-ME}.
For a fixed $t$, $k\in[1, t]$ and a set $\Omega \subseteq [t-k,t]$, consider two sets of random variables $\cW_1$ and $\cW_{2}$ each  jointly Gaussian with $\rvs_{t-k}$  such that the following Markov property holds:
\begin{align}
&\cW_1 \rightarrow \rvs_{t-k} \rightarrow \{\rvs_{t}, \rvbu_{\Omega}\}\\
&\cW_2 \rightarrow \rvs_{t-k} \rightarrow \{\rvs_{t}, \rvbu_{\Omega}\}
\end{align} If the MMSE error in $\rvs_{t-k}$ satisfies , $\sigma^2_{t-k}(\cW_1)\le \sigma^2_{t-k}(\cW_2)$ then we have
\begin{align}
&h(\rvs_{t} | \cW_1, \rvbu_{\Omega}) \le h(\rvs_{t} | \cW_2, \rvbu_{\Omega}), \qquad \forall \Omega \subseteq [t-k,t] \label{lem:AKs1}\\
&h(\rvu_{t} | \cW_1, \rvbu_{\Omega}) \le h(\rvu_{t} | \cW_2, \rvbu_{\Omega}), \qquad \forall \Omega \subseteq [t-k,t-1] \label{lem:AKs2}.
\end{align}
\end{lemma}
\begin{proof}
Since the underlying random variables are jointly Gaussian, we can express the MMSE estimates of $\rvs_{t-k}$ from $\cW_j$, $j\in\{1,2\}$ as follows (see e.g. \cite{willskyWornel03})
\begin{align}
&\hat{\rvs}_{t-k}(\cW_1) = \alpha_1 \rvs_{t-k} +\rve_1 \label{eq:AKlem1}\\
&\hat{\rvs}_{t-k}(\cW_2) = \alpha_2 \rvs_{t-k} +\rve_2 \label{eq:AKlem2}
\end{align} where $\rve_1 \sim \cN(0, E_1)$ and $\rve_2 \sim \cN(0, E_2)$ are  Gaussian  random variables both independent of $\rvs_{t-k}$. Furthermore the constants in~\eqref{eq:AKlem1} and~\eqref{eq:AKlem2} are given by 
\begin{align}
&\al_j = 1-\sigma^2_{t-k}(\cW_j)\\
& E_j = \sigma^2_{t-k}(\cW_j)(1-\sigma^2_{t-k}(\cW_j))
\end{align} for $j=1,2$. To establish~\eqref{lem:AKs1}, we have
\begin{align}
h(\rvs_{t} | \cW_1, \rvbu_{\Omega}) &= h(\rvs_{t} | \hat{\rvs}_{t-k}(\cW_1), \rvbu_{\Omega})\label{eq:Ak1}\\
 & = h(\rvs_{t} | \alpha_1 \rvs_{t-k} +\rve_1, \rvbu_{\Omega}) \label{eq:Ak2}\\
 & \le h(\rvs_{t} | \alpha_2 \rvs_{t-k} +\rve_2, \rvbu_{\Omega}) \label{eq:Ak3}\\
 &= h(\rvs_{t} | \hat{\rvs}_{t-k}(\cW_2), \rvbu_{\Omega})\label{eq:Ak4}\\
 &= h(\rvs_{t} | \cW_2, \rvbu_{\Omega})\label{eq:Ak5}
\end{align} where \eqref{eq:Ak1} and \eqref{eq:Ak5} follows from the following Markov property. 
\begin{align}
&\cW_1 \rightarrow \hat{\rvs}_{t-k}(\cW_{1}) \rightarrow \{\rvs_{t}, \rvu_{\Omega}\}\\
&\cW_2 \rightarrow \hat{\rvs}_{t-k}(\cW_{2}) \rightarrow \{\rvs_{t}, \rvu_{\Omega}\}
\end{align} \eqref{eq:Ak2} and \eqref{eq:Ak4} follows from \eqref{eq:AKlem1} and~\eqref{eq:AKlem2} and \eqref{eq:Ak3} follows from the fact that  $\sigma^2_{t-k}(\cW_1)\le \sigma^2_{t-k}(\cW_2)$ implies that 
\begin{align}
\frac{E_1}{\alpha^2_{1}} \le \frac{E_2}{\alpha^2_{2}}
\end{align} Thus the only difference between~\eqref{eq:Ak2} and~\eqref{eq:Ak3} is that the variance of the independent noise component in the first term is smaller in the former. Clearly we obtain a better estimate of $\rvs_t$ in~\eqref{eq:Ak2}, which justifies the inequality in~\eqref{eq:Ak3}.

Eq.~\eqref{lem:AKs2} can be established as an immediate consequence of~\eqref{lem:AKs1}.
Since the noise $\rvz_{t}$ in the test channel is Gaussian and independent of all other random variables, we have
\begin{align}
\textrm{Var}(\rvu_{t}|\cW_j, \rvbu_{\Omega} ) = \textrm{Var}(\rvs_{t}|\cW_j, \rvbu_{\Omega} ) + \sigma^2_{z}
\end{align} where the notation $\textrm{Var}(\rva|\mathcal{W})$ indicates the noise variance of estimating $\rva$ from $\mathcal{W}$. As a result, 
\begin{align}
h(\rvu_{t} | \cW_j, \rvbu_{\Omega} ) = \frac{1}{2} \log\left(2^{2h(\rvs_{t}|\cW_j, \rvbu_{\Omega} )}+ 2 \pi e \sigma^2_{z}\right). \label{eq:Ak6}
\end{align} Thus~\eqref{lem:AKs1} immediately imples~\eqref{lem:AKs2}. 
\end{proof} 

We now establish~\eqref{eq:hStep2App} and subsequently establish~\eqref{eq:hStep2DApp} in a similar fashion. Consider the following two steps. 

1) First by applying Lemma~\ref{lem:AchGM} we show.
\begin{multline}\sigma^{2}_{{t-k}}([\rvu]_{0}^{t-B'-k-1},\rvu_{t-k}, \rvs_{-1}) \le \\ \sigma^{2}_{{t-k}}([\rvu]_{0}^{t-B'-k-1},\rvu_{t-B'-k}, \rvs_{-1}),
 \label{lesnoise}
\end{multline}
i.e., knowing $\{[\rvu]_{0}^{t-B'-k-1},\rvu_{t-k}, \rvs_{-1} \}$ rather than 
$\{[\rvu]_{0}^{t-B'-k-1},\rvu_{t-B'-k}, \rvs_{-1}\}$,  improves the estimate of the source $\rvs_{t-k}$.
Let $\hat{\rvs}_{t-B'-k}([\rvu]_{0}^{t-B'-k-1}, \rvs_{-1})$ be the MMSE estimator of ${\rvs}_{t-B'-k}$ given $([\rvu]_{0}^{t-B'-k-1}, \rvs_{-1})$. 

Note that $\hat{\rvs}_{t-B'-k}([\rvu]_{0}^{t-B'-k-1}, \rvs_{-1})$ is a sufficient statistic of $\rvs_{t-B'-k}$ given $\{[\rvu]_{0}^{t-B'-k-1}, \rvs_{-1}\}$ and thus we have that:
\begin{multline}
\{[\rvu]_{0}^{t-B'-k-1}, \rvs_{-1}\} \rightarrow \\ \hat{\rvs}_{t-B'-k}([\rvu]_{0}^{t-B'-k-1}, \rvs_{-1})  \rightarrow {\rvs}_{t-B'-k}  \rightarrow {\rvs}_{t-k}. \label{useM3}
\end{multline}  
Therefore, by application of Lemma~\ref{lem:AchGM} for $\rvX_0 = \hat{\rvs}_{t-B'-k}([\rvu]_{0}^{t-B'-k-1}, \rvs_{-1})$, $\rvX_1 = \rvs_{t-B'-k}$, $\rvY_1 = \rvu_{t-B'-k}$, $\rvX_2 = \rvs_{t-k}$ and $\rvY_2 = \rvu_{t-k}$, we have 
\begin{align}
&\sigma^{2}_{{t-k}}([\rvu]_{0}^{t-B'-k-1},\rvu_{t-k}, \rvs_{-1})\notag\\
& = \sigma^{2}_{{t-k}}(\hat{\rvs}_{t-B'-k}([\rvu]_{0}^{t-B'-k-1}, \rvs_{-1}), \rvu_{t-k} ) \label{useM1}\\
& \le  \sigma^{2}_{{t-k}}(\hat{\rvs}_{t-B'-k}([\rvu]_{0}^{t-B'-k-1}, \rvs_{-1}), \rvu_{t-B'-k} )\label{useM2}\\
& = \sigma^{2}_{{t-k}}([\rvu]_{0}^{t-B'-k-1},\rvu_{t-B'-k}, \rvs_{-1}). \label{lesnoise2}
\end{align} where \eqref{useM1} and \eqref{lesnoise2} both follow from \eqref{useM3}. 
This completes the claim in~\eqref{lesnoise}.

2) In the second step, we apply Lemma~\ref{lem:AKSuggest} for 
\begin{align}
&\cW_{1}=  \{[\rvu]_{0}^{t-B'-k-1},\rvu_{t-k}, \rvs_{-1} \}\label{eq:sari1}\\
&\cW_{2}=  \{[\rvu]_{0}^{t-B'-k-1},\rvu_{t-B'-k}, \rvs_{-1} \}\label{eq:sari2}\\
& \Omega = [t-k+1,t-1]
\end{align} we have 
\begin{multline}
h\left(\rvu_{t}~|~[\rvu]_{0}^{t-B'-k-1},[\rvu]_{t-k}^{t-1},\rvs_{-1}\right) \le \\   h\left(\rvu_{t}~|~[\rvu]_{0}^{t-B'-k},[\rvu]_{t-k+1}^{t-1},\rvs_{-1}\right) 
\end{multline} and again by applying Lemma~\ref{lem:AKSuggest} for $\cW_{1}$ and $\cW_{2}$ in \eqref{eq:sari1} and \eqref{eq:sari2} and $\Omega = [t-k+1,t]$, we have 
\begin{multline}
h\left(\rvs_{t}~|~[\rvu]_{0}^{t-B'-k-1},[\rvu]_{t-k}^{t},\rvs_{-1}\right) \le \\ h\left(\rvs_{t}~|~[\rvu]_{0}^{t-B'-k},[\rvu]_{t-k+1}^{t},\rvs_{-1}\right)\label{aftercritD} 
\end{multline} 

This establishes~\eqref{eq:hStep2App} and \eqref{eq:hStep2DApp} and equivalently \eqref{eq:hStep2} and \eqref{eq:hStep2D}.

\section{Proof of Lemma~\ref{lem:help2}}
\label{App:help2}

{For reader's convenience, we first repeat the statement of the Lemma. Consider the two sets $A, B\subseteq \mathbb{N}$ each of size $r$ as $A=\{a_{1},a_{2},\cdots,a_{r}\}$, $B=\{b_{1},b_{2},\cdots,b_{r}\}$ such that $1\le a_1<a_2<\cdots<a_r$ and $1\le b_1<b_2<\cdots<b_r$ and for any $i \in \{1,\ldots,r\}$, $a_{i} \le b_{i}$. Then the test channel~\eqref{eq:testchannel} satisfies the following: 
\begin{align}
&h(\rvs_{t}|\rvbu_{A},\rvs_{-1}) \ge h(\rvs_{t}|\rvbu_{B},\rvs_{-1}), \quad \forall t \ge b_r \label{eq:lem2-help2App}\\
&h(\rvu_{t}|\rvbu_{A},\rvs_{-1}) \ge h(\rvu_{t}|\rvbu_{B},\rvs_{-1}), \quad \forall t > b_r \label{eq:lem1-help2App}. 
\end{align}
}  
 
We first prove \eqref{eq:lem2-help2App} by induction as follows. The proof of~\eqref{eq:lem1-help2App}  follows directly from~\eqref{eq:lem2-help2App} as discussed at the end of this section. 


$\bullet$ First we show that \eqref{eq:lem2-help2App} is true for $r=1$,  i.e. given $0\le a_1\le b_1$ and for all $t\ge b_1$ we need to show 
\begin{align}
&h(\rvs_{t}~|~ \rvu_{a_{1}}, \rvs_{-1}) \ge h(\rvs_{t}~|~ \rvu_{b_{1}}, \rvs_{-1}).\label{ni3sec}
\end{align} We apply Lemma~\ref{lem:AchGM} in Appendix~\ref{App:hStep2} for $\{\rvX_{0}, \rvX_{1}, \rvX_{2}, \rvY_1, \rvY_2\}= \{\rvs_{-1}, \rvs_{a_1}, \rvs_{b_1}, \rvu_{a_1}, \rvu_{b_1}\}$ which results in 
\begin{align}
h(\rvs_{b_1}|\rvu_{a_1}, \rvs_{-1}) \ge h(\rvs_{b_1}|\rvu_{b_1}, \rvs_{-1}) \label{eq:fori1}
\end{align} Thus~\eqref{ni3sec} holds for $t=b_1$. For any $t > b_1$ we can always express $\rvs_{t} = \rho^{t-b_1}\rvs_{b_1} +\tilde{n}$ where $\tilde{n} \sim \cN(0,1- \rho^{2(t-b_1)})$ and also we can express $\rvs_{b_1} = \hat{\rvs}_{b_1} (\rvu_{j}, \rvs_{-1}) + \rvw_j$ for $j\in \{a_1 , b_1\}$ where $\rvw_j \sim \cN(0, \sigma^2_{b_1}(\rvu_{j}, \rvs_{-1}))$ is the MMSE estimation error. For $j\in\{a_1, b_1\}$, we have
\begin{align}
\rvs_{t} = \rho^{t-b_1} \hat{\rvs}_{b_1}(\rvu_{j}, \rvs_{-1}) + \rho^{t-b_1} \rvw_{j} + \tilde{n}.
\end{align} Then we have
\begin{align}
\sigma^2_t(\rvu_{a_1}, \rvs_{-1}) &= \rho^{2(t-b_1)} \sigma^2_{b_1}(\rvu_{a_1}, \rvs_{-1}) + 1- \rho^{2(t-b_1)}\\
&\ge \rho^{2(t-b_1)} \sigma^2_{b_1}(\rvu_{b_1}, \rvs_{-1}) + 1- \rho^{2(t-b_1)} \label{eq:Asu}\\
&= \sigma^2_t(\rvu_{b_1}, \rvs_{-1})\label{eq:Asu1}
\end{align}  where \eqref{eq:Asu} immediately follows from \eqref{eq:fori1}. Thus \eqref{eq:Asu1} establishes \eqref{ni3sec} and the proof of the base case is now complete.


$\bullet$ Now assume that \eqref{eq:lem2-help2App} is true for $r$, i.e. for the sets $A_r,B_r$ of size $r$ satisfying $a_{i}\le b_{i}$ for $i \in \{1,\cdots, r\}$ and any $t\ge b_{r}$, 
\begin{align}
&h(\rvs_{t}|\rvbu_{A_r}, \rvs_{-1}) \ge h(\rvs_{t}|\rvbu_{B_r}, \rvs_{-1}) \label{ni5sec}
\end{align} We  show that  the lemma is also true for the sets $A_{r+1} = \{A_{r}, a_{r+1}\}$ and $B_{r+1}=\{B_{r}, b_{r+1}\}$ where $a_{r} \le a_{r+1}$, $b_{r} \le b_{r+1}$ and $a_{r+1} \le b_{r+1}$.  We establish this in two steps.

1) We show that 
\begin{align}
&h(\rvs_{t}| \rvbu_{A_{r+1}}, \rvs_{-1}) \ge h(\rvs_{t}| \rvbu_{A_{r}}, \rvu_{b_{r+1}}, \rvs_{-1})\label{help2Step1sec}.
\end{align} By application of Lemma~\ref{lem:AchGM} for \begin{multline}\{\rvX_{0}, \rvX_{1}, \rvX_{2}, \rvY_1, \rvY_2\}= \\ \{\hat{\rvs}_{a_{r}}(\rvbu_{A_{r}}, \rvs_{-1}), \rvs_{a_{r+1}}, \rvs_{b_{r+1}}, \rvu_{a_{r+1}}, \rvu_{b_{r+1}}\}\end{multline} we have 
\begin{multline}
h(\rvs_{b_{r+1}}| \hat{\rvs}_{a_{r}}(\rvbu_{A_{r}}, \rvs_{-1}), \rvu_{a_{r+1}})  \\ \ge h(\rvs_{b_{r+1}}| \hat{\rvs}_{a_{r}}(\rvbu_{A_{r}}, \rvs_{-1}), \rvu_{b_{r+1}}) \label{eq:ju1}
\end{multline} 
Thus~\eqref{help2Step1sec} holds for $t=b_{r+1}$. For $t \ge b_{r+1}$ we can use the argument analogous to that leading to~\eqref{eq:Asu1}. We omit the details as they are completely analogous. This establishes \eqref{help2Step1sec}.

2) Next we show that 
\begin{align}
h(\rvs_{t}| \rvbu_{A_{r}}, \rvu_{b_{r+1}}, \rvs_{-1}) \ge \\ h(\rvs_{t}| \rvbu_{B_{r+1}}, \rvs_{-1}) \label{help2Step2}.
\end{align}

First note that based on the induction hypothesis in \eqref{ni5sec} for $t=b_{r+1}$ we have 
\begin{align}
&h(\rvs_{b_{r+1}}|\rvbu_{A_r}, \rvs_{-1}) \ge h(\rvs_{b_{r+1}}|\rvbu_{B_r}, \rvs_{-1}) 
\end{align} and equivalently 
\begin{align}
&\sigma^{2}_{b_{r+1}}(\rvbu_{A_r}, \rvs_{-1}) \ge \sigma^2_{b_{r+1}}(\rvbu_{B_r}, \rvs_{-1}) 
\end{align} Now by application of Lemma~\ref{lem:AKSuggest} for $k= t-b_{r}$ and 
\begin{align}
&\cW_{1}=  \{\rvbu_{B_r}, \rvs_{-1}\}\label{eq:sari11}\\
&\cW_{2}=  \{\rvbu_{A_r}, \rvs_{-1} \}\label{eq:sari22}\\
& \Omega = \{b_{r+1}\}
\end{align} and noting that $\cW_j \rightarrow \rvs_{b_r} \rightarrow (\rvs_{b_{r+1}}, \rvu_\Omega)$ for $j=1,2$
we have
\begin{align}
h(\rvs_{t} | \rvbu_{A_r},\rvu_{b_{r+1}},  \rvs_{-1}) \ge h(\rvs_{t} |\rvbu_{B_r}, \rvu_{b_{r+1}}, \rvs_{-1})
\end{align} which is equivalent to \eqref{help2Step2}.

Combining \eqref{help2Step1sec} and \eqref{help2Step2} we have $h(\rvs_{t}| \rvbu_{A_{r+1}}, \rvs_{-1}) \ge h(\rvs_{t}| \rvbu_{B_{r+1}}, \rvs_{-1})$ which shows that \eqref{eq:lem2-help2App} is also true for $r+1$. 

This completes the induction and the proof of \eqref{eq:lem2-help2App} for general $r$. 

Finally note that \eqref{eq:lem2-help2App} implies \eqref{eq:lem1-help2App}  as follows. 
\begin{align}
h(\rvu_{t}| \rvbu_{A_r}, \rvs_{-1}) &=  \frac{1}{2} \log \left(2^{2h(\rvs_{t}| \rvbu_{A_r}, \rvs_{-1})} +  2 \pi e \sigma^2_{z}   \right) \label{eq:foa1} \\
& \ge \frac{1}{2} \log \left(2^{2h(\rvs_{t}| \rvbu_{B_r}, \rvs_{-1})} +  2 \pi e \sigma^2_{z}   \right) \label{eq:foa2}\\
& = h(\rvu_{t}| \rvbu_{B_r}, \rvs_{-1}) 
\end{align} where \eqref{eq:foa1} follows from the fact that the noise in the test channel is independent. Also \eqref{eq:foa2} follows from \eqref{eq:lem2-help2App}. 
This completes the proof.

\section{Proof of Lemma~\ref{lem:GM-ME}}
\label{App:Lemma7}

We prove each part separately as follows.

1) For any feasible set  $\Omega_t$ with size $\theta$ we have
\begin{align}
\lambda_{t}(\Omega_t) &= I(\rvs_{t}; \rvu_{t} | \rvbu_{\Omega_t}, \rvs_{-1})\notag\\
&=  h(\rvu_{t} | \rvbu_{\Omega_t}, \rvs_{-1})- h( \rvu_{t} | \rvs_{t})\notag\\
& \le h(\rvu_{t} | \rvbu_{\Omega^{\star}_t(\theta)}, \rvs_{-1})- h( \rvu_{t} | \rvs_{t})\label{eq:lemp1}\\
& = I(\rvs_{t}; \rvu_{t} | \rvbu_{\Omega^{\star}_t(\theta)}, \rvs_{-1}) \notag\\
& = \lambda_{t}(\Omega^{\star}_t(\theta)) \end{align} 
where \eqref{eq:lemp1} follows from the application of Lemma~\ref{lem:help2} with $A = \Omega^\star_t(\theta)$ and $B=\Omega_t$, which by construction of $\Omega^\star_t(\theta)$ clearly satisfy the required condition. Also note that 
\begin{align}
\frac{1}{2}\log\left(2 \pi e \gamma_{t}(\Omega_t)\right) &=  h (\rvs_{t} | \rvu_{t}, \rvbu_{\Omega_t}, \rvs_{-1})\notag\\
& \le h (\rvs_{t} | \rvu_{t}, \rvbu_{\Omega^{\star}_{t}(\theta)}, \rvs_{-1})\label{eq:lemmApp}\\
& =  \frac{1}{2}\log\left(2 \pi e \gamma_{t}(\Omega^{\star}_t(\theta))\right) 
\end{align} where~\eqref{eq:lemmApp} follows from Lemma~\ref{lem:help2} for the sets $A=\{\Omega^{\star}_{t}(\theta), t\}$ and $B= \{\Omega_t, t\}$. Thus we have $ \gamma_{t}(\Omega_t) \le \gamma_{t}(\Omega^{\star}_t(\theta))$.

2) We next argue that both $\lambda_{t}(\Omega^\star_t(\theta))$ and $\gamma_t(\Omega_t^\star(\theta))$ attain their maximum values with the minimum possible $\theta$.
Recall from Part 1 that when the number of erasures $n_e=t-\theta$ is fixed, the worst case sequence must have all erasure positions as close to $t$ as possible. Thus if $n_e \le B$ the worst case sequence consists of a single burst spanning $[t-n_e, t-1]$. If $B < n_e \le 2B$, the worst case sequence must have two burst erasures spanning $ [t-n_e-L, t-B-L-1] \cup [t-B, t-1]$. More generally the worst case sequence will consist of a sequence of burst erasures each (except possibly the first one) of length $B$ separated by a guard interval of length $L$. Thus the non-erased indices associated with decreasing values of $\theta$ are nested, i.e. $\theta_1 \le \theta_2$ implies that $\Omega_t^\star(\theta_1) \subseteq \Omega_t^\star(\theta_2)$. 
Further note that adding more elements in the non-erased indices $\Omega_t^\star(\cdot)$ can only decrease both $\lambda_t(\cdot)$ and $\gamma_t(\cdot)$, i.e. $\Omega_t^\star(\theta_1) \subseteq \Omega_t^\star(\theta_2)$ implies that $\lambda_{t}(\Omega^\star_t(\theta_1)) \ge \lambda_{t}(\Omega^\star_t(\theta_2))$ and $\gamma_{t}(\Omega^\star_t(\theta_1)) \ge \gamma_{t}(\Omega^\star_t(\theta_2))$.
Thus the worst case $\Omega_t^\star(\theta)$ must constitute the minimum possible value of $\theta$. The formal proof, which is analogous to the second part of Lemma~\ref{lem:3Step} will be skipped.

3) This property follows from the fact that in steady state the effect of knowing $\rvs_{-1}$ vanishes.  In particular we show below that $\lambda_{t+1}(\Omega^{\star}_{t+1}) \ge \lambda_{t}(\Omega^{\star}_{t}) $ and
$\gamma_{t+1}(\Omega^{\star}_{t+1}) \ge \gamma_t(\Omega^\star_t)$.

{\allowdisplaybreaks{\begin{align}
&\lambda_{t+1}(\Omega^{\star}_{t+1}) \notag\\
&= I(\rvs_{t+1}; \rvu_{t+1} | \rvbu_{\Omega^{\star}_{t+1}}, \rvs_{-1})\notag\\
& =  h(\rvu_{t+1} | \rvbu_{\Omega^{\star}_{t+1}}, \rvs_{-1}) - h(\rvu_{t+1} | \rvs_{t+1})\notag\\
& \ge h(\rvu_{t+1} | \rvbu_{\Omega^{\star}_{t+1}}, \rvs_{-1}, \rvs_{0}) - h(\rvu_{t+1} | \rvs_{t+1})\label{eq:lem3-1}\\
& = h(\rvu_{t+1} | \rvbu_{\Omega^{\star}_{t+1}\backslash \{0\}}, \rvs_{0}) -  h(\rvu_{t+1} | \rvs_{t+1}) \label{eq:lem3-2}\\
& = h(\rvu_{t} | \rvbu_{\Omega^{\star}_{t}}, \rvs_{-1}) -  h(\rvu_{t} | \rvs_{t}) \label{eq:lem3-3}\\
&= I(\rvs_{t}; \rvu_{t} | \rvbu_{\Omega^{\star}_{t}}, \rvs_{-1}) \notag\\
&= \lambda_{t}(\Omega^{\star}_{t}) 
\end{align}}} where \eqref{eq:lem3-1} follows from the fact that conditioning reduces the differential entropy. Also in \eqref{eq:lem3-2} the notation $\Omega^{\star}_{t+1}\backslash \{0\}$ indicates the set $\Omega^{\star}_{t+1}$ when the index $0$ is excluded if $0 \in \Omega^{\star}_{t+1}$. It can be easily verified that the set $\Omega^{\star}_{t}$ is equivalent to the set obtained by left shifting the elements of the set $\Omega^{\star}_{t+1}\backslash \{0\}$ by one. Then \eqref{eq:lem3-2} follows from this fact and the following Markov property.
\begin{align}
\{\rvu_{0}, \rvs_{-1}\} \rightarrow \{\rvbu_{\Omega^{\star}_{t+1}\backslash \{0\}}, \rvs_{0} \}\rightarrow \rvu_{t+1} 
\end{align} Eq.~\eqref{eq:lem3-3} follows from the time-invariant property of source model and the test channel. Also note that
\begin{align}
\frac{1}{2}\log\left(2 \pi e \gamma_{t+1}(\Omega^{\star}_{t+1})\right) &=  h (\rvs_{t+1} | \rvu_{t+1}, \rvbu_{\Omega^{\star}_{t+1}}, \rvs_{-1})\notag\\
& \ge  h(\rvs_{t+1} | \rvu_{t+1}, \rvbu_{\Omega^{\star}_{t+1}}, \rvs_{-1}, \rvs_{0}) \label{eq:lem3-4}\\
& = h(\rvs_{t+1} | \rvu_{t+1}, \rvbu_{\Omega^{\star}_{t+1}\backslash \{0\}}, \rvs_{0}) \label{eq:lem3-5}\\
& = h(\rvs_{t} | \rvu_{t}, \rvbu_{\Omega^{\star}_{t}}, \rvs_{-1})\label{eq:lem3-6}\\
& = \frac{1}{2}\log\left(2 \pi e \gamma_{t}(\Omega^{\star}_{t})\right)
\end{align} where \eqref{eq:lem3-4} follows from the fact that conditioning reduces the differential entropy, \eqref{eq:lem3-5} follows from the following Markov property  
\begin{align}
\{\rvu_{0}, \rvs_{-1}\} \rightarrow \{\rvbu_{\Omega^{\star}_{t+1}\backslash \{0\}}, \rvu_{t+1}, \rvs_{0} \}\rightarrow \rvs_{t+1} 
\end{align} and \eqref{eq:lem3-6} again follows from the time-invariant property of source model and the test channel. 
%
\section{Proof of Lemma~\ref{Claim:2}}
\label{App:Claim2}

We need to show 
\begin{multline}
I(\rvs_{t};\rvu_{t} | \tilde{\rvs}_{t-L-B} ,  [\rvu]_{t-L-B+1}^{t-B-1})\ge \lim_{t\to \infty} \lambda_{t}(\Omega^{\star}_t)\\ = \lim_{t\to \infty} I(\rvs_{t};\rvu_{t} | \rvbu_{\Omega^{\star}_{t}}, \rvs_{-1})\label{AppEq1}\end{multline}
\begin{multline} \sigma^{2}_{t}( \tilde{\rvs}_{t-L-B}, [\rvu]_{t-L-B+1}^{t-B-1},\rvu_{t})\ge \lim_{t\to \infty} \gamma_{t}(\Omega^{\star}_t) \\ = \lim_{t\to \infty} \sigma^{2}_{t}( \rvbu_{\Omega^{\star}_{t}}, \rvu_{t}, \rvs_{-1})\label{AppEq2}
\end{multline}

For any $t> L+B$, we can write
\begin{align}
 \lambda_{t}(\Omega^{\star}_t) &=  I(\rvs_{t};\rvu_{t} | \rvbu_{\Omega^{\star}_{t}}, \rvs_{-1})\\ 
&= I(\rvs_{t};\rvu_{t} | \rvbu_{\Omega^{\star}_{t-L-B}}, [\rvu]_{t-L-B+1}^{t-B-1}, \rvs_{-1}) \label{eq:omeg-tilde0}\\
&= I(\rvs_{t};\rvu_{t} | \hat{\rvs}_{t-L-B}(\rvbu_{\Omega^{\star}_{t-L-B}}, \rvs_{-1}) ,  [\rvu]_{t-L-B+1}^{t-B-1})\label{tilde-1}\\
& \le  I(\rvs_{t};\rvu_{t} | \tilde{\alpha} {\rvs}_{t-L-B} + \tilde{\rve} ,  [\rvu]_{t-L-B+1}^{t-B-1}) \label{tild_exp}\\
& = I(\rvs_{t};\rvu_{t} | {\rvs}_{t-L-B} + {\rve} ,  [\rvu]_{t-L-B+1}^{t-B-1})\notag\\
& = I(\rvs_{t};\rvu_{t} | \tilde{\rvs}_{t-L-B} ,  [\rvu]_{t-L-B+1}^{t-B-1})\label{tilde-2}\\
& = R^{+}_{\textrm{GM-ME}}(L,B,D) \label{eq:expre}
\end{align} where~\eqref{eq:omeg-tilde0} follows from the structure of $\Omega_t^\star$ in Lemma~\ref{lem:GM-ME},
~\eqref{tilde-1} follows from the Markov relation\begin{multline}
\{\rvbu_{\Omega^{\star}_{t-L-B}}, \rvs_{-1}\} \rightarrow \{\hat{\rvs}_{t-L-B}(\rvbu_{\Omega^{\star}_{t-L-B}}, \rvs_{-1} ), [\rvu]_{t-L-B+1}^{t-B-1} \}\\\rightarrow \rvs_{t} \label{eq:G-ME-Markov-1}
\end{multline}and in \eqref{tild_exp} we introduce ${\tilde{\alpha} = 1-D}$ and ${\tilde{\rve} \sim \cN(0, D(1-D))}$. This follows from the fact that the estimate $\hat{\rvs}_{t-L-B}(\rvbu_{\Omega^{\star}_{t-L-B}}, \rvs_{-1})$ satisfies the average distortion constraint of $D$. In \eqref{tilde-2} we re-normalize the test channel so that $\rve \sim \cN(0,D/(1-D))$. Taking the limit of \eqref{eq:expre} when $t\to\infty$, results in \eqref{AppEq1}. Also note that 

\begin{align}
 \gamma_{t}(\Omega^{\star}_t)  &= \sigma^{2}_{t}( \rvbu_{\Omega^{\star}_{t}}, \rvu_{t}, \rvs_{-1}) \notag\\
&=   \sigma^{2}_{t}( \rvbu_{\Omega^{\star}_{t-L-B}}, [\rvu]_{t-L-B+1}^{t-B-1},\rvu_{t}, \rvs_{-1})\notag\\
&=   \sigma^{2}_{t}( \hat{\rvs}_{t-L-B}(\rvbu_{\Omega^{\star}_{t-L-B}}, \rvs_{-1}), [\rvu]_{t-L-B+1}^{t-B-1}, \rvu_{t})\label{eq:lem3-7}\\
&\le   \sigma^{2}_{t}( \tilde{\alpha} {\rvs}_{t-L-B} + \tilde{\rve}, [\rvu]_{t-L-B+1}^{t-B-1}, \rvu_{t}) \label{eq:lem3-8}\\
&=   \sigma^{2}_{t}({\rvs}_{t-L-B} + {\rve}, [\rvu]_{t-L-B+1}^{t-B-1}, \rvu_{t}) \label{eq:lem3-9}\\
&=   \sigma^{2}_{t}( \tilde{\rvs}_{t-L-B}, [\rvu]_{t-L-B+1}^{t-B-1},\rvu_{t}) \label{eq:Disto}
\end{align} where \eqref{eq:lem3-7} follows from the following Markov property~\eqref{eq:G-ME-Markov-1}
 and \eqref{eq:lem3-8} again follows from the fact that the estimate $\hat{\rvs}_{t-L-B}(\rvbu_{\Omega^{\star}_{t-L-B}}, \rvs_{-1})$ satisfies the distortion constraint. All the constants and variables in \eqref{eq:lem3-8} and \eqref{eq:lem3-9} are as defined before. Again, taking the limit of \eqref{eq:Disto} when $t\to \infty$ results in \eqref{AppEq2}.

According to \eqref{eq:expre} and \eqref{eq:Disto} if we choose the noise in the test channel $\sigma^2_{z}$ to satisfy 
\begin{align}
\sigma^{2}_{t}( \tilde{\rvs}_{t-L-B}, [\rvu]_{t-L-B+1}^{t-B-1},\rvu_{t}) =D
\end{align} then the test channel and the rate $R^{+}_{\textrm{GM-ME}}(L,B,D)$ defined in \eqref{eq:expre} both satisfy rate and distortion constraints in \eqref{eq:limRate} and \eqref{eq:limtest} and therefore $R^{+}_{\textrm{GM-ME}}(L,B,D)$ is achievable.

\section{Proof of Lemma~\ref{claim:converse}}
\label{app:NEWapp}

\begin{figure*}
\begin{align}
h(\rvs_{t-B-W}^n, \ldots, \rvs_{t-W-1}^n)-h(\rvs_{t-B-W}^n, \ldots, \rvs_{t-W-1}^n|[\rvf]_{0}^{{t-B-W-1}}, [\rvf]_{t-W}^{t}, \rvs_{-1}^n)
 \ge \sum_{i=1}^{B}\frac{n}{2}\log{(\frac{1}{d_{W+i}})}\label{eq:Gauss_LB_T1App}
\end{align}
\end{figure*}

We first show that \eqref{eq:Gauss_LB_T1} which is repeated in~\eqref{eq:Gauss_LB_T1App} at the top of next page.
From the fact that conditioning reduces the differential entropy, we can lower bound the left hand side in~\eqref{eq:Gauss_LB_T1App} by
\begin{align}
& \sum_{i=0}^{B-1}\left(h(\rvs_{t-B-W+i}^n)\right.\notag\\&\left.\quad-h(\rvs_{t-B-W+i}^n|[\rvf]_{0}^{{t-B-W-1}}, [\rvf]_{t-W}^{t}, \rvs_{-1}^n)\right)\label{eq:exp3}
\end{align}
We show that for each $i=0,1,\ldots, B-1$
\begin{multline}
h(\rvs_{t-B-W+i}^n)-h(\rvs_{t-B-W+i}^n|[\rvf]_{0}^{{t-B-W-1}}, [\rvf]_{t-W}^{t}, \rvs_{-1}^n)\ge \\  \frac{n}{2}\log \left(\frac{1}{d_{B+W-i}}\right),\label{eq:Gauss_RD_Bound}
\end{multline}
which then establishes~\eqref{eq:Gauss_LB_T1App}.
Recall that since there is a burst erasure between time $t \in [t-B-W, t-W-1]$ the receiver is required to reconstruct 
\begin{align}
\hat{\rvbt}^n_{t} = \left[\hat{\rvs}_{t}^n, \ldots, \hat{\rvs}_{t-B-W}^n\right]
\end{align}
with a distortion vector $(d_0,\ldots, d_{B+W})$ i.e., a reconstruction of $\hat{\rvs}_{t-B-W+i}^n$ is desired with a distortion of $d_{B+W-i}$ for $i=0,1,\ldots, B+W$ when the decoder is revealed $([\rvf]_{0}^{{t-B-W-1}}, [\rvf]_{t-W}^{t})$. Hence
\begin{align}
&h(\rvs_{t-B-W+i}^n) - h(\rvs_{t-B-W+i}^n|[\rvf]_{0}^{{t-B-W-1}}, [\rvf]_{t-W}^{t}, \rvs_{-1}^n)\notag\\
&\quad =h(\rvs_{t-B-W+i}^n) - h\left(\rvs_{t-B-W+i}^n|[\rvf]_{0}^{{t-B-W-1}},\right.\notag\\&\left.\quad \quad \quad \quad \quad [\rvf]_{t-W}^{t}, \rvs_{-1}^n,\{\hat{\rvs}^n_{t-B-W+i}\}_{d_{B+W-i}}\right)\\
&\quad \ge h(\rvs_{t-B-W+i}^n) - h(\rvs_{t-B-W+i}^n|\{\hat{\rvs}^n_{t-B-W+i}\}_{d_{B+W-i}})\\
&\quad \ge h(\rvs_{t-B-W+i}^n) - h(\rvs_{t-B-W+i}^n-\{\hat{\rvs}^n_{t-B-W+i}\}_{d_{B+W-i}})\label{eq:Gauss_toSub}
\end{align}
Since we have that 
\begin{align}
E\left[ \frac{1}{n}\sum_{j=1}^n (\rvs_{t-B-W+i,j}-\hat{\rvs}_{t-B-W+i,j})^2 \right] \le d_{B+W-i}
\end{align}
It follows from standard arguments that~\cite[Chapter 13]{coverThomas} that
\begin{equation}
h(\rvs_{t-B-W+i}^n-\{\hat{\rvs}^n_{t-B-W+i}\}_{d_{B+W-i}}) \le \frac{n}{2}\log {2\pi e}{(d_{B+W-i})}. \label{eq:Gauss_Jensen}
\end{equation}
Substituting~\eqref{eq:Gauss_Jensen} into~\eqref{eq:Gauss_toSub} and the fact that $h(\rvs_{t-B-W+i}^n)= \frac{n}{2}\log 2\pi e$ establishes~\eqref{eq:Gauss_RD_Bound}.

Now we establish~\eqref{eq:Gauss_LB_T2} which is repeated in~\eqref{eq:Gauss_LB_T2App} at the top of next page.
\begin{figure*}
\begin{align}
&h(\rvs_{t-W}^n, \ldots, \rvs_{t}^n)- h(\rvs_{t-W}^n, \ldots, \rvs_{t}^n| [\rvf]_{0}^{{t-B-W-1}},[\rvf]_{t-W}^{t},\rvs_{t-B-W}^n, \ldots, \rvs_{t-W-1}^n, \rvs_{-1}^n)\notag\\ &+ H([\rvf]_{t-W}^{t}|[\rvf]_{0}^{{t-B-W-1}}, \rvbt_{t}^n, \rvs_{-1}^n) \ge \frac{n(W+1)}{2}\log(\frac{1}{d_0})\label{eq:Gauss_LB_T2App}
\end{align}
\end{figure*}
Since $(\rvs_{t-W}^n, \ldots, \rvs_{t}^n)$ are independent we can express the left-hand side in~\eqref{eq:Gauss_LB_T2App} as:
\begin{align}
&I\left(\rvs_{t-W}^n, \ldots, \rvs_{t}^n;[\rvf]_{t-W}^{t}~|[\rvf]_{0}^{{t-B-W-1}},\rvs_{t-B-W}^n, \right.\notag\\
&\left.\quad
\ldots, \rvs_{t-W-1}^n, \rvs_{-1}^n\right) \notag\\ &+ H([\rvf]_{t-W}^{t}|[\rvf]_{0}^{{t-B-W-1}}, \rvbt_{t}^n, \rvs_{-1}^n)\\
&=H([\rvf]_{t-W}^{t}|[\rvf]_{0}^{{t-B-W-1}},\rvs_{t-B-W}^n, \ldots, \rvs_{t-W-1}^n, \rvs_{-1}^n)\\
&\ge H([\rvf]_{t-W}^{t}|[\rvf]_{0}^{{t-W-1}},\rvs_{t-B-W}^n, \ldots, \rvs_{t-W-1}^n, \rvs_{-1}^n)\notag\\
&\ge I\left([\rvf]_{t-W}^{t};\rvs_{t-W}^n, \ldots, \rvs_{t}^n|[\rvf]_{0}^{t-W-1},\rvs_{t-B-W}^n, \right.\notag\\&\left.\quad
\ldots, \rvs_{t-W-1}^n, \rvs_{-1}^n\right)
\end{align}

The above mutual information term can be bounded as follows:
\begin{align}
&h(\rvs_{t-W}^n, \ldots, \rvs_{t}^n|[\rvf]_{0}^{t-W-1},\rvs_{t-B-W}^n, \ldots, \rvs_{t-W-1}^n, \rvs_{-1}^n)\notag\\
&\quad-h(\rvs_{t-W}^n, \ldots, \rvs_{t}^n|[\rvf]_{0}^{t},\rvs_{t-B-W}^n, \ldots, \rvs_{t-W-1}^n, \rvs_{-1}^n)\notag\\
&=h(\rvs_{t-W}^n, \ldots, \rvs_{t}^n)\notag\\ &\qquad -h(\rvs_{t-W}^n, \ldots, \rvs_{t}^n|[\rvf]_{0}^{t},\rvs_{t-B-W}^n, \ldots, \rvs_{t-W-1}^n, \rvs_{-1}^n)\label{eq:Gauss_Indep}\\
&\ge h(\rvs_{t-W}^n, \ldots, \rvs_{t}^n)-  h(\rvs_{t-W}^n, \ldots, \rvs_{t}^n|\{\hat{\rvs}_{t-W}^n\}_{d_0}, \ldots, \{\hat{\rvs}_{t}^n\}_{d_0})\label{eq:Gauss_Reconstr_d0}\\
&\ge \sum_{i=0}^{W}\left( h(\rvs_{t-W+i}^n)-h(\rvs_{t-W+i}^n-\{\hat{\rvs}_{t-W+i}^n\}_{d_0})\right)\notag\\
&\ge \sum_{i=0}^{W}\frac{n}{2}\log (\frac{1}{d_0}) = \frac{n(W+1)}{2}\log(\frac{1}{d_0})\label{eq:exp9}
\end{align}
where~\eqref{eq:Gauss_Indep} follows from the independence of $(\rvs_{t-W}^n, \ldots, \rvs_{t}^n)$ from the past sequences, and~\eqref{eq:Gauss_Reconstr_d0} follows from the fact that given the entire past $[\rvf]_0^{t}$ each source sub-sequence needs to be reconstructed with a distortion of $d_0$ and the last step  follows from the standard approach in the proof of the rate-distortion theorem. This establishes~\eqref{eq:Gauss_LB_T2App}. 

This completes the proof.

\bibliographystyle{IEEEtran}	
\bibliography{sm}		

\begin{thebibliography}{10}
\providecommand{\url}[1]{#1}
\csname url@samestyle\endcsname
\providecommand{\newblock}{\relax}
\providecommand{\bibinfo}[2]{#2}
\providecommand{\BIBentrySTDinterwordspacing}{\spaceskip=0pt\relax}
\providecommand{\BIBentryALTinterwordstretchfactor}{4}
\providecommand{\BIBentryALTinterwordspacing}{\spaceskip=\fontdimen2\font plus
\BIBentryALTinterwordstretchfactor\fontdimen3\font minus
  \fontdimen4\font\relax}
\providecommand{\BIBforeignlanguage}[2]{{%
\expandafter\ifx\csname l@#1\endcsname\relax
\typeout{** WARNING: IEEEtran.bst: No hyphenation pattern has been}%
\typeout{** loaded for the language `#1'. Using the pattern for}%
\typeout{** the default language instead.}%
\else
\language=\csname l@#1\endcsname
\fi
#2}}
\providecommand{\BIBdecl}{\relax}
\BIBdecl

\bibitem{wang2000error}
Y.~Wang, S.~Wenger, J.~Wen, and A.~Katsaggelos, ``Error resilient video coding
  techniques,'' \emph{IEEE Signal Processing Magazine}, vol.~17, no.~4, pp.
  61--82, 2000.

\bibitem{tan}
W.~Tan and A.~Zakhor, ``Video multicast using layered {FEC} and scalable
  compression,'' \emph{IEEE Transactions on Circuits and Systems for Video
  Technology}, vol.~11, pp. 373--386, 2001.

\bibitem{Huang:08}
Y.~Huang, Y.~Kochman, and G.~Wornell, ``Causal transmission of colored source
  frames over a packet erasure channel,'' in \emph{{DCC}}, 2010, pp. 129--138.

\bibitem{pradhanRamchandran:03}
S.~S. Pradhan and K.~Ramchandran, ``Distributed source coding using syndromes
  ({DISCUS}): Design and construction,'' \emph{IEEE Trans.\ Inform.\ Theory},
  vol.~49, pp. 626--643, Mar. 2003.

\bibitem{Wang:06}
J.~Wang, V.~Prabhakaran, and K.~Ramchandran, ``Syndrome-based robust video
  transmission over networks with bursty losses,'' in \emph{{ICIP}}, Atlanta,
  GA, 2006.

\bibitem{Witsenhausen:79}
H.~S. Witsenhausen, ``On the structure of real-time source coders,'' \emph{Bell
  Syst. Tech. J.}, vol.~58, no.~6, pp. 1437--1451, Jul-Aug 1979.

\bibitem{Teneketzis}
D.~Teneketzis, ``On the structure of optimal real-time encoders and decoders in
  noisy communication,'' \emph{IEEE Trans.\ Inform.\ Theory}, vol.~52, no.~9,
  pp. 4017--4035, sep 2006.

\bibitem{Asnani}
H.~Asnani and T.~Weissman, ``Real-time coding with limited lookahead,''
  \emph{IEEE Trans.\ Inform.\ Theory}, vol.~59, no.~6, pp. 3582--3606, 2013.

\bibitem{berger}
H.~Viswanathan and T.~Berger, ``Sequential coding of correlated sources,''
  \emph{IEEE Trans.\ Inform.\ Theory}, vol.~46, no.~1, pp. 236 --246, jan 2000.

\bibitem{equitzCover:91}
W.~H. Equitz and T.~M. Cover, ``Successive refinement of information,''
  \emph{IEEE Trans.\ Inform.\ Theory}, vol.~37, pp. 269--275, Mar. 1991.

\bibitem{Wu}
J.~Wang and X.~Wu, ``Information flows in video coding,'' in \emph{Data
  Compression Conference (DCC)}, 2010, pp. 149--158.

\bibitem{ishwar}
N.~Ma and P.~Ishwar, ``On delayed sequential coding of correlated sources,''
  \emph{IEEE Trans.\ Inform.\ Theory}, vol.~57, pp. 3763--3782, 2011.

\bibitem{songChen12}
L.~Song, J.~Chen, J.~Wang, and T.~Liu, ``Gaussian robust sequential and
  predictive coding,'' \emph{IEEE Trans.\ Inform.\ Theory}, vol.~59, no.~6, pp.
  3635--3652, 2013.

\bibitem{Chung:00}
S.~Y. Chung, ``On the construction of some capacity approaching coding
  schemes,'' Ph.D. dissertation, Mass.\ Instit.\ of Tech., 2000.

\bibitem{tuncel}
X.~Chen and E.~Tuncel, ``{Zero-delay joint source-channel coding for the
  Gaussian Wyner-Ziv problem},'' in \emph{ISIT}, 2011, pp. 2929--2933.

\bibitem{arildsen}
T.~Arildsen, M.~N. Murthi, S.~V. Andersen, and S.~H. Jensen, ``On predictive
  coding for erasure channels using a {K}alman framework,'' \emph{IEEE Trans.\
  Signal Processing}, vol.~57, no.~11, pp. 4456--4466, 2009.

\bibitem{Chang:07}
C.~Chang, ``Streaming source coding with delay,'' Ph.D. dissertation,
  U.~C.~Berkeley, 2007.

\bibitem{draper}
S.~Draper, C.~Chang, and A.~Sahai, ``Sequential random binning for streaming
  distributed source coding,'' in \emph{Proc.\ Int.\ Symp.\ Inform.\ Theory},
  2005, pp. 1396--1400.

\bibitem{gallager}
R.~G. Gallager, \emph{Information Theory and Reliable Communication}.\hskip 1em
  plus 0.5em minus 0.4em\relax John Wiley and Sons, 1968.

\bibitem{martinianThesis}
E.~Martinian, ``Dynamic information and constraints in source and channel
  coding,'' Ph.D. dissertation, Mass.\ Instit.\ of Tech., 2004.

\bibitem{badrinfo:13}
A.~Badr, A.~Khisti, W.~Tan, and J.~Apostolopoulos, ``Streaming codes for
  channels with burst and isolated erasures,'' in \emph{Proc.\ IEEE INFOCOMM},
  June 2013.

\bibitem{tekin}
O.~Tekin, T.~Ho, H.~Yao, and S.~Jaggi, ``On erasure correction coding for
  streaming,'' in \emph{ITA}, 2012, pp. 221--226.

\bibitem{leong}
D.~Leong and T.~Ho, ``Erasure coding for real-time streaming,'' in \emph{ISIT},
  2012.

\bibitem{markov}
D.~Aldous, ``Reversible {M}arkov chains and random walks on graphs ({C}hapter
  3),'' unpublished notes, available at
  {http://www.stat.berkeley.edu/$\sim$aldous/RWG/book.html}, Sep. 2002.

\bibitem{oohama1997}
Y.~Oohama, ``Gaussian multiterminal source coding,'' \emph{IEEE Trans.\
  Inform.\ Theory}, vol.~43, no.~6, pp. 1912--1923, 1997.

\bibitem{zigZag}
I.~Csisz\`{a}r and J.~Korner, ``Towards a general theory of source networks,''
  \emph{IEEE Trans.\ Inform.\ Theory}, vol.~26, no.~2, pp. 155 -- 165, mar
  1980.

\bibitem{Cover95}
T.~M. Cover, ``A proof of the data compression theorem of {S}lepian and {W}olf
  for ergodic sources,'' \emph{IEEE Trans.\ Inform.\ Theory}, vol. IT-21,
  no.~2, pp. 226--228, 1995.

\bibitem{slepianWolf:73}
D.~Slepian and J.~K. Wolf, ``Noiseless coding of correlated information
  sources,'' \emph{IEEE Trans.\ Inform.\ Theory}, vol.~19, pp. 471--480, July
  1973.

\bibitem{coverThomas}
T.~M. Cover and J.~A. Thomas, \emph{Elements of Information Theory}.\hskip 1em
  plus 0.5em minus 0.4em\relax John Wiley and Sons, 1991.

\bibitem{tavildar:10}
S.~Tavildar, P.~Viswanath, and A.~B. Wagner, ``The {G}aussian many-help-one
  distributed source coding problem,'' \emph{IEEE Trans.\ Inform.\ Theory},
  vol.~56, no.~1, pp. 564--581, 2010.

\bibitem{poor94}
H.~V. Poor, \emph{An Introduction to Signal Detection and Estimation, $2$nd
  Ed.}\hskip 1em plus 0.5em minus 0.4em\relax NewYork: Springer Verlag.

\bibitem{willskyWornel03}
A.~S. Willsky, G.~W. Wornell, and J.~H. Shapiro, \emph{Stochastic Processes,
  Detection and Estimation}.\hskip 1em plus 0.5em minus 0.4em\relax 6.432
  Course Notes, Department of Electrical Engineering and Computer Science, MIT,
  2003.

\bibitem{tung:78}
S.~Tung, ``Multiterminal source coding,'' Ph.D. dissertation, Cornell
  University, 1978.

\bibitem{etezadiKhistiDCC:12}
F.~Etezadi and A.~Khisti, ``Prospicient real-time coding of {M}arkov sources
  over burst erasure channels: Lossless case,'' in \emph{Proc.\ Data\
  Compression\ Conf.}, Apr. 2012.

\bibitem{rimoldi:94}
B.~Rimoldi, ``Successive refinement of information: Characterization of the
  achievable rates,'' \emph{IEEE Trans.\ Inform.\ Theory}, vol.~40, pp.
  253--259, Jan. 1994.

\bibitem{etezadi-khisti-14}
F.~Etezadi and A.~Khisti, ``Sequential transmission of {M}arkov sources over
  burst erasure channels,'' in \emph{International Zurich Seminar}, 2014.

\end{thebibliography}

\end{document}